\documentclass[]{sn-jnl}

\usepackage[table]{xcolor}
\usepackage{tcolorbox}
\usepackage[frozencache,cachedir=./_minted]{minted}
\def\us{\char`\_}
\def\coqin#1{\mintinline{ssr}{#1}}
\setminted[ssr]{fontsize=\small}
\usepackage{syntax}
\usepackage{tabularx}
\usepackage{stmaryrd}
\usepackage{simplebnf}
\usepackage{comment}
\usepackage{rotating}
\usepackage{hyperref}
\usepackage{tablefootnote}
\usepackage{booktabs}
\usepackage{cite}
\usepackage{amssymb}
\usepackage{amsmath}
\usepackage{subcaption}
\usepackage{proof}
\usepackage{interval,derivative}
\usepackage{float}
\usepackage{tikz}
\usepackage{ebproof}
\usetikzlibrary{arrows.meta, positioning}

\def\itvcc#1#2{\interval{#1}{#2}}
\def\itvoo#1#2{\interval[open]{#1}{#2}}
\def\itvco#1#2{\interval[open right]{#1}{#2}}
\def\itvoc#1#2{\interval[open left]{#1}{#2}}

\def\mylim#1#2#3#4{#1 \xrightarrow[ #3 \to #4 ]{} #2}

\def\listfont{}%

\newcommand{\Real}{{\mathbb R}}
\newcommand{\Pp}{{\mathcal P}}
\newcommand{\Sp}{{\cal S}}

\newcommand{\godel}{G\"{o}del}
\newcommand{\hseqOne}{\ensuremath{\mathcal{G}}}
\newcommand{\hseqTwo}{\ensuremath{\mathcal{H}}}
\newcommand{\seqOne}{\ensuremath{Q}}
\newcommand{\seqTwo}{\ensuremath{P}}
\newcommand{\seqThree}{\ensuremath{R}}
\newcommand{\seqFour}{\ensuremath{S}}
\newcommand{\forml}{\ensuremath{p}}

\newcommand{\mix}{\text{(MIX)}}
\newcommand{\spl}{\text{(SPLIT)}}
\newcommand{\ec}{\text{(EC)}}
\newcommand{\ew}{\text{(EW)}}
\newcommand{\eex}{\text{(EEX)}}
\newcommand{\lex}{\text{(LEX)}}
\newcommand{\rex}{\text{(REX)}}
\newcommand{\com}{\text{(COM)}}
\newcommand{\emp}{\text{(EMP)}}
\newcommand{\weak}{\text{(W)}}
\newcommand{\init}{\text{(init)}}
\newcommand{\contr}{\text{(C)}}

\newcommand{\losssymbol}{{\cal L}}
\newcommand{\lossfn}{\losssymbol}

\newcommand{\naryConj}{\bigwedge\nolimits_M}
\newcommand{\naryDisj}{\bigvee\nolimits_M}

\def\DL{\textrm{differentiable logic}}
\def\DLtwo{\textrm{DL2}}
\def\STL{\textrm{STL$_\nu$}}
\def\STLinfty{\textrm{STL$_\infty$}}
\def\product{\textrm{product}}
\def\Yager{\textrm{Yager$_r$}}
\def\Luka{\textrm{Łukasiewicz}}
\def\Godel{\textrm{Gödel}}

\newcommand{\tDLtwo}[1]{\llbracket #1 \rrbracket_{\DLtwo}}
\newcommand{\tGodel}[1]{\llbracket #1 \rrbracket_{\textrm{G}}}
\newcommand{\tSTL}[1]{\llbracket #1 \rrbracket_{\STL}}
\newcommand{\tSTLinfty}[1]{\llbracket #1 \rrbracket_{\STLinfty}}
\newcommand{\tlukasiewicz}[1]{\llbracket #1 \rrbracket_{\textrm{Ł}}}
\newcommand{\tyager}[1]{\llbracket #1 \rrbracket_{\textrm{Y}}}
\newcommand{\tproduct}[1]{\llbracket #1 \rrbracket_{\textrm{P}}}
\newcommand{\tdl}[1]{\llbracket #1 \rrbracket_{\textrm{DL}}}

\newcommand{\tempty}[1]{\llbracket #1 \rrbracket}

\newcommand{\tdlbig}[1]{\bigg\llbracket #1 \bigg\rrbracket_{\textrm{DL}}}

\newcommand{\vect}[1]{\mathbf{#1}}
\newcommand{\x}{\vect{x}} 
\newcommand{\y}{\vect{y}} 

\newcommand{\Nat}{\mathbb{N}}

\newcommand{\pval}{p}

\newcommand{\VecType}[1]{\ensuremath{\text{Vec } #1}}
\newcommand{\IndexType}[1]{\ensuremath{\text{Index } #1}}
\newcommand{\BoolType}{\ensuremath{\text{Bool}}}
\newcommand{\RealType}{\ensuremath{\text{Real}}}
\newcommand{\FunType}[2]{\ensuremath{\text{Fun } #1 \ #2}}
\newcommand{\FunTwoType}[3]{\ensuremath{\text{Fun}_2\ #1 \ #2 \ #3}}

\newcommand*{\impl}{\Rightarrow}
\newcommand{\semSpace}{\qquad\qquad}

\newcommand{\setBaseUrl}[1]{\def\baseurl#1}
\setBaseUrl{{https://github.com/ndslusarz/formal_LDL/blob/main/}}

\newtcolorbox{authorComment}[1]{colback=#1}

\usepackage[textwidth=2.5cm,textsize=footnotesize,colorinlistoftodos]{todonotes}

\def\rocq{\textsc{Rocq}}
\def\mathcomp{\textsc{MathComp}}
\def\analysis{\textsc{MathComp-Analysis}}

\theoremstyle{definition}\newtheorem{example}{Example}[section]

\theoremstyle{definition}\newtheorem{definition}{Definition}[section]

\theoremstyle{definition}\newtheorem{theorem}{Theorem}[section]

\theoremstyle{definition}

\theoremstyle{definition}\newtheorem{counterexample}{Counterexample}[section]

\begin{document}

\title[A Foundation for Differentiable Logics using Dependent Type Theory]{A 
Foundation for Differentiable Logics using Dependent Type Theory} 


\author*[1]{Reynald Affeldt}\email{reynald.affeldt@aist.go.jp}
\equalcont{Author Contribution: All authors contributed equally to this work.}

\author[2]{Alessandro Bruni}\email{brun@itu.dk}
\equalcont{Author Contribution: All authors contributed equally to this work.}

\author[3,4]{Ekaterina Komendantskaya}\email{e.komendantskaya@soton.ac.uk}
\equalcont{Author Contribution: All authors contributed equally to this work.}

\author[4]{Natalia Ślusarz}\email{nds1@hw.ac.uk}
\equalcont{Author Contribution: All authors contributed equally to this work.}

\author[4]{Kathrin Stark}\email{k.stark@hw.ac.uk}
\equalcont{Author Contribution: All authors contributed equally to this work.}

\affil[1]{National Institute of Advanced Industrial Science and  Technology 
(AIST), Japan}

\affil[2]{IT-University of Copenhagen, Denmark}

\affil[3]{Southampton University, UK}

\affil[4]{Heriot-Watt University, UK}



\abstract{
  Differentiable logics are a family of quantitative logics originated
  in the machine learning literature. Because of their origin,
  differentiable logics often come equipped with analytic properties
  that guarantee that they are differentiable. However, they usually
  lack an accompanying theory that describes their algebraic and
  proof-theoretic properties.
  Meanwhile, fuzzy logics, seen as substructural logics, have been
  studied algebraically and proof-theoretically, and some fuzzy logics
  with desirable analytic properties have also been used in machine
  learning.
  Our aim is to systematically compare analytic, algebraic and
  proof-theoretical properties of both fuzzy and differentiable
  logics.

  To this end, we formalize differentiable and fuzzy logics in a
  unified framework, encoded using the \mathcomp{} library in the
  \rocq{} proof assistant.
  We propose a single language specification to encompass multiple
  logics, using intrinsic typing to only allow valid and well-typed
  formulas for each of the logics that we encode: Yager, \Luka{},
  \Godel{} and \product{} fuzzy logics, as well as the differentiable
  logics \DLtwo{} and STL.
  Algebraically, we show how these logics can be interpreted using
  residuated lattices, which are prevalent in the theory
  of substructural logics.
  Analytically, we formalise the existence of a positive derivative
  for certain logical connectives, and to this end we
  formalise L'H\^opital's, contributing it to the \mathcomp{} library.
  Proof-theoretically, we formalise established sequent calculi for
  fuzzy logics, and we propose new sequent calculi for \DLtwo{} and
  \STLinfty{}, and formalise their soundness in our framework.
%

}

\maketitle

\def\newterm#1{{\sl #1}}
\def\mydef{\overset{\textrm{def}}{=}}

\section*{Differences with the conference paper}

Some of the results presented in this journal submission appeared in the 
paper~\cite{ldl-coq}, published at ITP'24.
This journal submission extends considerably the results presented at ITP'24.

We have revised the ITP'24 results and \rocq{} formalisation in two main ways.

Firstly, we have added a new characterisation of \DL{}s in terms of 
\emph{residuated lattices}, an algebraic structure that was traditionally used 
to give algebraic semantics to fuzzy logics~\cite{galatos2007residuated}. This 
characterisation was missing in the \DL{}/machine learning literature; our work 
brings semantic uniformity to the analysis of fuzzy \DL{}s and \DL{}s proposed 
in machine learning papers.

Secondly, while fuzzy logics received a thorough proof-theoretic treatment in 
the literature~\cite{fuzzy-proof, galatos2007residuated},
new, machine-learning inspired \DL{}s had no proof-theoretic semantics. Here, 
we 
formulate a new sound sequent calculus for \DLtwo{} and a 
sound sequent calculus for a restricted version of STL.


Specifically, the following sections and subsections are new:

\begin{itemize}
\item In Sect.~\ref{sec:diff-logic}, the syntax and semantics of 
  \DL{}s is reformulated to fit within a residuated lattice framework.
  In Sect.~\ref{sec:algebraic}, we provide the corresponding proofs.


\item In Sect.~\ref{sec:STLconjunction}, we proved hat the $n$-ary 
conjunction of \STL{} converges to $\min$
when $\nu$ 
tends towards $+\infty$ and is therefore 
a lattice operation.
    
\item In Sect.~\ref{sec:lhopital}, we have generalised L'H\^opital's
  rule. 
  
\item In Sect.~\ref{sec:soundness},
we define hypersequent calculi for \DL{}s and prove their soundness. 

\item In Sect.~\ref{sec:ex-formalised}, we give an 
application of \DL{}s in machine learning.

\end{itemize}

\section{Introduction}
\label{sec:introduction}

This work aims at contributing to the field of formal verification of
machine learning seen as the
study of algorithms that learn statistically from data. Neural networks
are the most common technical device used in machine learning.
The standard learning algorithms (such as gradient descent) use a
\newterm{loss function}
$\losssymbol: \Real^m \times \Real^n \rightarrow \Real$ to optimise
the neural network's parameters (say, $\theta$) to fit the input-output
vectors given by the data in a way that the loss $\losssymbol(\x, \y)$
is minimised. This optimisation objective is usually denoted as
$\min_\theta \lossfn(\x, \y)$. 

Most approaches to verification of neural networks consist of an
automated procedure based on SMT solving, abstract interpretation, or
branch-and-bound techniques (see, e.g., Albarghouthi's
survey~\cite{albarghouthi2021introduction}).
Verification typically applies after training because traditional
learning is purely data-driven and thus agnostic to verification
properties.
In contrast, \newterm{property-guided training\/} takes place once the
verification properties are stated.
More precisely, verification of neural networks consists of two parts:
statement and verification of a given property, and training
of the neural network that optimises the neural network's parameters
towards satisfying the given property.

However, naively or manually performed mapping of a logical property
to an optimisation task results in major discrepancies (as shown by
Casadio et al.~\cite{CasadioKDKKAR22}). This suggests the need to have 
principled
tools for property-guided training. One approach is to provide
programming language support for property-guided development of neural
networks that involves specification, verification, and optimisation in
a safe-by-construction environment.
As an example, Vehicle~\cite{FoMLAS2023,DaggittKAKS025} provides this 
support in the form of
a Haskell DSL, with
(1)~a higher-order typed specification
language, in which required neural network properties can be clearly
documented, and 
(2)~type-driven compilation, which can take care of
correct-by-construction translation of properties into both the
language of neural network verifiers it interfaces with and loss 
functions.

To generate loss functions from a logical property, one can use
\newterm{differentiable logics} (DLs).
One possibility for such \DL{}s are the
well-studied \newterm{fuzzy logics} that date back to the works of
Łukasiewicz and G\"{o}del~\cite{van2022analyzing}.
\begin{example}
\label{ex:small-dl-example}
The semantics of a \DL{} is a map $\tempty{\cdot}$ from a given syntax to a 
real number-valued function.
Consider, as an example of a fuzzy logic, \product{} logic~\cite[Sect. 
3]{van2022analyzing}.
Let $\phi$ be a small fragment of propositional syntax: 
$\phi := a \in A\ |\ \bot\ |\ \varphi_1 \land 
\varphi_2$, where $\varphi_1,\varphi_2 \in \phi$,
and $A$ is a set of real number-valued propositional variables. 
In \product{} logic, the semantics of this small syntax would be
$\tempty{a} := a$,
$\tempty{\bot} := 0$, and
$\tempty{\varphi_1 \land \varphi_2} := \tempty{\varphi_1} \times 
\tempty{\varphi_2}$.
\end{example}

As an alternative, both verification and machine-learning communities
formulated \DL{}s such as \DLtwo~\cite{fischer2019dl2} and
STL~\cite{varnai}.

Existing \DL{}s are very different from one another. For example, they do
not agree on the domains of the resulting loss functions: fuzzy logics
have as domain the closed interval\footnote{We use the notation $\itvoo{a}{b}$ for 
intervals open on both sides, $\itvcc{a}{b}$ for closed, and $\itvco{a}{b}$ and 
$\itvoc{a}{b}$ for closed only on the left/right side respectively.} $\itvcc{0}{1}$,
the original domain of \DLtwo{} (reversed in this work for the purpose of 
consistent treatment of greatest value as true) corresponds to the 
\newterm{Lawvere quantale}~\cite{Law73}
$\itvco{0}{+\infty}$, 
and
\STL{}'s domain is $\itvoo{-\infty}{+\infty}$ (all intervals begin
equipped with the usual ordering on reals).  Each domain has a
designated value or interval for truth 
 and falsity. 
 In 
consequence, it is unclear which \DL{} is more
suitable for property-guided training and how the \DL{}s compare.
Empirical evaluations, such as Flinkow et al.~\cite{flinkow2025comparing}, 
cannot
provide a clear principled answer to the problem of developing suitable
\DL{}s.

%
This situation calls for a principled unified mathematical study of \DL{}s.
One way to do this is to take into account a century-long tradition of 
algebraic and proof-theoretic studies of fuzzy 
logics~\cite{galatos2007residuated,fuzzy-proof}, and show how those 
well-established results extend to, and shed the light on, the new \DL{}s 
proposed in the machine learning context. On the other hand, machine learning 
 comes with is own approach to analysing mathematical (this time 
understood as \emph{analytical}) properties of \DL{}s. Thus, the second 
research question to ask is whether, and how, the century-old fuzzy logics 
(rendered as \DL{}s) fit into this machine learning context. 
This research program poses a challenge, due to both 
the large zoo of existing \DL{}s as well as the diversity of the mathematical 
methods used by 
different 
communities to motivate and analyse these logics.

\emph{Why reconciling algebraic and analytic traditions is hard 
for \DL{}s?}
On the one hand, even though \DL{}s derived from fuzzy logics can be 
characterised by
the existing algebraic semantics of fuzzy 
logics~\cite{galatos2007residuated,fuzzy-proof},
the majority of these logics do not enjoy desirable analytic properties,
such as \newterm{smoothness} or \newterm{shadow-lifting}~\cite{varnai}, which 
are key for machine learning applications.
On the other hand, the new \DL{s} inspired by 
machine learning~\cite{fischer2019dl2,varnai}
are not very well understood from the algebraic point of view.
For example, Varnai and Dimarogonas~\cite{varnai} prove that 
\DL{} operations cannot satisfy simultaneously the properties of 
shadow-lifting, idempotence, and associativity. Yet, associativity is 
traditionally the key property in the 
algebraic analysis of \DL{}s~\cite{galatos2007residuated}.

Similarly, proof-theoretic study of \DL{}s is lacking.
While some fuzzy logics enjoy sound and/or complete propositional sequent
calculi~\cite{fuzzy-proof}, the new \DL{}s suggested by machine
learning communities have not been studied from the proof-theoretic
perspective.


We believe that a systematic investigation of \DL{}s is a necessary
step towards the development of a reliable neural network verification
tool, and towards formal verification of machine learning in general.
We therefore propose to formalise them in an interactive theorem prover, 
providing a
unified formalisation of \DL{}s in which their properties can be
compared rigorously and in which one can verify the translation from
properties to loss functions.
For that purpose, we chose the \rocq{} prover:
firstly, because it has been used to carry out numerous formalisations of logics
and programming languages; 
secondly, because the Mathematical Components project provides libraries to
handle the algebraic properties of \DL{}s (using
\mathcomp{}~\cite{mathcomp}) and the analytic properties such as
shadow-lifting (using \analysis{}~\cite{analysis}).

Our approach to a systematic investigation of \DL{}s is to build on
top of previous work that has already proposed a common presentation
of \DL{}s~\cite{ldl} and to advance the understanding of mathematical
properties of \DL{}s along the following three axes:
\begin{itemize}
\item \newterm{Algebraic\/}: The idea is to define algebraic semantics
  for \DL{}s that guide our general understanding and comparison of
  their properties.
  For that purpose, we use \newterm{residuated
    lattices}.
  Residuated lattices are algebraic structures that have been
  suggested as a suitable model for quantitative
  \newterm{substructural 
    logics}~\cite{galatos2007residuated,fuzzy-proof}.
  We show which of the new \DL{}s can be mapped to residuated lattices and make 
  a comparative inventory
  of the structural properties of the logical connectives of \DL{}s.
  %
\item \newterm{Analytic\/}: This axis concerns the properties of
  \DL{}s that rely on the use of real analysis, because of the use of
  limits and differentiation. A central example is shadow-lifting
  already mentioned above. It states that the resulting functions
  should have partial derivatives that can characterise the idea of
  gradual improvement in training~\cite{varnai}. For example, the
  translation of a conjunction should evaluate to a higher value if
  the value of one of its conjuncts
  increases. Example~\ref{ex:small-dl-example} illustrates how such a
  conjunction may look.
  We formally prove or disprove (the latter on pen and paper) shadow-lifting, 
  for each chosen \DL{}.
\item \newterm{Proof-theoretic\/}:
  Proof theory, and in particular Gentzen sequent calculi, have long
  been used to study computational properties of logics.  While fuzzy
  logics already have well-defined proof-theoretic
  semantics~\cite{fuzzy-proof}, the question of whether new \DL{}s may
  enjoy such semantics is still open.  We present sound sequent
  calculi for \DLtwo~\cite{fischer2019dl2} and \STLinfty{}, and argue why 
  \STL{} is
  unlikely to have a Gentzen-style proof calculus (see 
  Sect.~\ref{sec:realvaluedsemantics} for definitions of \STLinfty{} and 
  \STL{}).
\end{itemize}


%


The approach explained above led us to the following contributions:
\begin{enumerate}
\item We explain how to encode known \DL{}s in a single generic syntax
  using \rocq{}, building on known techniques for logic embedding,
  such as taking advantage of dependent types to achieve intrinsic
  typing.  The formalisation is comprehensive and extensible for
  future use.
\item We demonstrate how to use the Mathematical Components libraries
  for our purpose. For example, we take advantage of algebraic
  tactics~\cite{algebratactics} to automate algebraic and soundness
  proofs, and we develop reusable lemmas, such as L'H\^opital's rule
  to establish shadow-lifting.
\item As a result of our formalisation, we are able to find and fix
  errors in the literature, and to produce original results
  missing parts of the
  shadow-lifting proofs, and soundness of \DLtwo{} and \STLinfty{} with respect 
  to their
  new sequent calculi.
\end{enumerate}

The paper proceeds as follows.
Section~\ref{sec:preliminary} gives a preliminary introduction to 
differentiable logics, including their original definitions and initial broad 
comparison.
Section~\ref{sec:residuatedlattices} provides mathematical background on 
residuated 
lattices and their axioms.
%
Section~\ref{sec:propertiesDLs} gives an overview of the main results of the 
paper.
Section~\ref{sec:syntax} explains how one can define \DL{}s in \rocq{} using a 
generic
encoding, including a translation function producing the semantics.
Section~\ref{sec:algebraic} focuses on the formalisation of algebraic
properties  

 for \DL{}s.
In Sect.~\ref{sec:realanalysis}, we demonstrate the formal
verification of properties of \DL{}s relying on real analysis.
In Sect.~\ref{sec:soundness}, we show the formalisation of soundness with 
respect to hypersequent calculi of fuzzy \DL{}s and define a new calculus for 
\DLtwo{}. 
Section~\ref{sec:example} explains the application of \DL{}s in property-guided 
training by means of an example, and shows the formalisation of an example 
property.
We discuss related work in Sect.~\ref{sec:relatedwork} and conclude in
Sect.~\ref{sec:conclusion}.

\section{Preliminary Introduction of Differentiable Logics}
\label{sec:preliminary}

In this section, we preliminarily introduce the differentiable logics used 
within 
this work 
just as they are outlined in existing literature, which we extend in subsequent 
sections.
We highlight the initial relative limitations of their definitions when 
compared to one another.

We begin with \textbf{fuzzy differentiable logics}~\cite{van2022analyzing}, which include \Godel{}, \Luka{}, \Yager{} and \product{} logics.

\begin{definition}[Fuzzy Differentiable Logics~\cite{van2022analyzing}]
\label{def:fuzzy-lit}
Let $\phi$ be a propositional syntax: 
$\phi := a \in A\ | \top\ |\ \bot\ |\ \neg \varphi\ |\ \varphi_1 \land 
\varphi_2 \ |\ \varphi_1 \lor \varphi_2 \ |\ \varphi_1 \odot 
\varphi_2 \ |\ \varphi_1 \oplus
\varphi_2\ |\ \varphi_1 \impl \varphi_2$, where $\varphi_1,\varphi_2 \in \phi$ and $A$ is a set of real number-valued propositional variables. 
The semantics of fuzzy logic, denoted $\tempty{\cdot}_{\text{DL}}$, where DL $\in$ \{\Godel, \Luka, \Yager, \product\}, are defined in Table~\ref{tab:fuzzy-logics}.

\end{definition}

\begin{table}[ht!]
\footnotesize

{\renewcommand{\arraystretch}{1.4}
\begin{tabularx}{\textwidth}{|p{4.7em}|X|X|X|p{2.3em}|p{2.3em}|}
\hline
 \rowcolor{lightgray}            & $ \tdl{\varphi_1 \land \varphi_2} $ & $ \tdl{\varphi_1 \lor \varphi_2} $ &  $ \tdl{\lnot p} $ & $\tdl{\top}$ & $\tdl{\bot}$\\
\hline
\text{\Godel}      & $ \min(\tGodel{\varphi_1},\tGodel{\varphi_2}) $ & $ \max(\tGodel{\varphi_1},\tGodel{\varphi_2}) $
 &$\begin{cases}
                 1 & \text{if } \tGodel{p}=0\\
                 0 & \text{otherwise}
               \end{cases}$  & 1 & 0\\
\hline
\text{\Luka}  & $ \min(\tlukasiewicz{\varphi_1},\tlukasiewicz{\varphi_2})$ & $ \max(\tlukasiewicz{\varphi_1},\tlukasiewicz{\varphi_2})$ & $ 1-\tlukasiewicz{p} $ & 1 & 0\\
\hline
\text{\Yager}        & $ \min(\tyager{\varphi_1},\tyager{\varphi_2})$ & $ \max(\tyager{x},\tyager{\varphi_2})$ & $1 - (1 - (1-\tyager{p})^r)^{1/r} $ & 1 & 0\\
\hline
\text{\product}      & $ \min(\tproduct{\varphi_1},\tproduct{\varphi_2})$ & $ \max(\tproduct{x},\tproduct{\varphi_2})$ & $\begin{cases}
                 1 & \text{if } \tproduct{p}=0\\
                 0 & \text{otherwise}
               \end{cases}$  & 1 & 0\\
\hline
\end{tabularx}
}

{\renewcommand{\arraystretch}{1.4}
\begin{tabularx}{\textwidth}{|p{4.7em}|X|X|p{15em}|}
\hline
 \rowcolor{lightgray}    &   $ \tdl{\varphi_1 \odot \varphi_2} $ & $ \tdl{\varphi_1 \oplus \varphi_2} $ & $ \tdl{\varphi_1 \impl \varphi_2} $  \\
\hline
\text{\Godel}    & $ \min(\tGodel{\varphi_1},\tGodel{\varphi_2}) $ & $ \max(\tGodel{\varphi_1},\tGodel{\varphi_2}) $ &
$\begin{array}{@{}l@{}}
\begin{cases}
1 & \text{if } \tGodel{\varphi_1} \leq \tGodel{\varphi_2}\\
\tGodel{\varphi_2} & \text{otherwise}
\end{cases}
\end{array}$ \\
\hline
\text{\Luka} & $ \max(\tlukasiewicz{\varphi_1} + \tlukasiewicz{\varphi_2} - 1, 0) $ 
  & $ \min(\tlukasiewicz{\varphi_1} + \tlukasiewicz{\varphi_2},1) $
  & $\min(1 - \tlukasiewicz{\varphi_1} + \tlukasiewicz{\varphi_2}, 1)$ \\
\hline
\text{Yager} & $ \max (1 - ((1-\tyager{\varphi_1})^r + (1-\tyager{\varphi_2})^r)^{1/r},0)$
& $\min((\tyager{\varphi_1}^r + \tyager{\varphi_2}^r)^{\frac{1}{r}},1)$
& $\begin{array}{@{}l@{}}
  \begin{cases}
  1 \text{ if } \tyager{\varphi_1} \leq \tyager{\varphi_2} &\\
  1 - ((1-\tyager{\varphi_1})^r - & \\
  \text{ \ \ \  }(1-\tyager{\varphi_2})^r)^{1/r} & \text{ o.w.}
  \end{cases}
\end{array}$ \\
\hline
\text{Product}   & $  \tproduct{\varphi_1} \times \tproduct{\varphi_2}  $   & $\tproduct{\varphi_1} + \tproduct{\varphi_2} - \tproduct{\varphi_1} \times \tproduct{\varphi_2}$ &
$\begin{array}{@{}l@{}}
\begin{cases}
1 & \text{if } \tproduct{\varphi_1} \leq \tproduct{\varphi_2} \\
\frac{\tproduct{\varphi_2}}{\tproduct{\varphi_1}} & \text{otherwise}
\end{cases}
\end{array}$ \\
\hline
\end{tabularx}
}
\caption{Semantics of connectives of fuzzy logics, where $\varphi_1, \varphi_2 
\in \phi$  and $r \in \Real^+$.
}
\label{tab:fuzzy-logics}
\end{table}

Fuzzy differentiable logics are the most well-researched in terms of their 
algebraic properties, due to the long history of logic research. 
In particular, while not presented this way by Krieken et 
al.~\cite{van2022analyzing}, who first suggest their application as 
differentiable logics, they are traditionally defined using residuated 
lattices~\cite{galatos2007residuated}---the definition followed in 
Definition~\ref{def:fuzzy-lit}---which we explain in more detail in 
Sect.~\ref{sec:residuatedlattices}. 
Additionally, they are the only differentiable logics to come with a proof 
theoretic interpretation.

The remaining \DL{}s have their origin in machine learning.

\textbf{DL2} is defined around the concept of incorporating comparisons between terms, often present in machine learning properties (e.g., robustness), as the core part of the language~\cite{fischer2019dl2}:

\begin{definition}[DL2~\cite{fischer2019dl2}]
\label{def:dl2-lit}
Let $\phi$ be a propositional syntax extended with comparisons between real values: 
$\phi := a \in A\ |\ a_1 \leq a_2\ |\ a_1 \neq a_2\ |\ \top\ |\ \varphi_1 \land 
\varphi_2 \ |\ \varphi_1 \lor 
\varphi_2$, where $\varphi_1,\varphi_2 \in \phi$,
and $A$ is a set of real number-valued propositional variables. 
The semantics of DL2, denoted $\tDLtwo{\cdot}$, is:
\begin{align*}
\tDLtwo{a} &= a\\
\tDLtwo{\top} &= 0\\
\tDLtwo{a_1 \leq a_2} & = \max (\tDLtwo{a_1} - \tDLtwo{a_2}, 0)\\
\tDLtwo{a_1 \neq a_2} & = \epsilon \times [\tDLtwo{(r_1)} = \tDLtwo{(r_2)}]\\
\tDLtwo{\varphi_1 \land p_2} & = \tDLtwo{\varphi_1} + \tDLtwo{\varphi_1} \\
\tDLtwo{p_1 \lor p_2} & = \tDLtwo{\varphi_1} \times \tDLtwo{\varphi_1}
\end{align*}
\noindent where $\epsilon \in \Real^+$ is a parameter and $[\cdot]$ is an indicator function.

\end{definition}

The remaining comparisons between terms can be defined in terms of the above (e.g., $r_1 < r_2 = (r_2 \leq r_1) \land (r_2 \neq r_1)$).
 DL2 does not have negation as part of the syntax---it is instead pushed to the 
 level of comparisons between real numbers, which serve as base expressions. 
 
 \DLtwo{} is the only \DL{} to explicitly include comparisons between real 
 numbers in its syntax.
 Comparisons are a core part of the verification properties (e.g., 
 robustness~\cite{CasadioKDKKAR22}) and their inclusion this allows \DLtwo{} 
 more control over their behaviour, simultaneously making their treatment more 
 well-defined.

Last of the \DL{}s present in this work, \textbf{STL}, takes its name from 
Signal Temporal Logic for which it was originally designed~\cite{varnai}.
It has the most algebraically complex semantics of the individual connectives among the DLs considered in this work, as it was defined specifically to fulfil a set of desirable properties, with the most important one being shadow-lifting (defined in Section~\ref{sec:axes-of-study}).

Its primary connective is $n$-ary, non-associative conjunction.

\begin{definition}[STL~\cite{varnai}]
\label{def:stl-lit}
Let $\phi$ be a propositional syntax extended with comparisons between real values: 
$\phi := a \in A\ |\ \naryConj{(\varphi_1,\ldots,\varphi_M)} \ |\ \neg \varphi$, where $\varphi_i \in \phi$,
and $A$ is a set of real number-valued propositional variables. 
The semantics of STL, denoted $\tDLtwo{\cdot}$, is:

$\begin{array}{rl}
\tSTL{a} &  = a\\
\tSTL{\neg \varphi} & = - \tSTL{\varphi}\\
\tSTL{\naryConj{(\varphi_1,\ldots,\varphi_M)}}  & = \begin{cases}
  \dfrac{\sum_i \varphi_{\min} e^{\widetilde{\varphi_i}} 
  e^{\nu 
  \widetilde{\varphi_i}}}{\sum_i e^{\nu 
  \widetilde{\varphi_i}}} & 
                      \text{if}\ 
                      \varphi_{\min} < 0 \\
  \dfrac{\sum_i \varphi_i e^{-\nu 
  \widetilde{\varphi_i}}}{\sum_i e^{-\nu 
  \widetilde{\varphi_i}}} & \text{if}\ 
                      \varphi_{\min} > 0 \\
  0 & \text{if}\ \varphi_{\min} = 0 \\
\end{cases} \\
\end{array}$

 where $\nu \in \Real^+$, $\varphi_{\min} = \min(\varphi_1, \ldots, \varphi_M)$, and $\widetilde{\varphi_i} = \dfrac{p_i - \varphi_{\min}}{\varphi_{\min}}$.		

\end{definition}

The domain of STL is $\itvoo{-\infty}{\infty}$, where all positive values 
denote truth, and all negative values falsity. Its disjunction can be obtained 
using 
negation and conjunction (i.e., $p_1 \lor p_2 = \neg ((\neg p_1) \land (\neg 
p_2))$).

This initial introduction allows us to point out the main differences in the 
treatment of differentiable logics as well as compare them in terms of the 
properties they satisfy in their original state.
We have already noted that fuzzy \DL{}s have better investigated properties, 
particularly algebraic, as they can be defined using the well-known structure 
of residuated lattices. In contrast \DLtwo{} introduces novel syntactic elements
tailored to the application in neural network training (see 
Sect.~\ref{sec:example} for a detailed explanation of the training process).

Additionally, from the perspective of their application, only \DLtwo{} and STL, 
the 
logics with their roots in machine learning are differentiable, as fuzzy 
differentiable logics include connectives interpreted as $\min$ and $\max$. 
This can make them more convenient in usage for training, where 
differentiability is a desired property.

\section{Mathematical background on residuated lattices}
\label{sec:residuatedlattices}

In this section, we recall the notion 
of residuated lattice that will be key
to organise our formalisation of differentiable logics. 

Residuated lattices have been proposed as a suitable algebraic
structure to characterise a range of \newterm{substructural 
logics}~\cite{galatos2007residuated}, i.e., logics that are stronger than 
intuitionistic logic, but weaker than classical logic.
These include linear logic, paraconsistent and relevant logics, and,
most importantly for our current study, real number-valued and fuzzy
logics~\cite{galatos2007residuated,fuzzy-proof}.  We devote this
section to definition of residuated lattices, and will use the
knowledge of these algebraic properties as a guiding principle in
later sections.

In the below definition we use real numbers as a carrier set, since
real numbers are of specific interest in this work. Often, the below
definition is given with an abstract carrier set.

\begin{definition}
\label{def:lattice}
An algebra $\mathbf{RL} = (\Real, \lor, \land, \odot, \impl, 1 )$ is called 
a 
\newterm{residuated lattice} if
\begin{itemize}
    \item $(\Real, \lor, \land)$ is a lattice (i.e., $\lor, \land$ are 
    commutative, associative, and mutually absorptive, and thus give rise to 
    a lattice order $\leq$),
    \item $(\Real, \odot, 1)$ is a monoid (i.e., $\odot$ is commutative and
    associative with unit element 1),
    \item $x \odot y \leq z$ iff $y \leq x \impl z$ for all $x,y,z \in 
    \Real$.
\end{itemize}

We say that a residuated lattice is \newterm{distributive} if 
$x \land (y \lor z) \leq (x \land y) \lor (x \land z)$ for all $x,y,z \in 
    \Real$.
\end{definition}

The above definition gives rise to the axioms R1--R10 of 
Table~\ref{table:axiom-schema-lattice}. 

To see why this algebra has played a pivotal role in substructural
logics, notice that the lattice structure $(\Real, \lor, \land)$ 
corresponds to the additive fragment of linear
logic~\cite{agliano2025algebraic}. The monoidal structure
$(\Real, \odot, 1)$ corresponds to the multiplicative fragment of linear logic; 
for
example, $\odot$ can be given by multiplication on real numbers. 
See Galatos et al.~\cite{galatos2007residuated} for a more detailed analysis of
linear logic interpretation in residuated lattices.

Finally, thanks to the ``residuation axiom'' $x \odot y \leq z$ iff
$y \leq x \impl z$ (R10 in Table~\ref{table:axiom-schema-lattice}), we can
interpret implication as an adjoint operator to $\odot$, which gives a
quantitative interpretation of implication in residuated lattices.

\begin{example}
\label{ex:product-resid-lattice}
As an illustration, consider again \product{} logic 
(Example~\ref{ex:small-dl-example}) extended to a 
residuated lattice:

$$\mathbf{P} =  \left( \itvcc{0}{1}, \max, \min, \times, \Rightarrow, 1 
 \right) \textrm{where} $$

 $$\begin{array}{@{}l@{}}
	x \Rightarrow y \mydef
 \begin{cases}
1 & \text{if } x \leq y\\
\frac{y}{x} & \text{otherwise}
\end{cases}.
\end{array}$$
Its carrier set is $\itvcc{0}{1}$ and is characteristic for fuzzy logics. The 
lattice connectives are $\max$ and $\min$; they are commonly used within this 
context 
for other fuzzy logics. 
It is easy to check that $\times$ and $\Rightarrow$ satisfy the residuation 
axiom
(R10).
\end{example}

When defining a logic we expect to be given a definition of $\bot$ and/or $\impl$.
We define negation directly and, if $\bot$ and $\impl$ are defined, check whether it satisfies property (N1) because we observed in practice that its absence complicates the design of sequent calculi (see Sect.~\ref{sec:weak-compl}).
In line with the
literature~\cite[Sect. 2]{agliano2025algebraic},
we also define the properties of involutive negation (N1--N4) in Table~\ref{table:axiom-schema-neg}.
Properties (M1--M3) of Table~\ref{table:axiom-schema-mdual} define what it means for an operator $\oplus$ to be a \newterm{monoidal dual}~\cite[Sect. 2]{agliano2025algebraic}.
\begin{table}[t]
Axiomatisation of distributive residuated lattices:
\begin{tabularx}{\textwidth}{rXr}
R1 & $x \land y = y \land x$ & (Commutativity of $\land$) \\
R2& $x \land (y \land z) = (x \land y) \land z$ & (Associativity 
 of $\land$) 
 \\
R3& $x \lor y = y \lor x$ &  (Commutativity of $\lor$)\\
R4& $x \lor (y \lor z) = (x \lor y) \lor z$ &  (Associativity of 
 $\lor$)\\
R5& $x \land (x \lor y) = x $ &  (Mutual absorption 1)\\
R6& $x \lor (x \land y) = x $ &  (Mutual absorption 2)\\
R7& $x \land (y \lor z) \leq (x \land y) \lor (x \land z)$&  
(Distributivity) \\
R8& $x \odot  (y \odot  z) = (x \odot  y) \odot z$ &  
 (Associativity of $\odot$)\\
R9& $x \odot 1 = x$ & (Unit element) \\
R10& $x \odot y \leq z$ iff $y \leq x \impl z$ & (Residuation) \\
\end{tabularx}
\caption{Standard axiomatisation of residuated lattices.}
\label{table:axiom-schema-lattice}
\end{table}
\begin{table}[t]
Axiomatisation of negation:
\begin{tabularx}{\textwidth}{rXr}
N1 & $\neg x = x \impl \bot$ &  (Negation as implication)\\
N2 & $\neg \neg x = x$ &  (Involution)\\
N3 & $\neg (x \land y) = \neg x \lor \neg y$ &  (De Morgan 1)\\
N4 & $\neg (x \lor y) = \neg x \land \neg y$ &  (De Morgan 2)\\
\end{tabularx}
\caption{Properties of negation.}
\label{table:axiom-schema-neg}
\end{table}
\begin{table}[t]
Axiomatisation of monoidal dual:
\begin{tabularx}{\textwidth}{rXr}
M1 & $x_0 \oplus  (x_1 \oplus  x_2) = (x_0 \oplus  x_1) \oplus x_2$ &  
 (Associativity of $\oplus$)\\
M2 & $\neg (x \odot y) = \neg x \oplus \neg y$ &  (De Morgan 3)\\
M3 & $\neg (x \oplus y) = \neg x \odot \neg y$ &  (De Morgan 4)\\
\end{tabularx}
\caption{Properties of monoidal dual.}
\label{table:axiom-schema-mdual}
\end{table}

\section{Overview of the approach and main results}
\label{sec:propertiesDLs}

We first recall the syntax and semantics of differentiable logics
using one generic syntax (generic enough to represent all \DL{}s as
well as concrete applications) and an informal, yet complete
explanation of their semantics (Sect.~\ref{sec:diff-logic}).
We split one of the \DL{}s, STL, into two versions: \STL{} and \STLinfty{}, 
based on the value of parameter $\nu$ used which is used in its semantics.
%
We then explain our approach which consists of three axes (algebraic,
analytic, proof-theoretic, see Sect.~\ref{sec:introduction}) for the
properties we wish to establish across all \DL{}s.
We then summarise the known properties of different \DL{}s as well as
an overview of our results which we formalise throughout this work
(Sect.~\ref{sec:summary-props}).

\subsection{A unified presentation of differentiable logics}
\label{sec:diff-logic}

Ślusarz et al.~\cite{ldl} suggest a common syntax for multiple \DL{}s,
calling it the \newterm{logic of differentiable logics}, and
subsequently obtain different \DL{}s via different interpretation
functions.  In the following, we summarise the syntactic and semantic
features of \DL{}s following this formulation.  Compared to
Ślusarz et al.~\cite{ldl} and the previous version of this
paper~\cite{ldl-coq}, we extend this syntax to cohere with the set of
connectives used in other studies of substructural 
logics~\cite{galatos2007residuated,fuzzy-proof}.

\begin{figure}
\begin{tabular}{rcl}
\text{type} & ::= & \BoolType{} $\mid$ \RealType{} $\mid$ \\
& & \IndexType{n} $\mid$ \VecType{n} $\mid$ \FunType{m}{n} $\mid$
    \FunTwoType{l}{m}{n} $\quad$ $(l,m,n\in\Nat)$ \\
\text{exprIndex} & ::= & $i \in \Nat$ \\ 
\text{exprFun} & ::= & $f \in \Real^m \to \Real^n$ \\
\text{exprFun2} & ::= & $f_2 \in \Real^l \to \Real^m \to \Real^n$ \\
\text{exprVec} & ::= & $v \in \Real^n \mid f \; v \mid f_2 \; v_0 \; v_1$ \\
\text{exprR} & ::= &  $r \in \Real \;\mid\; [v]_i$  \\
$p, q \in$ \text{exprB} & ::= &
  $\top$ $\mid$ $\bot$ $\mid$  $r_1 \leq r_2$ $\mid$ $r_1 = r_2$ \\ 
& & $\neg p$ $\mid$ $p_0 \odot p_1$  $\mid$ $p_0 \oplus p_1$ $\mid$ $p_0 \impl 
p_1$ $\mid$
  $p_0 \land p_1$ $\mid$ $p_0 \lor p_1$ $\mid$  \\
\end{tabular}
\caption{Types and expressions of \DL{}s.}
\label{fig:syntax-types-math}
\end{figure}

\subsubsection{Syntax of \DL{}s}

The syntax of \DL{}s consists of types and expressions
(Fig.~\ref{fig:syntax-types-math}). It includes a mixture of features present 
in the individual \DL{}s, as outlined in Sect.~\ref{sec:preliminary}.

Types are given by Boolean values or real numbers, indices,
vectors, and function types (\FunType{m}{n} for unary functions and
\FunTwoType{l}{m}{n} for binary functions). 
The addition of the binary 
functions, absent in the previous version of this work~\cite{ldl-coq}, was 
deemed necessary when considering real-life properties, which can often include 
functions with two arguments.
Real expressions are either real numbers (ranged over by $r$) or
the result of addressing a vector of real numbers.
Vectors (ranged over by $v$) can be addressed using indices (ranged
over by $i$) or can be obtained by applying unary or binary functions
(ranged over by $f$).
Formulas (or Boolean expressions, ranged over by $p$) are built from $\top$ and 
$\bot$, from application of the predicates $\leq$ and $=$ to real 
expressions, or by combining formulas using logical connectives ($\neg$,
$\odot$, $\oplus$, $\impl$, $\land$, and $\lor$).
Here, we combine the best features of the original syntax of both \DLtwo{} 
(Definition~\ref{def:dl2-lit}), with its comparisons and fuzzy \DL{}s 
(Definition~\ref{def:fuzzy-lit}), which follow residuated lattices.

Because \STL{} by Varnai and Dimarogonas~\cite{varnai} lacks
associativity for $\lor$ and $\land$ (Definition~\ref{def:stl-lit}), we define 
them as $n$-ary connectives
separately using the following notation:
$$ \naryConj (p_0, \ldots, p_M) = p_0 \land \ldots \land p_M,$$
$$\naryDisj (p_0, \ldots, p_M) = p_0 \lor \ldots \lor p_M.$$

We forgo the originally included
quantifiers, lambda and let expressions~\cite{ldl} to obtain a simpler core
language in which the three types of properties of interest
(algebraic, analytic, and proof-theoretic) can be studied.

\subsubsection{Semantics of \DL{}s}
\label{sec:realvaluedsemantics}

To define a \DL{}, one defines an interpretation function
$\tdl{\cdot}$ that, given an expression in the syntax of
\DL{}s, returns a function on real numbers.
We introduce all \DL{}s in a generic way and use the meta notation
$\tdl{\cdot}$ to refer to a range of interpretation
functions, with
$\DL \in \{\Godel, \Luka, \Yager, \product, \DLtwo, \STL,
\STLinfty\}$. Some logics---\Yager{} and \STL---are parametrised by a 
non-negative $r, \nu \in \Real$ which are logics with a different 
fixed $\nu$ each.
Note also the existence of \STLinfty{}---a particular case of \STL{} where 
$\nu$ is $\infty$; its relevance and behaviour is discussed in detail in
Sect.~\ref{sec:propertiesDLs}.

Table~\ref{tab:semantics-math} shows the interpretation of all
\DL{}s. 
First are the three \DL{}s based on well-known fuzzy logics: \Godel,
\Luka~\cite{lukasiewicz1920three}, and \product{}~\cite{van2022analyzing}. 
We also include a fuzzy logic based on \Yager{} \newterm{t-norms} (binary functions 
which are commutative, associative, monotonic and have an identity element), 
which we refer to as \Yager{} logic from now on). 
\Yager{} t-norms belong to a 
wider family of parametrised t-norms~\cite{klement2013triangular} the behaviour 
of which changes depending on the value of the parameter. 
We assume $r > 0$ 
which makes the \Yager{} logic isomorphic to \Luka{} (and equal to \Luka{} 
at $r = 1$) and which at $r$ tending to $\infty$ converges to \Godel{}'s 
connectives.

All fuzzy logics have the interpretation domain of $\itvcc{0}{1} \subset 
\Real$.  Other logics have
different domains: \DLtwo~\cite{fischer2019dl2} has the interpretation
domain $\itvoc{-\infty}{0} $, and \STL~\cite{varnai} the domain 
$\itvoo{-\infty}{+\infty}$.

Note the lack of semantics of $\top$ and $\bot$ for \STL{} and \DLtwo{}. Both 
DLs define the notion of truth and falsity using intervals $\itvoo{0}{\infty}$ 
and 
$\itvoo{-\infty}{0}$ for \STL{} respectively, and $\itvoo{-\infty}{0}$ for 
false for 
\DLtwo. 

Also note that while negation is undefined for 
\DLtwo{} as a structural transformation of the syntax of formulas, it is 
implemented at
the level of atomic comparison only, and negation on composite formulas is 
provided
as syntactic sugar~\cite[Sect.~3]{fischer2019dl2}. For example, consider
$\tempty{3 = 3}_{\DLtwo} = 0$. In $\DLtwo$, it is not
defined in term of  $\tempty{3 = 3}_{\DLtwo}=0$, but $\DLtwo$ provides an
interpretation for the symbol~$\neq$ separately (we do not include all 
comparisons within the syntax).

The binary predicates $\leq$ and $=$ are defined in a way that ensures that 
they are interpreted within 
the chosen real interval for the  given \DL. 
Note that the interpretation of $\leq$ does not necessarily cohere with the 
interpretation of $\impl$ as might be expected. 
This has been done to follow more closely the behaviour of the interpretation 
function for the comparisons for \DLtwo{} (with additional allowances made for 
the domain)---that is, emphasising more the standard way in which real numbers 
can be compared more than the behaviours of the individual logics.
Because of this, we have a single, general comparison interpretation applicable 
to all fuzzy logics, as it is connected to their shared domain rather than 
different interpretations of implication.
 
\newcommand{\topline}{\arrayrulecolor{black}\specialrule{0.1em}{\abovetopsep}{0pt}%
	\arrayrulecolor{lightgray}\specialrule{\belowrulesep}{0pt}{0pt}%
	\arrayrulecolor{black}
}

\newcommand{\midline}{\arrayrulecolor{lightgray}\specialrule{\aboverulesep}{-1pt}{0pt}%
	\arrayrulecolor{black}\specialrule{\lightrulewidth}{0pt}{0pt}%
	\arrayrulecolor{white}\specialrule{\belowrulesep}{0pt}{0pt}%
	\arrayrulecolor{black}
}
\begin{table}[ht!]
\footnotesize

\newcommand\TL{\rule{0pt}{3.3ex}}
\newcommand\BL{\rule[-2.1ex]{0pt}{0pt}}
\newcommand\TY{\rule{0pt}{4ex}}
\newcommand\BY{\rule[-2.8ex]{0pt}{0pt}}
\newcommand\BP{\rule[-1.5ex]{0pt}{0pt}}
\newcommand\BDLtwo{\rule[-1.5ex]{0pt}{0pt}}

{\renewcommand{\arraystretch}{1.4}
\begin{tabularx}{\textwidth}{|p{4.7em}|X|X|X|}
\hline
 \rowcolor{lightgray}            & $ \tdl{ p_0 \land p_1 } $ & $ \tdl{ 
 p_0 \lor p_1 } $ &  $ \tdl{ \lnot p }$ \\
\hline
\text{\godel}      & $ \min(\tGodel{p_0},\tGodel{p_1}) $ & $ 
\max(\tGodel{p_0},\tGodel{p_1}) 
$ &$\begin{cases}
                 1 & \text{if } \tGodel{p}=0\\
                 0 & \text{otherwise}
               \end{cases}$  \\
\hline
\text{\Luka}  & $ \min(\tlukasiewicz{p_0},\tlukasiewicz{p_1})$ & $ 
\max(\tlukasiewicz{p_0},\tlukasiewicz{p_1})$ & $ 1-\tlukasiewicz{p}$ \\
\hline
\text{\Yager}        & $ \min(\tyager{p_0},\tyager{p_1})$ & $ 
\max(\tyager{p_0},\tyager{p_1})$ &  $1 - (1 - (1-\tyager{p})^r)^{1/r} $\\
\hline
\text{\product}      & $ \min(\tproduct{p_0},\tproduct{p_1})$ & $ 
\max(\tproduct{p_0},\tproduct{p_1})$ & $\begin{cases}
                 1 & \text{if } \tproduct{p}=0\\
                 0 & \text{otherwise}
               \end{cases}$  \\
\hline
\text{\DLtwo}       & $ \min(\tDLtwo{p_0},\tDLtwo{p_1}) $ & $ 
\max(\tDLtwo{p_0},\tDLtwo{p_1}) $ & \multicolumn{1}{c|}{\text{N.A.}}  \\
\hline
\text{\STL} & \multicolumn{1}{c|}{\text{N.A.}}  & 
\multicolumn{1}{c|}{\text{N.A.}}  &
$ {-}\tSTL{p} $ \\
\hline
\text{\STLinfty}         & $ \min(\tSTLinfty{p_0},\tSTLinfty{p_1}) $ & $ 
\max(\tSTLinfty{p_0},\tSTLinfty{p_1}) $ 
      & $ {-}\tSTLinfty{p} $  \\
\hline
\end{tabularx}
}

{\renewcommand{\arraystretch}{1.4}
\begin{tabularx}{\textwidth}{|p{4.7em}|p{11.3em}|p{11.2em}|X|}
\hline
 \rowcolor{lightgray}    &   $ \tdl{ p_0 \odot p_1 } $ & $ \tdl{ p_0 
 \oplus p_1 } $ & $ \tdl{ p_0 
 \impl p_1} $  \\
\hline
\text{\godel}    & $ \min(\tGodel{p_0},\tGodel{p_1})  $ & $ 
\max(\tGodel{p_0},\tGodel{p_1})  $ &
$\begin{array}{@{}l@{}}
\begin{cases}
1 & \text{if } \tGodel{p_0} \leq \tGodel{p_1} \\
p_1 & \text{otherwise}
\end{cases}
\end{array}$ \\
\hline
\text{\Luka} & $ \max(\tlukasiewicz{p_0} + \tlukasiewicz{p_1} - 1, 0) $ 
  & $ \min(\tlukasiewicz{p_0} + \tlukasiewicz{p_1},1) $
  & $\min(1 - \tlukasiewicz{p_0} + \tlukasiewicz{p_1}, 1)$ \\
\hline
\text{\Yager} & $ \max (1 - ((1-\tyager{p_0})^r + (1-\tyager{p_1})^r)^{1/r},0)$
& $\min((\tyager{p_0}^r + \tyager{p_1}^r)^{\frac{1}{r}},1)$
& $\begin{array}{@{}l}
  \begin{cases}
  1 \ \ \ \ \  \text{   if } \tyager{p_0} \leq \tyager{p_1} \\
  1 - ((1-\tyager{p_1})^r\\
    \qquad - (1-\tyager{p_0})^r)^{1/r} \;\; \text{o.w.} 
  \end{cases}
\end{array}$ \\
\hline
\text{\product}   & $  \tproduct{p_0} \times \tproduct{p_1}  
$   & $\tproduct{p_0} + \tproduct{p_1} - \tproduct{p_0} \times \tproduct{p_1}$ &
$\begin{array}{@{}l@{}}
\begin{cases}
1 & \text{if } \tproduct{p_0} \leq \tproduct{p_1} \\
\frac{\tproduct{p_1}}{\tproduct{p_0}} & \text{otherwise}
\end{cases}
\end{array}$ \\
\hline
\text{\DLtwo}    & 
$\tDLtwo{p_0} + \tDLtwo{p_1} $    & $- \tDLtwo{p_0} \times \tDLtwo{p_1}$ &
 $ - \max(\tDLtwo{p_0}-\tDLtwo{p_1}, 0) $\\
\hline
\text{\STL}      & \multicolumn{1}{c|}{\text{N.A.}} & 
\multicolumn{1}{c|}{\text{N.A.}} &
\multicolumn{1}{c|}{\text{N.A.}} \\
\hline
\text{\STLinfty}   & $ \min(\tSTLinfty{p_0}, \tSTLinfty{p_1}) $  & $ 
\max(\tSTLinfty{p_0}, \tSTLinfty{p_1}) $ &
$\begin{array}{@{}l@{}}
\begin{cases}
\infty & \text{if } \tSTLinfty{p_0}\\
& \leq \tSTLinfty{p_1} \\
\tSTLinfty{p_1} & \text{otherwise} 
\end{cases}
\end{array}$\\
\hline
\end{tabularx}
}
{\renewcommand{\arraystretch}{1.5}
\begin{tabularx}{\textwidth}{|p{4.7em}|X|X|}
\hline
 \rowcolor{lightgray} & $\tdl{\naryConj (p_1,\ldots,p_M)}$ & $\tdl{\naryDisj 
 (p_1,\ldots,p_M)}$ \\
\hline
\text{\STL} & $ \tSTL{\naryConj{(p_1,\ldots,p_M)}} $ & $ 
\tSTL{\naryDisj{(p_1,\ldots,p_M)}}$\\
\hline
\end{tabularx}
}
\newcommand\BF{\rule[-4ex]{0pt}{0pt}}

{\renewcommand{\arraystretch}{1.3}
\begin{tabularx}{\textwidth}{|p{4.6em}|p{12.5em}|p{16em}|c|c|}
\hline
\rowcolor{lightgray} & $ \tdl{ r_1 = r_2 } $ &
$ \tdl{ r_1 \le r_2 } $ & $ \tdl{ \top } $ & $ \tdl{ \bot }$ \\
\hline
\text{fuzzy}        & 
$\begin{array}{l}
\text{if }{\tempty{r_1} = -\tempty{r_2}}\\
    \text{then }{1}\\
\text{else }{\max\left(1-\left|{\tempty{r_1} - 
\tempty{r_2} \over \tempty{r_1} + 
 \tempty{r_2}}\right|, 0\right)}
\end{array} $ 
&
$\begin{array}{l}
\text{if }{\tempty{r_1} = -\tempty{r_2}}\\
 \text{then }{1}\\
\text{else }
 {\max\left(1-\max\left({\tempty{r_1} - \tempty{r_2} 
 \over \tempty{r_1} + \tempty{r_2}}, 0\right), 0\right)}
\end{array} $ & $ 1 $ & $ 0 $\\
\hline
\text{\DLtwo}       & $ -|\tDLtwo{r_2} - \tDLtwo{r_1}| $ &
$ - \max(\tDLtwo{r_1}-\tDLtwo{r_2}, 0) $ & $ 0 $ & 
\text{N.A.} \\
\hline
\text{\STL} & $ -|\tSTL{r_2} - \tSTL{r_1}| $ &
$ \tSTL{r_2} - \tSTL{r_1} $ & N.A. & N.A. \\
\hline
\text{\STLinfty} & $ -|\tSTLinfty{r_2} - \tSTLinfty{r_1}| $ &
$ \tSTLinfty{r_2} - \tSTLinfty{r_1} $ & $\infty$ & $-\infty$ \\
\hline
\end{tabularx}
}

\begin{tabular}{lcc}
$
\tSTL{\naryConj{(p_1,\ldots,p_M)}}  =
\begin{cases}
  \dfrac{\sum_i p_{\min} e^{\widetilde{p_i}} 
  e^{\nu 
  \widetilde{p_i}}}{\sum_i e^{\nu 
  \widetilde{p_i}}} & 
                      \text{if}\ 
                      p_{\min} < 0 \\
  \dfrac{\sum_i p_i e^{-\nu 
  \widetilde{p_i}}}{\sum_i e^{-\nu 
  \widetilde{p_i}}} & \text{if}\ 
                      p_{\min} > 0 \\
  0 & \text{if}\ p_{\min} = 0 \\
\end{cases}
$ 
\quad where \quad
$
\begin{array}{rcl}
  \nu &\in& \Real^+ \\
  p_{\min} &=& \min(p_1, \ldots, p_M) \\
  \widetilde{p_i} &=& \dfrac{p_i - p_{\min}}{p_{\min}}		
\end{array}
$ \\[2em]
$
\tSTL{\naryDisj (p_1,\ldots,p_M)} =
\begin{cases}
	\dfrac{\sum_i p_{\max} e^{\tilde{p_i}} e^{\nu \tilde{p_i}}}{\sum_i 
	e^{\nu \tilde{p_i}}} & \text{if}\ p_{\max} > 0 \\
	\dfrac{\sum_i p_i e^{-\nu \tilde{p_i}}}{\sum_i e^{-\nu 
	\tilde{p_i}}} & \text{if}\ p_{\min} > 0 \\
	0 & \text{if}\ p_{\min} = 0 \\
\end{cases}
$
\quad where \quad
$
\begin{array}{rcl}
\nu &\in &\Real^+ \\
p_{\max} & = & \max (p_1, \ldots, p_M) \\
\widetilde{p_i} &=& \dfrac{p_{\max} - p_i}{a_{\max}} 
\end{array}$
\end{tabular}

\caption{Interpretation function $\tdl{\cdot}$, where DL $\in 
\{\Godel, \Luka, \Yager, \product, \DLtwo,\STL, \STLinfty\}$.\\
}
\label{tab:semantics-math}
\end{table}

\subsection{Three axes of study for \DL{}s}
\label{sec:axes-of-study}

We divide the properties proven in this work into three groups: algebraic, 
analytic, and proof-theoretic.

\paragraph{(I) Algebraic properties}

\subparagraph{Interpretation via residuated lattices}

\begin{definition}
\label{def:interprationvia}
We say that a \DL{} is interpretable in \emph{a distributive residuated 
lattice\/}
if its operations satisfy
axioms R1--10 of Table~\ref{table:axiom-schema-lattice}.

In addition, we say that 
\begin{itemize}
\item a \DL{} \emph{has negation defined in terms of implication} if the operations defined for it in 
Table~\ref{tab:semantics-math} satisfy axiom N1 (Table~\ref{table:axiom-schema-neg});
\item a \DL{} \emph{has involutive negation} if axioms axioms N2--N4 (Table~\ref{table:axiom-schema-neg});
\item a \DL{}  \emph{has monoidal dual} if axioms M1--M3 are satisfied (Table~\ref{table:axiom-schema-mdual}).
\end{itemize}
\end{definition}

All of the chosen
fuzzy logics are known to satisfy R1--10 and N1--N4, 
and some of them have a
monoidal dual.  Our extension of \DLtwo{} satisfies only R1--10. \STL{}
operations do not satisfy either lattice or monoidal
axioms. However, we will show that \STLinfty{} does satisfy R1--R10 and N2--N4, signifying 
that at the
limit, STL's 
operators do have desirable algebraic properties.
Alas, they lose their desirable analytic properties at the limit!
This shows that the parameter $\nu$ of \STL{} may in fact serve as a
switch between the discrete and continuous worlds.

\subparagraph{Idempotence} 
The idempotence of the lattice connectives is a direct consequence of the mutual absorption axioms (R5--R6) which are part of the axiomatisation of a residuated lattice (Table~\ref{table:axiom-schema-lattice}), hence we do not consider it separately. Take $\lor$ as an example:

$$ x \lor x \overset{\textrm{(R5)}}{=} x \lor (x \land (x \lor x)) \overset{\textrm{(R6)}}{=} x .$$

The idempotence property for monoidal connectives, i.e.,
$x \odot x = x$ was proposed as one of the desirable properties by the
authors of \STL{}~\cite{varnai}. Informally, the property is justified as it
disables situations when duplication of identical data swings the
training process. However, since Varnai and Dimarogonas~\cite{varnai} show that 
idempotence,
associativity, and shadow-lifting cannot be satisfied at the same
time, idempotence stands at odds with definition of residuated
lattice, that does require associativity of both monoid and lattice
operators. On the other hand \STLinfty{} operators are
idempotent and associative, but not shadow-lifting.

\paragraph{(II) Analytic properties}
Varnai and Dimarogonas~\cite{varnai} introduce three properties in this 
category: weak 
smoothness, scale-invariance, and shadow-lifting.
The latter is the most important as it accounts for gradual improvement in 
training.
We only consider shadow-lifting here as it is the most complex of those 
properties and leave the remaining properties to future work.

\begin{definition}[Shadow-lifting property~\cite{varnai}]
\label{def:shadow}
The \DL{} satisfies the \newterm{shadow-lifting} property if, for any
$ \tdl{\pval} \neq 0$,
\begin{equation*}
	\left. \dfrac{\partial \tdlbig{\naryConj(\pval_0, \ldots, \pval_i, \ldots, 
			\pval_M)}}{\partial 
		\tdl{\pval_i}}\right\rvert_{\pval_j = \pval \text{ where } i 
		\neq j} >0
\end{equation*}
holds for all $0 \leq i \leq M$,
where $ \partial $ denotes partial differentiation.
\end{definition}

\noindent Notice that classical conjunction does not satisfy the
property of shadow-lifting: no matter how ``true'' the value of $p_2$ is, if 
$p_1$ is
false, then $p_1 \land p_2$ will remain false. Likewise, all \DL{}s that use 
$\min$ or
$\max$ to define conjunction will fail shadow-lifting.

Shadow-lifting was originally defined for conjunction only, as \STL{} had no 
disjunction. 
In our formalisation, we could, in principle, extend shadow-lifting to any 
logical connective of interest.
However, we left this 
 extension for future work.   

For \DL{}s that do not naturally have $\naryConj$, we prove shadow-lifting for 
$\odot$.

\paragraph{(III) Proof-theoretic properties}
For proof-theoretic properties, we will establish whether a logic has
a sound sequent calculus. This will be discussed in
Sect.~\ref{sec:proof-theoretic}.

\subsection{Summary of properties of \DL{}s}
\label{sec:summary-props}

\begin{table}[t]
\centering
\footnotesize
\begin{tabularx}{\textwidth}{|p{0.11\textwidth}|X|X|X|}
\hline
& Residuated Lattice (R1--10) & Negation as Implication (N1) & Involutive Negation (N2--4)\\
\hline
\Godel & yes (Sect.~\ref{sec:logical-positive-results}) & yes (Sect.~\ref{sec:logical-positive-results}) &no 
(Sect.~\ref{sec:logical-neg-results}) \\
\hline
\Luka & yes (Sect.~\ref{sec:logical-positive-results}) & yes (Sect.~\ref{sec:logical-positive-results}) & yes 
(Sect.~\ref{sec:logical-positive-results})  \\
\hline
\Yager & yes (Sect.~\ref{sec:logical-positive-results}) & yes (Sect.~\ref{sec:logical-positive-results}) & yes 
(Sect.~\ref{sec:logical-positive-results})  \\
\hline
\product & yes (Sect.~\ref{sec:logical-positive-results}) & yes (Sect.~\ref{sec:logical-positive-results}) & no 
(Sect.~\ref{sec:logical-neg-results}) 
\\
\hline
\DLtwo & \cellcolor{orange} yes (Sect.~\ref{sec:logical-positive-results}) & no 
(Sect.~\ref{sec:logical-neg-results}) & no 
(Sect.~\ref{sec:logical-neg-results}) \\
\hline
\STL & \cellcolor{orange}no (Sect.~\ref{sec:algebraic})  & \cellcolor{orange} no 
(Sect.~\ref{sec:logical-neg-results}) & \cellcolor{orange}yes 
(Sect.~\ref{sec:logical-positive-results})  \\
\hline
\STLinfty & \cellcolor{orange}yes (Sect.~\ref{sec:logical-positive-results}, \ref{sec:STLconjunction})& \cellcolor{orange} no 
(Sect.~\ref{sec:logical-neg-results}) & \cellcolor{orange}yes 
(Sect.~\ref{sec:logical-positive-results})  \\
\hline
\end{tabularx}
\begin{tabularx}{\textwidth}{|p{0.11\textwidth}|X|X|X|X|}
\hline
& Monoidal Dual (M1--3) & Idempotence (of $\odot$)& Shadow-lifting (of $\naryConj$ or $\odot$)& 
Sound Gentzen Calculus \\
\hline
\Godel & yes 
(Sect.~\ref{sec:logical-positive-results}) & yes (Sect.~\ref{sec:logical-positive-results}) & no 
(Sect.~\ref{sec:shadow-lifting}) & yes 
(Sect.~\ref{sec:proof-theoretic}--\ref{sec:soundness-fuzzy})\\
\hline
\Luka & yes 
(Sect.~\ref{sec:logical-positive-results}) & no (Sect.~\ref{sec:logical-neg-results})& no 
(Sect.~\ref{sec:shadow-lifting}) & yes 
(Sect.~\ref{sec:proof-theoretic}--\ref{sec:soundness-fuzzy})\\
\hline
\Yager & yes 
(Sect.~\ref{sec:logical-positive-results}) & no (Sect.~\ref{sec:logical-neg-results})  & no 
(Sect.~\ref{sec:shadow-lifting}) & unknown \\
\hline
\product & yes 
(Sect.~\ref{sec:logical-positive-results}) & no (Sect.~\ref{sec:logical-neg-results})  & yes 
(Sect.~\ref{sec:sldltwoproduct}) & yes 
(Sect.~\ref{sec:proof-theoretic}--\ref{sec:soundness-fuzzy})\\
\hline
\DLtwo & \cellcolor{orange} no 
(Sect.~\ref{sec:logical-neg-results}) & \cellcolor{orange} no (Sect.~\ref{sec:logical-neg-results}) & 
\cellcolor{orange} yes (Sect.~\ref{sec:sldltwoproduct}) & 
\cellcolor{orange} yes (Sect.~\ref{sec:sound-dl2})\\
\hline
\STL & 
\cellcolor{orange} no 
(Sect.~\ref{sec:logical-neg-results}) & yes (Sect.~\ref{sec:logical-positive-results})  & \cellcolor{yellow} yes 
(Sect.~\ref{sec:slstl})& 
unknown \\
\hline
\STLinfty & 
\cellcolor{orange}yes (Sect.~\ref{sec:logical-positive-results}) & \cellcolor{orange} yes (Sect.~\ref{sec:logical-positive-results}) & 
no (Sect.~\ref{sec:shadow-lifting}) & 
\cellcolor{orange} yes (Sect.~\ref{sec:sound-stl})\\
\hline
\end{tabularx}
\caption{
Properties of the different \DL{}s formalised in this paper 
\cite{github}.
We distinguish previously known proofs that we mechanise from
known results published with incomplete or
semi-formal proofs (\colorbox{yellow}{yellow}) and new results 
(\colorbox{orange}{orange}).
}
\setlength{\belowcaptionskip}{-10pt}
\label{tab:properties}
\vspace*{-1em}
\end{table}

Table~\ref{tab:properties} summarises all properties covered in our \rocq{} 
formalisation and highlights the ones for which we provide 
original proofs.

No \DL{} satisfies all desirable properties, given by columns of 
Table~\ref{tab:properties}. For example, \Godel{} 
interpretation 
satisfies axioms of the residuated lattice, has a sound sequent calculus, and 
in addition its 
connectives are idempotent. However, it is not
shadow-lifting. On the other hand, \Luka{} is not adequate relative to the 
Boolean Logic, \STL{} is not
associative, and while \DLtwo{} is adequate relative to the Boolean Logic and 
shadow-lifting, it fails
idempotence, and its negation is not compositional because it is not structural 
(in the sense of Fig.~\ref{tab:semantics-math}).
Varnai and Dimarogonas~\cite{varnai} have proven that it is impossible for any 
\DL{} to 
be idempotent, associative, and shadow-lifting at the same time.

When one has to make a choice of a \DL{}, different pragmatic considerations may
influence that choice.  If 
shadow-lifting is
strictly desirable, then \Godel{}, \Luka{}, and \Yager{} logics are
probably less desirable than the rest.  However, further choice between logical 
properties is less 
clear.  For example, one can
imagine a scenario when the specification language avoids negation,
and in a style of substructural logics, treats differently $p$ and
$p \odot p$ and thus sacrifices idempotence; in this case, \DLtwo{}
may provide an ideal translation function.

Conceptually, if one wishes to reconcile the three traditions: 
the algebraic theory of substructural logics via residuated lattices, 
the analytic properties needed for machine learning, and the
proof-theoretic semantics via the Gentzen sequent calculus, 
then one is looking, as a minimal requirement, to satisfy the properties listed 
in the first and the final two columns of Table~\ref{tab:properties}.
In that case, two logics, 
\product{} and \DLtwo{}, bear promise.



\emph{Parametrised logics.} An alternative approach to reconciling the 
properties of Table~\ref{tab:properties} is to make a greater use of 
logics whose connectives are parametrised by a real 
number as in \STL{} and \Yager{}. 

For example, taken together, \STL{} and \STLinfty{}, 
cover all properties in Table~\ref{tab:properties}: 
if $\nu \in \itvco{1}{\infty}$, \STL{} satisfies shadow-lifting; 
for $\nu \to \infty$, \STL{} convergences to \STLinfty{}, and hence satisfies the algebraic theory of residuated lattices.  
Similarly, for $r \to \infty$, \Yager{} converges to \Godel{}, with  \Yager{} satisfying involutive negation and \Godel{} satisfying idempotence.

\paragraph{Discrepancies in pen and paper proofs}
\label{sec:errors}
Note that the formalisation further revealed some errors as explained below.


While doing the \rocq{} formalisation, we found two sources of
errors in our pen and paper proofs \cite{ldl} as well as in the
original \STL{} paper~\cite{varnai}.  Firstly, the work in
\cite{ldl} concerned completing results of
Table~\ref{tab:properties} in the uniform notation of \DL{}s 
(Fig.~\ref{fig:syntax-types-math}). 

Some errors came from extension of fuzzy logics with 
comparison 
operators. The interpretation of comparisons needed to be scaled between $0$ 
and $1$, and it seemed obvious that the scaling was done correctly. Therefore 
we did not provide any proofs concerning the interval properties of these 
operations in previous work, leading to errors found and fixed in the course of 
the formalisation. 



While a sketch of the shadow-lifting proof was provided in
the original \STL{} paper~\cite{varnai}, it was incomplete and did not
mention either the existence of other cases as well as the reliance 
on L'H\^opital's rule.

\section{A generic encoding of \DL{}s in \rocq{} and its instantiations}
\label{sec:syntax}

In this section we address the aim of creating a generic and
extendable encoding for different \DL{}s, using uniform syntactic
conventions. We highlight the features of the encoding and the
benefits of the general framework in proving properties of
similar~\DL{}s.

\subsection{A generic encoding of the syntax of types and expressions}
\label{sec:genericencoding}

The generic encoding of the types is the matter of declaring
an inductive type corresponding to the syntax of
Fig.~\ref{fig:syntax-types-math}. However, as different \DL{}s are defined over 
different fragments of syntax, we need to account for this in the 
formalisation. We do so by refining the Boolean 
type with flags that allow control over which fragments of the syntax 
have defined semantics: 
\begin{minted}{ssr}
Inductive flag_neg := neg_def | neg_undef.
Inductive flag_impl := impl_def | impl_undef. 
Inductive flag_monoid := m_def | m_undef. 
Inductive flag_lattice := l_def | l_undef.
\end{minted}
The individual flags control, respectively, whether negation $\lnot$,
implication $\Rightarrow$, monoidal operations $\odot$ and $\oplus$, and lattice operations
$\wedge$ and $\vee$ are defined. This allows the inclusion of \DL{}s for
which not all connectives have defined semantics, such as \DLtwo.
Using the above flags we encode the syntax of the types of \DL{}s 
(Fig.~\ref{fig:syntax-types-math}) as follows: 
\begin{minted}{ssr}
Inductive dl_type :=
  | boolT of flag_neg & flag_impl & flag_monoid & flag_lattice
  | indexT of nat | realT | vectorT of nat
  | funT of nat & nat | fun2T of nat & nat & nat.
\end{minted}
Following the semantics from Table~\ref{tab:semantics-math} we define the types 
of the various \DL{}s by selecting the right sets of flags, with all fuzzy 
\DL{}s 
grouped together, for ease of reading:
\begin{minted}{ssr}
Definition boolT_fuzzy := boolT neg_def impl_def m_def l_def.
Definition boolT_dl2 := boolT neg_undef impl_def m_def l_undef.
Definition boolT_stl := boolT neg_def impl_undef m_undef l_def.
\end{minted}

As for \DL{}s' expressions, their encoding is displayed in
Fig.~\ref{fig:syntax-coq}.  It is an inductive type indexed by
\coqin{dl_type} so that the resulting syntax is intrinsically typed:
one cannot write ill-typed expressions.
The \rocq{} inductive type matches the syntax already
explained in Fig.~\ref{fig:syntax-types-math}.
Real expressions (line \ref{line:ldlreal}) use a type \coqin{R} of
type \coqin{realType} coming from \analysis{}~\cite{analysis} that represents 
real
numbers.
Boolean expressions (line \ref{line:ldlbool}) use the native \rocq{}
type \coqin{bool}.
Indices embed an \newterm{ordinal\/} from \mathcomp\
(line~\ref{line:ldlidx}).  More specifically, \coqin{'I_n} is the type
of natural numbers smaller than~\coqin{n}.
Last, vectors embed finite functions of type \coqin{'I_n -> R}, which in
\mathcomp{} have a special notation \coqin{R ^ n}
(line~\ref{line:ldlvec}).
We name the lattice connectives $\wedge$ and $\vee$ as \coqin{dl_and}
and \coqin{dl_or}, and 
include both the monoidal $\odot$ as
\coqin{dl_mand} and $\oplus$ as \coqin{dl_mor}.
We encode those connectives as functions from the finite domain
\coqin{'I_n} into expressions to allow the inclusion of
STL\footnote{Note: the use of finite functions for encoding vectors
  and $n$-ary operators is a departure from the previous version of
  the framework published in~\cite{ldl-coq}.  We previously encoded
  vectors as $n$-tuples and $n$-ary operators as polymorphic lists (of
  type \coqin{seq}).  Our new encoding does not fundamentally change
  the expressiveness of the framework, however it greatly simplifies
  the custom induction principle, and removes the need for
  list-membership checks.  This change removed $\approx 650$ lines of
  code from the development.}.

To ease reading, we will use notation such as \coqin{a `/\ b} or
\coqin{a `** b} to denote binary connectives (lattice and monoidal
conjunction respectively) from now on.
For a generic definition of the syntax, we need to allow for the case
of \DL{}s in which negation is not defined (here: \DLtwo).
The explicit flags, such as \coqin{neg_def} for \coqin{dl_not} (line 
\ref{line:ldlnot}), signify that this flag is necessary for this connective to 
be defined.
%
The constructor \coqin{dl_cmp} (line \ref{line:ldlcmp}) is for binary
comparison operators over the real numbers; hereafter, we will use
notations such as \coqin{`<=} for the comparison corresponding to
$\leq$ instead of ``\coqin{dl_cmp cmp_le}'' to ease reading.
As for the last constructors, they are respectively for functions,
their application, both for one- and two-argument versions (\coqin{dl_fun}, 
  \coqin{dl_fun2} respectively for functions, with the same naming convention 
  used for their application) as per Fig.~\ref{fig:syntax-types-math}.

\begin{figure}[h]
\begin{minted}[numbers=left,xleftmargin=1.5em,escapeinside=88]{ssr}
Inductive comparison := cmp_le | cmp_eq.

Inductive expr : dl_type -> Type :=
(* base expressions *)
| dl_bool : forall p r s t, bool -> expr (boolT p r s t) 8\label{line:ldlbool}8
| dl_idx : forall n, 'I_n -> expr (indexT n) 8\label{line:ldlidx}8
| dl_real : R -> expr realT 8\label{line:ldlreal}8
| dl_vec : forall n, R ^ n -> expr (vectorT n) 8\label{line:ldlvec}8
(* connectives *)
| dl_and : forall fn fi fm n, ('I_n -> expr (boolT fn fi fm l_def)) 
                               -> expr (boolT fn fi fm l_def)
| dl_or : forall fn fi fm n, ('I_n -> expr (boolT fn fi fm l_def)) 
                               -> expr (boolT fn fi fm l_def)
| dl_not : forall fi fm fl, expr (boolT neg_def fi fm fl)  8\label{line:ldlnot}8
                               -> expr (boolT neg_def  fi fm fl) 
| dl_impl :forall fn fm fl, expr (boolT fn impl_def fm fl) 
                               -> expr (boolT fn impl_def fm fl) 
                               -> expr (boolT fn impl_def fm fl)
| dl_mand : forall fn fi fl n, ('I_n -> expr (boolT fn fi m_def fl)) 
                               -> expr (boolT fn fi m_def fl)
| dl_mor : forall fn fi fl n, ('I_n -> expr (boolT fn fi m_def fl)) 
                               -> expr (boolT fn fi m_def fl) 
(* comparisons *)
| dl_cmp : forall fn fi fm fl, comparison -> expr realT -> expr realT 8\label{line:ldlcmp}8
                                  -> expr (boolT fn fi fm fl)
(* networks and applications *)
| dl_fun : forall n m, (R ^ n -> R ^ m) -> expr (funT n m)
| dl_fun2 : forall n m l, (R ^ n -> R ^ m -> R ^ l) -> expr (fun2T n m l)
| dl_app : forall n m, expr (funT n m) -> expr (vectorT n)
                          -> expr (vectorT m)
| dl_app2 : forall n m l, expr (fun2T n m l) -> expr (vectorT n)
                             -> expr (vectorT m) -> expr (vectorT l)
| dl_lookup : forall n, expr (vectorT n) -> expr (indexT n)
                           -> expr realT.
\end{minted}
\caption{Generic syntax for \DL{}s in \rocq{}}
\label{fig:syntax-coq}
\end{figure}

\subsection{Specific interpretation of the generic syntax}

We now proceed to the translation function that
interprets the syntax.
Types are mapped to their obvious semantics:
\begin{minted}[numbers=left,xleftmargin=1.5em,escapeinside=88]{ssr}
Definition type_translation (t : dl_type) : Type:=
  match t with
  | boolT x y z v  => R
  | realT => R
  | vectorT n => R ^ n
  | indexT n => 'I_n
  | funT n m => R ^ n -> R ^ m
  | fun2T n m l => R ^ n -> R ^ m -> R ^ l
end.
\end{minted}

In particular, Boolean values are mapped to \coqin{R} of type \coqin{realType}, the 
type of real numbers.
This translation accommodates the many interpretations of the DLs: the domain
$\itvcc{0}{1}$ for fuzzy logic, $\itvoc{-\infty}{0}$ for \DLtwo{}, and 
$\itvoo{-\infty}{+\infty}$ for \STL{}.
For \DLtwo{} and \STL{}, we also provide an alternative translation
function \coqin{ereal_type_translation} that maps Boolean values to
real numbers extended with $-\infty$ and $+\infty$, i.e., the type
\coqin{\bar R} of extended real numbers as provided by \analysis.
%

Each logic requires a separate interpretation function (as explained
in Table~\ref{tab:semantics-math}).  Here, we only show an excerpt of
the translation function for \STL{}, namely the cases for conjunction
(constructor \coqin{dl_and}), negation (notation \coqin{`~}), and comparison
(notation \coqin{`<=}) of \STL{} in Fig.~\ref{fig:semantics-coq}. 
In case of logics parametrised by $r$ and $\nu$ (\Yager, \STL) a corresponding 
variable 
is set in context (see an example in Fig.~\ref{fig:semantics-lt0-gt0}).

\begin{figure}[h]
\begin{minted}[xleftmargin=1.5em,escapeinside=88]{ssr}
Fixpoint stl_translation {t} (e : expr t) : type_translation t :=
  match p in expr t return type_translation t with
  | dl_and _ _ _ 0 _ => 1
  | dl_and _ _ _ n.+1 s =>
      let A := stl_translation \o s in
      let a_min : R := \big[minr/A ord0]_(i < n.+1) A i in
      stl_and a_min A                                           
  | `~ E1               => - {[ E1 ]}
  | E1 `<= E2           => {[ E2 ]} - {[ E1 ]}
  ... (* see 8\cite[dl.v]{github}8 for omitted connectives *)
  end where "{[ e ]}" := (stl_translation e).
\end{minted}
\setlength{\belowcaptionskip}{-5pt}
\caption{Excerpt of the semantics of \STL{}: conjunction, negation, and 
comparison. \\
{\footnotesize (See Fig.~\ref{fig:semantics-lt0-gt0} for intermediate 
definitions \coqin{stl_and_lt0} and \coqin{stl_and_gt0}.)}}
\label{fig:semantics-coq}
\end{figure}

The case for conjunction of \STL{} is the most complex in our
formalisation because dealing formally with it requires the theories
of exponentiation (\coqin{expR}), big sums (\coqin{\sum}), inverses
(\coqin{^-1}), and minima (\coqin{minr}).
To reduce the clutter, we define the cases for $p_{\min} > 0$ and $p_{\min} < 
0$ separately
as \coqin{stl_and_gt0} and \coqin{stl_and_lt0} reproduced in 
Fig.~\ref{fig:semantics-lt0-gt0}.
This will allow us to state intermediate lemmas about subexpressions.

\begin{figure}[h]
\begin{subfigure}[]{0.66\textwidth}
\begin{minted}[xleftmargin=1em,escapeinside=88]{ssr}
Variable (nu : R). 
Definition min_dev n i (f : 'I_n.+1 -> R) : R :=
  let r := \big[minr/f ord0]_(j < n.+1) f j 
  	in (f i - r) / r.
		
Definition stl_and_gt0 n (v : 'I_n.+1 -> R) :=
  (\sum_(a < n.+1) v a * 
  	expR (- nu * min_dev a v)) /
    \sum_(a < n.+1) expR (-nu * min_dev a v).

Definition stl_and_lt0 n (v : 'I_n.+1 -> R) :=
  (\sum_(a < n.+1)
    (\big[minr/v ord0]_(i < n.+1) v i) * 
    expR (min_dev a v) * expR (nu * min_dev a v)) /
      \sum_(a < n.+1) expR (nu * min_dev a v).
\end{minted}
\end{subfigure}
\;
\vrule
\;
\begin{subfigure}[]{0.3\textwidth}
\footnotesize
$
\begin{array}{l}
p_{\min} = \min(p_1, \ldots, p_M) \\
\widetilde{p_i} = \dfrac{p_i - p_{\min}}{p_{\min}}\\
\\
\nu \in \Real^+
\quad\quad\quad\;\;\;\text{ (constant)} \\
\dfrac{\sum_i p_i e^{-\nu \widetilde{p_i}}}{\sum_i e^{-\nu \widetilde{p_i}}}
\quad\quad\text{(case }p_{\min} > 0\text{)} \\
\\
\\
\dfrac{\sum_i p_{\min} e^{\widetilde{p_i}} 
	e^{\nu 
		\widetilde{p_i}}}{\sum_i e^{\nu 
		\widetilde{p_i}}}  
\;\text{(case }p_{\min} < 0\text{)} \\
\\
\\
\\
\\
\end{array}
$
\end{subfigure}
\setlength{\belowcaptionskip}{5pt}
\caption{Intermediate definitions to define the conjunction of \STL{}.\\
{\footnotesize (The right subfigure reproduces part of 
Table~\ref{tab:semantics-math} for reading convenience.)}}
\label{fig:semantics-lt0-gt0}
\end{figure}

\section{Algebraic properties of \DL{}s}
\label{sec:algebraic}

This section focuses on the formalisation of algebraic properties of
\DL{}s. We first present these results in mathematical notation,
followed by their \rocq{} formalisation.  We finish with a
formalisation of the negative results from Table~\ref{tab:properties}.

%
%

\subsection{Differentiable logics as residuated lattices}
\label{sec:logical-positive-results}

The following theorems characterise differentiable logics in terms of
residuated lattices.  Note that many of the algebraic proofs rely on
$\Real$ being a total order, that is, for any $x,y \in \Real$, either
$x\leq y$ or $y \leq x$.
\begin{theorem}[Differentiable logics as residuated lattices]
\label{thm:dl-residuated}
The following hold:
\begin{enumerate}
\item \Godel, \Luka, \Yager, \product, \DLtwo, and \STLinfty{} each form a residuated lattice,
\item \Luka{}, \Yager{} and \STLinfty{} are equipped with involutive negation,
\item \Godel, \Luka, \Yager, \product, and \STLinfty{} have a monoidal dual, and
\item \Godel{} and \STLinfty{} have idempotent $\odot$
\end{enumerate}
as described in Table~\ref{tab:properties}.
\end{theorem}
\begin{proof}\
For each logic, we prove the relevant axioms from 
Tables~\ref{table:axiom-schema-lattice}--\ref{table:axiom-schema-mdual}.
As an example, we present a proof for \DLtwo{}, formalised in~\cite[\coqin{dl2.v}]{github}.

Most cases follow directly by the corresponding properties of $\min$ and $\max$.
We will only show the two most interesting cases:




\textbf{(R7)}.  To prove $\tDLtwo{p_0 \land 
(p_1 \lor p_2)} \leq \tDLtwo{(p_0 \land p_1) \lor (p_0 \land p_2)}$ we actually
have to show that:
 $$\min(\tDLtwo{p_0}, \max (\tDLtwo{p_1}, \tDLtwo{p_2})) \leq 
\max(\min(\tDLtwo{p_0},\tDLtwo{p_1}), \min(\tDLtwo{p_0},\tDLtwo{p_2})).$$

For this, we consider all possible orders 
of $\tDLtwo{p_0}$, $\tDLtwo{p_1}$, and  $\tDLtwo{p_2}$. 
For example, if 
$\tDLtwo{p_0} \leq \tDLtwo{p_1} \leq \tDLtwo{p_2}$:
\begin{align*}
\min(\tDLtwo{p_0}, \max (\tDLtwo{p_1}, \tDLtwo{p_2})) & = \min(\tDLtwo{p_0}, \tDLtwo{p_2}) \\ 
& = \tDLtwo{p_0}  \\ 
& \leq \tDLtwo{p_0} \\ 
& = \max(\tDLtwo{p_0}, \tDLtwo{p_0}) \\ 
& = \max(\min(\tDLtwo{p_0},\tDLtwo{p_1}), \min(\tDLtwo{p_0},\tDLtwo{p_2})) \\ 
\end{align*}



\textbf{(R10)} To prove that $\tDLtwo{p_0 \odot p_1} \leq \tDLtwo{p_2}$ and 
$\tDLtwo{p_1} \leq \tDLtwo{p_0 \impl p_2}$ are equivalent,
we prove both directions separately.
Here, we only show the case for left to right implication as the other 
case is analogous.
We have:

\begin{align*}
\tDLtwo{p_0 \odot p_1} & \leq \tDLtwo{p_2} \\
\tDLtwo{p_0} + \tDLtwo{p_1} & \leq \tDLtwo{p_2}. \\
\end{align*}

We then need to prove 
$\tDLtwo{p_1} \leq -\max(\tDLtwo{p_0} - \tDLtwo{p_2}, 0)$. 

If
$\tDLtwo{p_0} - \tDLtwo{p_2} = 0$, this simplifies to $\tDLtwo{p_1} \leq 
0$ which is follows directly from the previous assumption as we know that $\tDLtwo{p_0} = 
\tDLtwo{p_2}$.

Otherwise, if $\tDLtwo{p_0} - \tDLtwo{p_2} > 0$,  we have to show that 
$\tDLtwo{p_1} \leq \tDLtwo{p_2} - \tDLtwo{p_0}$ which is equivalent to 
our assumption that $\tDLtwo{p_0} + \tDLtwo{p_1} \leq \tDLtwo{p_2}$.



\end{proof}

\subsection{Formalisation and proof in \rocq}

Proving these logical properties in \rocq{} consists of showing that they hold 
for 
the semantic
interpretation.
%
For example, the conjunction of \DLtwo{} being interpreted as addition
on real numbers inherits its associativity from the ring structure of
real numbers. As a consequence, its proof is a one-liner:
\begin{minted}{ssr}
Lemma dl2_mandA f1 f2 (e1 e2 e3 : expr (boolT_def f1 m_def f2)) :
  [[ e1 `** (e2 `** e3) ]]_dl2 = [[ (e1 `** e2) `** e3 ]]_dl2.
Proof.
by rewrite /= !big_ord_recl /= !big_ord_recl !big_ord0 !addr0 addrA (*assoc. of 
real number*).
Qed.
\end{minted} 
Proofs for \DLtwo{} are similarly succinct.

In contrast, the proofs for \Yager{} logic and \STL{} are more demanding. 
For example, the proof of associativity for \Yager{}:
\begin{minted}{ssr}
Theorem Yager_andA (e1 e2 e3 : expr boolT_def) : 0 < p ->
  [[ e1 `** (e2 `** e3) ]]_Yager = [[ (e1 `** e2) `** e3]]_Yager.
\end{minted}
consists of about 100 lines of code, as its interpretation
uses the power function of
\analysis{}, even with automation tactics of the Mathematical Components libraries, 
such as the decision procedure 
\coqin{lra}~\cite{sakaguchi2022itp,algebratactics}
for linear and rational arithmetics,
simplfying the formalisation significantly.
%

\subsection{Negative results}
\label{sec:logical-neg-results}

We finish this section with the negative results from
Table~\ref{tab:properties} along with their \rocq{} formalisations:

\begin{counterexample}
  Negation in \STLinfty{} logic is not expressible as $\impl \bot$. 
\end{counterexample}

In \rocq{} we formalise it as:
\begin{minted}{ssr}
Lemma stl_infty_not_neg_impl :
  exists e : expr boolT_fuzzy,  [[ `~e : expr boolT_fuzzy ]]_stli 
    <> [[ e `=> (dl_bool _ _ _ _ false)]]_stli.
\end{minted}
Note that we need not provide a counterexample for \STL{}, as it does not have implication.

\begin{counterexample}
  Negation in \Godel{} and \product{} logic is not involutive. 
\end{counterexample}

In \rocq{} we formalise it as:
\begin{minted}{ssr}
Lemma Godel_negation_not_involutive :
  exists e, [[ `~ `~ e ]]_Godel != [[ e ]]_Godel.
\end{minted}
This is proven similarly with a lemma for \product{} logic.

\begin{counterexample}
  \Luka{} and \Yager{} and \product{} logics are not monoidally idempotent.
\end{counterexample}

We show this in \rocq{} as:
\begin{minted}{ssr}
Lemma Lukasiewicz_not_idempotent :
  exists e, [[ e `** e ]]_Lukasiewicz <> [[ e ]]_Lukasiewicz.
\end{minted}
We provide similar proofs for \Yager{}, which is a generalisation of
\Luka{}, and for \product{} logics.

These lemmas rely on the existence of a value with interpretation that
is neither 0 nor 1.
We show this with:
\begin{minted}{ssr}
Lemma h1neq2 l : [[ dl_real 1 `== dl_real 2 ]]_l = 2/3.
\end{minted}

\section{Analytic properties of \DL{}s}
\label{sec:realanalysis}

In the previous section (Sect.~\ref{sec:algebraic}), we showed that
algebraic properties of \DL{}s can
be formalised using properties of the semantic interpretation.
In this section, we focus on properties of \DL{}s relying on real
analysis.

In Sect.~\ref{sec:STLconjunction}, we give justification for the definition of 
the connectives of \STLinfty{} defined as $\min$ and $\max$ in 
Table~\ref{tab:semantics-math}. In particular we show that if in the definition 
of $\tSTL{\naryConj{p_1,...,p_M}}$ the parameter $\nu$ tends towards $+\infty$ 
then $\tSTL{\naryConj{p_1,...,p_M}}$ tends towards 
$\min(\tSTL{p_1},...,\tSTL{p_M})$. An analogical convergence result can be 
proven for $\tSTL{\naryDisj{p_1,...,p_M}}$ and $\max$ though we do not provide 
it.
%
 This result is key to justifying the creation of
\STLinfty, a version of \STL{}
which has more intuitive logical properties and can be represented as
a residuated lattice (Sect.~\ref{sec:algebraic}).

In the following sections, we focus on shadow-lifting. 
We formally define shadow-lifting in Sect.~\ref{sec:shadow-lifting}
and prove it for \DLtwo{} and \product{} (which trivially enjoy
shadow-lifting) in Sect.~\ref{sec:sldltwoproduct}.
For \STL{}, this result was first proven (via a pencil-and-paper proof)
by Varnai and Dimarogonas~\cite[Sect.~V]{varnai} (along with the definition of
the \STL{} conjunction).
We both formalise this result and actually complete it
since the original proof only covers one of the two non-trivial
cases. The main technical aspect of the proof is high-school level
mathematics: an application of L'H\^opital's rule, which was not yet
available in \analysis{} and that we prove in Sect.~\ref{sec:lhopital}.
We finally provide an overview of the missing part of Varnai and Dimarogonas'
proof of shadow-lifting for \STL{} in Sect.~\ref{sec:slstl}.
Note that the logics \Godel{}, \Luka{}, \STLinfty{}, and \Yager{} fail 
shadow-lifting
as they are not differentiable everywhere, due to their use of $\min$
or $\max$ to define conjunction.

\subsection{\STL{} conjunction as a lattice connective}
\label{sec:STLconjunction}

As mentioned in Sect.~\ref{sec:propertiesDLs}, we cannot consider
\STL{} connectives as lattice connectives in their standard form
(Table~\ref{tab:semantics-math}) since they are non-associative by 
design~\cite{varnai}
and therefore do not fulfill the requirement of a lattice (see Definition~\ref{def:lattice}).
They can instead be
viewed as approximations of the common lattice operators $\min$ and
$\max$.
This amounts to proving the following limit, where $p_1,...,p_M$ are
formulas:
$$\lim_{\nu\to\infty} \tSTL{\naryConj (p_1,...,p_M)} = 
\min\large(\tSTL{p_1},...,\tSTL{p_M}\large).$$
 
As the encoding of the semantics of the conjunction provides separate functions 
for the different 
cases (Fig.~\ref{fig:semantics-lt0-gt0}), we prove this property separately 
for each case. 
Consider the case where $p_{\min}= \min(\tSTL{p_1}, \dots, \tSTL{p_M}) > 0$
and we have 
$\widetilde{p_i} = \dfrac{\tSTL{p_i} - p_{\min}}{p_{\min}}$. 
The proof relies on noting 
that when $p_i = p_{\min}$, $\widetilde{p_i} = 0$ and it is positive otherwise. 
In the latter case, we know that $\lim_{\nu\to\infty} -\nu \widetilde{p_i} = 
-\infty$ 
and therefore the exponential $e^{-\nu \widetilde{p_i}}$ tends towards zero.
The proof relies on splitting the sums to take 
advantage of those properties:
\begin{flalign*}
& \lim_{\nu\to\infty} \tSTL{\naryConj (p_1,...,p_M)} \\
& = \lim_{\nu\to\infty} \dfrac{\sum_{i=1}^M \tSTL{p_i}
  e^{-\nu \widetilde{p_i}}}{\sum_{i=1}^M e^{-\nu \widetilde{p_i}}} \\
& = \lim_{\nu\to\infty} \dfrac{\sum_{i=1,{\tSTL{p_i} = p_{\min}}}^M \tSTL{p_i}
e^{-\nu \widetilde{p_i}} + 
\sum_{i =1, {\tSTL{p_i} \not= p_{\min}}}^{M} \tSTL{p_i} e^{-\nu 
\widetilde{p_i}}}{\sum_{i=1,{\tSTL{p_i} = p_{\min}}}^M 
e^{-\nu \widetilde{p_i}} + \sum_{i=1,{\tSTL{p_i} \not= p_{\min}}}^M
e^{-\nu \widetilde{p_i}}} \\
& = \dfrac{\lim_{\nu\to\infty}\sum_{i=1,{\tSTL{p_i} = p_{\min}}}^M \tSTL{p_i} 
e^0 + 
 \lim_{\nu\to\infty}\sum_{i=1,{\tSTL{p_i} \not= p_{\min}}}^M \tSTL{p_i} e^{-\nu 
 \widetilde{p_i}}}{\lim_{\nu\to\infty}\sum_{i=1,{\tSTL{p_i} = p_{\min}}}^M 
1 + \lim_{\nu\to\infty}\sum_{i=1,{\tSTL{p_i} \not= p_{\min}}}^M
e^{-\nu \widetilde{p_i}}} \\
& = \dfrac{\lim_{\nu\to\infty}\sum_{i=1,{\tSTL{p_i} = p_{\min}}}^M
\tSTL{p_i}}{\lim_{\nu\to\infty}\sum_{i=1,{\tSTL{p_i} = p_{\min}}}^M 1}
= p_{\min}
\end{flalign*}


To formally state this property, we use the generic notation of a
limit from \analysis{}, \coqin{f x @[x --> F] --> l}, which stands for
$\mylim{f(x)}{l}{x}{F}$ where $F$ is a filter.  The semantics of this
case is represented by the function \coqin{stl_and_gt0}.
\begin{minted}{ssr}
Lemma stl_and_gt0_cvg_infty (p : R) (v : seq R)  : 
  v != [::] ->
  (forall x, x \in v -> x > 0) ->
  (stl_and_gt0 p v) @[p --> +oo] --> \big[minr/head 0 v]_(i <- v) i.
\end{minted}

While the mathematical proof is straightforward, the corresponding
mechanisation is less so with the main factors being the amount of
sub-limits that need to be proven---and their existence being reliant
on multiple conditions, e.g., positivity of all arguments or the
minimum---and additionally due to the necessity to develop a number of
lemmas on the interactions of iterated operators (\coqin{\big}
notation) and \coqin{minr}.

This result provides justification for the introduction of \STLinfty{} in which 
the semantics of conjunction are defined by $\min$.
We note that while $\STLinfty$'s connectives are associative 
and idempotent, it comes at the cost of shadow-lifting; all three cannot 
hold at once~\cite{varnai}.

\subsection{Formalisation of shadow-lifting}
\label{sec:shadow-lifting}

As seen in Sect.~\ref{sec:axes-of-study},
shadow-lifting is defined in terms of partial derivatives, for which
there was however no theory yet in \analysis{}. They can be
easily defined on the model of derivatives \cite[\coqin{derive.v}]{analysis}.
First, we define \newterm{error vectors} as row vectors (type \coqin{'rV[R]_k} 
where \coqin{R} is some ring
and \coqin{k} the size of the vector)
that are $0$ everywhere except at one coordinate \coqin{i}:
\begin{minted}{ssr}
Definition err_vec {R : ringType} (i : 'I_n.+1) : 'rV[R]_n.+1 := 
  \row_(j < n.+1) (i == j)%:R.
\end{minted}
The notation \mintinline{coq}{
note that here the Boolean equality (notation \coqin{==}) is implicitly coerced 
to a natural number.
Then, given a function~\coqin{f} that takes as input a row vector, we
define a function \coqin{partial} that given a row vector~\coqin{a}
and an index~\coqin{i} returns the limit
$\lim_{\substack{h\to 0\\h\neq 0}}\frac{\texttt{f}(\texttt{a} +
  h\texttt{err_vec i}) - \texttt{f}(\texttt{a})}{h}.$
Put formally:
\begin{minted}{ssr}
Definition partial {R} {n} (f : 'rV[R]_n.+1 -> R) (a : 'rV[R]_n.+1) i :=
  lim (h^-1 * (f (a + h *: err_vec i) - f a) @[h --> 0^']).
\end{minted}
In this syntax, \coqin{0^'} represents the deleted neighbourhood of \coqin{0},
and \coqin{lim g @ F} represents the limit of the function \coqin{g} at the 
filter \coqin{F} \cite{affeldt2018jfr}.
The notation \coqin{*:} represents scaling but is equivalent to the 
multiplication of real numbers here.
Hereafter, we use the \rocq{} notation \coqin{d f '/d i} (reminiscent of 
$\pdv{f}{x_i}$) for \coqin{partial f i}.

Using partial derivatives, the definition of shadow-lifting 
(Definition~\ref{def:shadow}) translates directly into \rocq{}:
\begin{minted}{ssr}
Definition shadow_lifting {R : realType} n (f : 'rV_n.+1 -> R) :=
  forall p, p > 0 -> forall i, ('d f '/d i) (const_mx p) > 0.
\end{minted}
The \coqin{const_mx} function comes from \mathcomp's matrix theory and 
represents a matrix where all coefficients are the given constant; we use it to 
implement the restriction ``$p_j = p$'' seen in Definition~\ref{def:shadow}.

\subsection{Shadow-lifting for \DLtwo{} and \product{}}
\label{sec:sldltwoproduct}

The proof of shadow-lifting for \DLtwo{} and \product{} \DL{}s provides
an easy illustration of the use of the definition of the previous section 
(Sect.~\ref{sec:shadow-lifting}).

For \DLtwo{}, the first thing to observe is that the semantics of a
vector of real numbers can simply be written as an iterated sum
using the notation \coqin{``_} to address elements of row-vectors:
\begin{minted}{ssr}
Definition dl2_and {R : fieldType} {n} (v : 'rV[R]_n) := \sum_(i < n) v ``_ i.
\end{minted}
Shadow-lifting for \DLtwo{} really just amounts to checking that the
partial derivatives of the function
$\vec{v}\mapsto \sum_{j < |\vec{v}|}\vec{v}_j$ are 1, i.e.,
considering vectors of size \coqin{M.+1}:
\begin{minted}{ssr}
Lemma shadowlifting_dl2_andE (p : R) : p > 0 ->
  forall i, ('d (@dl2_and R M.+1) '/d i) (const_mx p) = 1.
\end{minted}
Since the partial derivatives are all positive, \DLtwo{} satisfies the
\coqin{shadow_lifting} predicate, see \cite[\coqin{dl2.v}]{github}.

Similarly, we observe for the \product{} \DL{} that the semantics of a vector
is the function $\vec{v}\mapsto \prod_{j < |\vec{v}|}\vec{v}_j$ whose 
partial derivatives are $p^M$, which is positive, see 
\cite[\coqin{fuzzy.v}]{github}:
\begin{minted}{ssr}
Lemma shadowlifting_product_andE p : p > 0 ->
  forall i, ('d (@product_and R M.+1) '/d i) (const_mx p) = p ^+ M.
\end{minted}

\subsection{Formalisation of L'H\^opital's rule using \analysis{}}
\label{sec:lhopital}

As indicated in the introduction of this section, the key lemma to
prove shadow-lifting for \STL{} is L'H\^opital's rule.
Here follows a standard statement of it:
\begin{theorem}[L'H\^opital's rule {\cite[Thm.\ 5.13]{rudin1976}}]
\label{thm:lhopitalrudin}
Suppose $f$ and $g$ are real and differentiable in $\itvoo{a}{b}$, and 
$g'(x)\neq 0$
for all $x\in \itvoo{a}{b}$, where $-\infty\leq a < b\leq +\infty$. Suppose
\begin{align}
\frac{f'(x)}{g'(x)}\to A \textrm{ as }x\to a. \label{eqn:lhopital13}
\end{align}
If
\begin{align}
f(x)\to 0 \textrm{ and }g(x)\to 0 \textrm{ as }x\to a, \label{eqn:lhopital14}
\end{align}
or if
\begin{align}
g(x)\to +\infty \textrm{ as }x\to a, \label{eqn:lhopital15}
\end{align}
then
\begin{align}
\frac{f(x)}{g(x)}\to A \textrm{ as }x\to a. \label{eqn:lhopital16}
\end{align}
The analogous statement is of course also true if $x\to b$, or if 
$g(x)\to-\infty$ in
\eqref{eqn:lhopital15}.
\end{theorem}

To formally state L'H\^opital's rule,
the main ingredients from \analysis{} we use are the relation
\coqin{is_derive} between a function and its derivative at some
point\footnote{The relation {\tt is\us{}derive} also takes into
  account the direction of the derivative, hence the {\tt 1} in the
  formal statement of L'H\^opital's rule, which is dealing with real
  functions.}  and the generic notation for a limit.

Here follows our formal statement of L'H\^opital's rule in the case
where $a$ and $l$ are real numbers.
One difference with the informal statement is that we make explicit
the fact that the limit is taken when $a$ is approached from the right
by using the \newterm{right filter\/} \coqin{a^'+}, i.e., the filter
of neighbourhoods of \coqin{a} intersected with $\itvoo{a}{+\infty}$.
Otherwise, the formal statement syntactically matches the informal
one.
The fact that $f$ and $g$ are differentiable in $\itvoo{a}{b}$ is
stated at lines~\ref{lhopital:fdf}--\ref{lhopital:gdg} and the fact
that $g'(x)\neq 0$ in $\itvoo{a}{b}$ is stated at
line~\ref{lhopital:cdg}. The limits of $f(x)$ and $g(x)$ as $x$
approaches $a$ appear at line~\ref{lhopital:fa0ga0} (this corresponds
to Equation \eqref{eqn:lhopital14} in Rudin's statement).
The fact that the limits of $\frac{f'(x)}{g'(x)}$ and
$\frac{f(x)}{g(x)}$ are the same (more precisely, that Equation
\eqref{eqn:lhopital13} implies Equation \eqref{eqn:lhopital16}) is the
conclusion at line~\ref{lhopital:ccl}:
\begin{minted}[numbers=left,xleftmargin=1.5em,escapeinside=77]{ssr}
Context {R : realType}.
Variables (f df g dg : R -> R) (a b : R) (l : R) (ab : a < b).
Hypotheses (fdf : forall x, x \in `]a, b[ -> is_derive x 1 f (df x)) 7\phantomsection\label{lhopital:fdf}7
           (gdg : forall x, x \in `]a, b[ -> is_derive x 1 g (dg x)). 7\label{lhopital:gdg}7
Hypotheses (fa0 : f x @[x --> a^'+] --> 0) (ga0 : g x @[x --> a^'+] --> 0) 7\label{lhopital:fa0ga0}7
           (cdg : forall x, x \in `]a, b[ -> dg x != 0). 7\label{lhopital:cdg}7

Lemma lhopital_at_right :
  df x / dg x @[x --> a^'+] --> l -> f x / g x @[x --> a^'+] --> l. 7\label{lhopital:ccl}7
\end{minted}

When formalising Rudin's proof of L'H\^opital's rule, we found out
that it is a good illustration of difficulties that typically occur
when dealing formally with continuity in real analysis. Below we
comment on the proof by focusing on this aspect.

Rudin's proof of L'H\^opital's rule relies on a generalisation of the
Mean Value Theorem often referred to as Cauchy's Mean Value Theorem:
\begin{theorem}[Cauchy's Mean Value Theorem]
\label{thm:cauchyMVT}
Suppose $f$ and $g$ are real, differentiable in $\itvoo{a}{b}$, and
continuous {\em within\/} $\itvcc{a}{b}$, i.e., that they are
continuous in $\itvoo{a}{b}$ and that $\mylim{f(x)}{f(a)}{x}{a^+}$ and
$\mylim{f(x)}{f(b)}{x}{b^-}$ (and resp.\ for $g$). Suppose moreover
that $g'(x)\neq 0$ for all $x\in\itvoo{a}{b}$. Then, there exists
$c\in\itvoo{a}{b}$ such that $\frac{f'(c)}{g'(c)}=\frac{f(b)-f(a)}{g(b)-g(a)}.$
\end{theorem}
Note that we do not require $f$ and $g$ to be continuous on
$\itvcc{a}{b}$, which is too strong because indirectly assuming
limit-values on the left of~$a$ and on the right of~$b$.
The proof of Cauchy's Mean Value Theorem in \analysis{} can be carried
out using the already-available Mean Value and Rolle's Theorems
\cite[\coqin{derive.v}]{analysis} \cite[Sect.~A.2.5.]{affeldt2018jfr}.

We are now ready to sketch the proof of L'H\^opital's rule.
We use Cauchy's Mean Value Theorem to prove:
\begin{align}
\forall q, l < q \to
  \exists c_2, a < c_2 \land \forall x, a < x < c_2 \to \frac{f(x)}{g(x)} < q 
  \label{eqn:cauchymvt1}
\end{align}
and
\begin{align}
\forall p, p < l \to
  \exists c_3, a < c_3 \land \forall x, a < x < c_3 \to p < \frac{f(x)}{g(x)} 
  \label{eqn:cauchymvt2}.
\end{align}
Since our goal is to prove $\mylim{\frac{f(x)}{g(x)}}{l}{x}{a^+}$, we
can pick an $e>0$ and use Equations \eqref{eqn:cauchymvt1} and
\eqref{eqn:cauchymvt2} to produce a $c_2$ and a $c_3$ (using
$q = l + e$ and $p = l - e$) such that $\frac{f(x)}{g(x)}$ can be
chosen to be close enough to~$l$.

Let us take a closer look at the proof of Equation \eqref{eqn:cauchymvt1}.
First we prove 
\begin{align}
\exists c, c\in\itvoo{a}{b} \land \forall x, a < x < c \to \frac{f'(x)}{g'(x)} 
< r
\label{eqn:use13}
\end{align}
for some $r$ near (the right of) $l$ using Equation
\eqref{eqn:lhopital13} and the {\tt near} tactics of \analysis{}
\cite[Sect.~3.2]{affeldt2018jfr}.
Second we show that 
\begin{align}
\forall x, y, a < x < y < c \to \frac{f(x)-f(y)}{g(x)-g(y)}<r 
\label{eqn:usecauchy}
\end{align}
using Cauchy's Mean Value Theorem. Finally, we use Equations
\eqref{eqn:use13} and \eqref{eqn:usecauchy} to show
$$
\forall y, a < y < c \to \frac{f(y)}{g(y)}\leq r < q
$$
and conclude. The difficulty is the following.
For the last step, we need to take the limit of the left-hand side
of the conclusion of formula \eqref{eqn:usecauchy} when $x\to a^+$.
This means that we implicitly assume that $g(y)$ is not $0$.
It is possible because if $g(y)$ is 0 then we can establish a
contradiction using the fact that $g(a)=g(y)=0$.
However, nothing is actually assumed explicitly about the value of
$g(a)$; it is not necessarily $0$ because we only have continuity from
the right.
What happens is that Rudin implicitly assumes that $g$ and $f$ are
extended by continuity by taking $g(a)=f(a)=0$.
For this reason, the very first step of the formal of proof of
L'H\^opital's rule is to define $g_0$ and $f_0$ by continuity and work
internally with these functions. See
\cite[\coqin{realfun.v}]{analysis} for details.

The case of Theorem \ref{thm:lhopitalrudin} where the limit is taken
as $b$ is approached from the left is a direct consequence of
\coqin{lhopital_at_right} using the properties of negation.  The case
where the limit is taken for any $c\in\itvoo{a}{b}$ (regardless of
right or left limits) is a also a direct consequence.  We have not
formalised the cases where the limits involve $\pm\infty$ because they
are not needed to prove shadow-lifting of \STL{}.

\subsection{Shadow-lifting for \STL{}}
\label{sec:slstl}

Compared with \DLtwo{} and \product{}, the conjunction of \STL{}
(Table~\ref{tab:semantics-math}) is much more involved: it
consists of two non-trivial cases (marked as $p_{\min} < 0$ and $p_{\min} > 0$
in Table~\ref{tab:semantics-math}) 
whose computation requires summations of exponentials of deviations.
Varnai and Dimarogonas provide a proof sketch for the case
$p_{\min} > 0$~\cite[Sect.~V]{varnai} which we have successfully
formalised, using in particular l'H\^opital's rule from the previous
section (Sect.~\ref{sec:lhopital}). Below we explain the formalisation of the 
other case
$p_{\min} < 0$ that Varnai and Dimarogonas did not treat.

The case $p_{\min} < 0$ actually refers to the semantics provided
by the function \coqin{stl_and_lt0} already presented in
Fig.~\ref{fig:semantics-lt0-gt0}.
The positive limit we are looking for is actually $\frac{1}{M+1}$
(where $M$ is the size of vectors), i.e., our goal is to prove formally
the following (the notation \coqin{\o} is for function composition and we
are considering vectors of size \coqin{M.+1}):
\begin{minted}{ssr}
Lemma shadowlifting_stl_and_lt0 (p : R) : p > 0 -> forall i,
  ('d (@stl_and_lt0 M.+1 \o @fun_of_rV _ _) '/d i) (const_mx p) = M.+2%:R^-1.
\end{minted}
This boils down to proving the existence of the limit ``from below''
and ``from above''.
The ``from below'' case consists in the following convergence lemma:
\begin{minted}{ssr}
Lemma shadowlifting_stl_and_lt0_cvg_at_left (p : R) i : p > 0 ->
  h^-1 *
  (stl_and_lt0 (fun_of_rV M.+1 (const_mx p + h *: err_vec i)) -
   stl_and_lt0 (fun_of_rV M.+1 (const_mx p))) @[h --> 0^'-] --> M.+2%:R^-1.
\end{minted}
For the sake of clarity, let us switch to standard mathematical
notations and assume without loss of generality that \coqin{i} is actually 
\coqin{M}.
By mere algebraic transformations (using \mathcomp's algebra theory),
the goal can be turned into a sum of two limits:
$$
\begin{array}{ll}
\lim\limits_{h \rightarrow 0^-} \dfrac{\tSTL{\naryConj(\pval, \ldots, \pval, 
		\pval + h)} - \tSTL{\naryConj(\pval, \ldots, \pval)}}{h} \\ 
= \lim\limits_{h \rightarrow 0^-} \frac{1}{h} \left(
	\dfrac{(\pval + h) M e^{\frac{-h}{\pval+h}} e^{\nu 
	\frac{-h}{\pval+h}} + \pval + h }{Me^{\nu \frac{-h}{\pval+h}} + 1}
	- p
	\right) & \text{by definition (see Table~\ref{tab:semantics-math})}\\
= \underbrace{\lim\limits_{h \rightarrow 0^-} \frac{h}{h(M + e^{\nu 
\frac{h}{\pval+h}})}}_{(a)}
 +
\underbrace{\lim\limits_{h \rightarrow 0^-} \frac{M (\pval + 
h)e^{\frac{-h}{\pval+h}} - \pval M }{h(M + e^{\nu \frac{h}{\pval+h}})}}_{(b)} & 
\text{by simplification} \\
\end{array}
$$
We can show directly that $(a)=\frac{1}{M+1}$
but the computation of $(b)$ requires L'H\^opital's rule:
$$
\begin{array}{rcl}
(b) & = & \lim\limits_{h \rightarrow 0^-} \frac{\frac{h M 
	e^{\frac{-h}{\pval+h}}}{\pval + h}}
	{e^{\nu \frac{h}{\pval+h}} +
	h e^{\nu\frac{-h}{\pval+h}}(\frac{\nu}{\pval + h} - \frac{h \nu}{(\pval 
	+h)^2}) + M} \\
& = & \lim\limits_{h \rightarrow 0^-} h 
	\lim\limits_{h \rightarrow 0^-} \frac{M e^{\frac{-h}{\pval+h}}}{\pval + h}
	\lim\limits_{h \rightarrow 0^-} \frac{1}{e^{\nu \frac{h}{\pval+h}} +
		h e^{\nu\frac{-h}{\pval+h}}(\frac{\nu}{\pval + h} - \frac{h 
		\nu}{(\pval 
			+h)^2}) + M} \\
& = & 0 \cdot \frac{M}{\pval} \cdot \frac{1}{1 + M} = 0 \\
\end{array}
$$
 
Barring the necessity of finding the most convenient breakdown of the
limit in the two penultimate steps, this proof is arguably
mathematically straightforward.  The corresponding mechanised proof is
however significantly less trivial than in the cases of \product{} and
\DLtwo{} (Sect.~\ref{sec:sldltwoproduct}): length-wise the first
tentative formal proof we wrote was an order of magnitude larger.

Proving the ``from above'' above consists of a simpler but similar
argument:
\begin{flalign*}
& \lim\limits_{h \rightarrow 0^+} \dfrac{\tSTL{\naryConj(\pval, \ldots, \pval, 
		\pval + h)} - \tSTL{\naryConj(\pval, \ldots, \pval)}}{h}
\\
& = \lim\limits_{h \rightarrow 0^+}  \frac{1}{h} \left (\frac{\pval M + 		
e^{\frac{h}{\pval}} e^{\frac{\nu h}{\pval}}} {M + e^{\frac{h}{\pval}}} - 
p\right) \\
& = \lim\limits_{h \rightarrow 0^+} \frac{1}{M + e^{\frac{\nu h}{\pval}}}
\lim\limits_{h \rightarrow 0^+} e^{\frac{\nu h}{\pval}}
\lim\limits_{h \rightarrow 0^+} \frac{ e^{\frac{ h}{\pval}} 
-1}{\frac{h}{\pval}}
= \frac{1}{M + 1} \cdot 1 \cdot 1 = \frac{1}{M + 1}\\
\end{flalign*}

Combined with the formalisation of the case $p_{\min} > 0$
sketched by Varnai and Dimarogonas~\cite[Sect.~V]{varnai}, this completes
the formal proof of shadow-lifting for \STL{}.

\section{Proof-theoretic properties of \DL{}s}
\label{sec:soundness}

In the previous sections we showed the formalisation of properties relying 
both on algebra (Sect.~\ref{sec:algebraic}) and real analysis
(Sect.~\ref{sec:realanalysis}), leaving this section to define and explore
proof-theoretic properties of \DL{}s.

The main property we investigate is soundness (defined shortly in subsequent 
sections).
The soundness of fuzzy logics was proven by Metcalfe et 
al.~\cite{fuzzy-proof,metcalfe2004analytic}.
This section provides the first 
formalisation of this result and extends the existing calculi of fuzzy logics 
with derivable rules for 
the remaining connectives in comparison to the minimal fragments presented in 
Figures~\ref{fig:godel-seq-calc}--\ref{fig:product-seq-calc}.
Last, this section presents new calculi along with soundness and weak completeness 
proofs for \DLtwo{} as well as \STLinfty{}, neither of 
which has previously appeared 
in the literature. 

\subsection{Hypersequent calculi for fuzzy logics}
\label{sec:proof-theoretic}
Real-valued logics, and in particular fuzzy logics, have been studied 
proof-theoretically for decades. The full exposition of their proof-theoretic 
significance is beyond the scope of this paper for which we refer the interested 
reader to the existing comprehensive 
texts~\cite{fuzzy-proof,galatos2007residuated}.

Generally, these logics are shown to be more expressive than intuitionistic 
logic, 
but less expressive than classical logic, and for this reason they are 
sometimes called \newterm{superintuitionistic 
logics}~\cite{galatos2007residuated}. 
Both Hilbert-style and Gentzen-style calculi for fuzzy logics exist. We
only consider Gentzen-style (or sequent) calculi here or, more precisely,
we consider
\newterm{hypersequent} calculi, as they have been established as the most 
standard way to define fuzzy logics~\cite[Chapt.~VI]{fuzzy-proof}. The calculi 
presented in this work follow exactly those established in literature---which 
results in a single conclusion calculus for \Godel{} and multiple conclusion 
calculi for both \Luka{} and \product{}.


Finally, recall that the connectives $\odot$ and $\impl$ correspond to the 
multiplicative fragment of linear logic, and that the connectives $\land$, 
$\lor$, and $\neg$
correspond to the additive fragment of linear logic. Thus, absence of some of 
the  structural rules should not come as 
a surprise; this has been studied as well in, e.g., 
\cite{galatos2007residuated,fuzzy-proof}.  

\begin{definition}[Sequent, Hypersequent~\cite{fuzzy-proof}]
\label{def:hyper}
A \newterm{sequent} is an ordered pair of finite lists of formulas, written
$q_1, \ldots, q_n \vdash p_1, \ldots, p_m$. 

A \newterm{hypersequent\/} is a finite list of sequents of the form
$$ \seqOne_1 \vdash \seqTwo_1 | \ldots\ |\ \seqOne_n \vdash \seqTwo_n$$ 
where, for all $i \in \{1, \ldots, n\}$, $\seqOne_i \vdash \seqTwo_i$ is a 
sequent.
\end{definition} 

\begin{figure}[t]
		\centering
        \footnotesize{
\begin{gather*}
\infer[\init]{\hseqOne\ |\ \forml \vdash \forml}{}
\semSpace
\infer[(\bot)]{\hseqOne\ |\ \seqOne, \bot \vdash 
\seqTwo}{}
\semSpace
\infer[(\top)]{\hseqOne\ |\ \seqOne \vdash 
\top}{}
\\
\\
\infer[\ew]{\hseqOne\ |\ \hseqTwo}{\hseqOne}
\semSpace
\infer[\ec]{\hseqOne\ |\ \hseqTwo}{\hseqOne | \hseqTwo | \hseqTwo}
\semSpace
\infer[\eex]{\hseqOne_1\ |\ \hseqTwo_2\ |\ \hseqTwo_1\ |\ 
\hseqOne_2}{\hseqOne_1\ |\ \hseqTwo_1\ |\ \hseqTwo_2\ |\ \hseqOne_2}
\\ \\
\infer[\com]{\hseqOne\ |\ \seqOne_1, \seqOne_2 \vdash \seqTwo\ |\ 
\seqThree_1, \seqThree_2 \vdash \seqFour}{\hseqOne\ |\ \seqOne_1, \seqThree_1 
\vdash \seqTwo\ \ \ \ \hseqOne|\ \seqOne_2, \seqThree_2 \vdash 
 \seqFour}
\\ \\
\infer[\weak]{\hseqOne\ |\ \seqOne,\ \seqThree \vdash\seqTwo}{
	\hseqOne\ |\ \seqOne \vdash \seqTwo}
\semSpace
\infer[\contr]{\hseqOne\ |\ \seqOne,\ \seqThree \vdash\seqTwo}{
	\hseqOne\ |\ \seqOne,\seqThree,\seqThree \vdash\seqTwo}
\\ 
\\
\infer[\lex]{\hseqOne\ |\ \seqOne_1, \forml_0, \forml_1, \seqOne_2 \vdash 
\seqTwo}{
\hseqOne\ |\ \seqOne_1, \forml_1, \forml_0, \seqOne_2 \vdash \seqTwo}
\semSpace
\infer[\rex]{\hseqOne\ |\ \seqOne \vdash 
\seqTwo_1, \forml_0, \forml_1, \seqTwo_2 }{
\hseqOne\ |\ \seqOne \vdash 
\seqTwo_1, \forml_1, \forml_0, \seqTwo_2}
	\\
	\\
	\infer[(\text{L}\wedge)]{\hseqOne\ |\ \seqOne, \forml_0 \wedge \forml_1 
	\vdash 
	\seqTwo}{\hseqOne\ |\ \seqOne, \forml_0 \vdash \seqTwo\  |\ 
	\seqOne, \forml_1 \vdash \seqTwo}
	\semSpace
	\infer[(\text{R}\wedge)]{\hseqOne\ 
	|\ \seqOne \vdash \forml_0 \wedge \forml_1}{
	\hseqOne\ |\ \seqOne \vdash \forml_0\ \ \ \ \hseqOne\ |\ \seqOne 
	\vdash \forml_1}
	\\
	\\
	\infer[(\text{L}\vee)]{\hseqOne\ |\ \seqOne, \forml_0 \vee \forml_1 \vdash 
	\seqTwo}{\hseqOne\ |\ \seqOne, 
	\forml_0 \vdash \seqTwo \ \  \hseqOne\ |\ \seqOne, \forml_1 \vdash \seqTwo}
	\semSpace
	\infer[(\text{R}\vee)]{\hseqOne\ |\ 
		\seqOne \vdash 
		\forml_0 \vee \forml_1}{\hseqOne\ |\ \seqOne \vdash \forml_0\ |\ 
		\seqOne \vdash \forml_1}
	\\
	\\
	\infer[(\text{L}\impl)]{\hseqOne\ |\ \seqOne, \forml_0 \impl \forml_1 
	\vdash 
	\seqTwo}{\hseqOne\ |\ \seqOne \vdash \forml_0 \ \ \hseqOne\ |\ 
	\seqOne, \forml_1 \vdash 
	\seqTwo}
	\semSpace
	\infer[(\text{R}\impl)]{\hseqOne\ |\ \seqOne \vdash \forml_0 \impl 
	\forml_1}{\hseqOne\ |\ \seqOne, 
	\forml_0 \vdash \forml_1}
\end{gather*}}
\caption{Hypersequent calculus for \godel{} logic where $\hseqOne$ and 
$\hseqTwo$ 
are hypersequents, $\seqOne,\ \seqThree,\ \seqTwo,\ \seqFour$ 
are sequents and $\forml,\forml_0,\forml_1$ are formulas.}
\label{fig:godel-seq-calc}
\end{figure}

	\begin{figure}[t]
		\centering
                \footnotesize{
\begin{gather*}
\infer[\init]{\hseqOne\ |\ \forml_0 \vdash \forml_0}{}
\semSpace
\infer[\emp]{\hseqOne\ |\  \vdash }{}
\semSpace
\infer[(\bot)]{\hseqOne\ |\ \seqOne, \bot \vdash  \forml_0}{}
\\ \\
%
\infer[\ew]{\hseqOne\ |\ \hseqTwo}{\hseqOne}
\semSpace
\infer[\ec]{\hseqOne\ |\ \hseqTwo}{\hseqOne\ |\ \hseqTwo\ |\ \hseqTwo}
\semSpace
\infer[\eex]{\hseqOne_1\ |\ \hseqTwo_2\ |\ \hseqTwo_1\ |\ 
\hseqOne_2}{\hseqOne_1\ |\ \hseqTwo_1\ |\ \hseqTwo_2\ |\ \hseqOne_2}
\\ \\
\semSpace
\infer[\weak]{\hseqOne\ |\ \seqOne,\ \seqThree \vdash\seqTwo}{
	\hseqOne\ |\ \seqOne \vdash \seqTwo}
\\
\\
\infer[\lex]{\hseqOne\ |\ \seqOne_1, \forml_0, \forml_1, \seqOne_2 \vdash 
\seqTwo}{
\hseqOne\ |\ \seqOne_1, \forml_1, \forml_0, \seqOne_2 \vdash \seqTwo}
\semSpace
\infer[\rex]{\hseqOne\ |\ \seqOne \vdash 
\seqTwo_1, \forml_0, \forml_1, \seqTwo_2 }{
\hseqOne\ |\ \seqOne \vdash 
\seqTwo_1, \forml_1, \forml_0, \seqTwo_2}
\\ \\
\infer[\spl]{\hseqOne\ |\ \seqOne_1 \vdash\seqTwo_1\ |\ \seqOne_2 
\vdash\seqTwo_2}{\hseqOne\ |\ \seqOne_1, \seqOne_2 \vdash \seqTwo_1, \seqTwo_2}
\semSpace
\infer[\mix]{\hseqOne\ |\ \seqOne_1, \seqOne_2 \vdash \seqTwo_1, 
\seqTwo_2}{\hseqOne\ |\ \seqOne_1 \vdash\seqTwo_1\ \ \ \hseqOne\ |\ \seqOne_2 
\vdash\seqTwo_2}
\\ \\
	\infer[(\text{L}\impl)]{\hseqOne\ |\ \seqOne, \forml_0 \impl \forml_1 
	\vdash 
	\seqTwo}{\hseqOne\ |\ \seqOne, \forml_1 \vdash \forml_0, \seqTwo}
	\semSpace
	\infer[(\text{R}\impl)]{\hseqOne\ |\ \seqOne \vdash \forml_0 \impl 
	\forml_1, 
	\seqTwo}{\hseqOne\ |\ \seqOne \vdash \seqTwo\ \ \ \hseqOne\ |\ \seqOne, 
	\forml_0 \vdash \forml_1, \seqTwo}
\end{gather*}}
\caption{Hypersequent calculus for Łukasiewicz logic where $\hseqOne$ and 
$\hseqTwo$ 
are hypersequents, $\seqOne,\ \seqThree,\ \seqTwo,\ \seqFour$ 
are sequents and $\forml_0,\ \forml_1$ are formulas.}
\label{fig:luka-seq-calc}
\end{figure}

\begin{figure}[t]
		\centering
                \footnotesize{
\begin{gather*}
\infer[\init]{\hseqOne\ |\ \forml_0 \vdash \forml_0}{}
\semSpace
\infer[\emp]{\hseqOne\ |\  \vdash }{}
\semSpace
\infer[(\bot)]{\hseqOne\ |\ \seqOne, \bot \vdash  \seqTwo}{}
\\ \\
\infer[\ew]{\hseqOne\ |\ \hseqTwo}{\hseqOne}
\semSpace
\infer[\ec]{\hseqOne\ |\ \hseqTwo}{\hseqOne\ |\ \hseqTwo\ |\ \hseqTwo}
\semSpace
\infer[\eex]{\hseqOne_1\ |\ \hseqTwo_2\ |\ \hseqTwo_1\ |\ 
\hseqOne_2}{\hseqOne_1\ |\ \hseqTwo_1\ |\ \hseqTwo_2\ |\ \hseqOne_2}
\\ \\
\infer[\weak]{\hseqOne\ |\ \seqOne,\ \seqThree \vdash\seqTwo}{
	\hseqOne\ |\ \seqOne \vdash \seqTwo}
\\
\\
	\infer[\lex]{\hseqOne\ |\ \seqOne_1, \forml_0, \forml_1, \seqOne_2 \vdash 
	\seqTwo}{
	\hseqOne\ |\ \seqOne_1, \forml_1, \forml_0, \seqOne_2 \vdash \seqTwo}
	\semSpace
	\infer[\rex]{\hseqOne\ |\ \seqOne \vdash 
	\seqTwo_1, \forml_0, \forml_1, \seqTwo_2 }{
	\hseqOne\ |\ \seqOne \vdash 
	\seqTwo_1, \forml_1, \forml_0, \seqTwo_2}
\\ \\
\infer[\spl]{\hseqOne\ |\ \seqOne_1 \vdash\seqTwo_1\ |\ \seqOne_2 
\vdash\seqTwo_2}{\hseqOne\ |\ \seqOne_1, \seqOne_2 \vdash \seqTwo_1, \seqTwo_2}
\semSpace
\infer[\mix]{\hseqOne\ |\ \seqOne_1, \seqOne_2 \vdash \seqTwo_1, 
\seqTwo_2}{\hseqOne\ |\ \seqOne_1 \vdash\seqTwo_1\ \ \ \hseqOne\ |\ \seqOne_2 
\vdash\seqTwo_2}
	\\ \\
	\infer[(\text{L}\neg)]{\hseqOne\ |\ \seqOne, \neg \forml_0 \vdash \seqTwo}
	{\hseqOne\ |\ \seqOne \vdash \forml_0}
	\semSpace
	\infer[(\text{R}\odot)]{\hseqOne\ 
		|\ \seqOne \vdash \forml_0 \odot \forml_1, \seqTwo}{
		\hseqOne\ |\ \seqOne \vdash \forml_0, \forml_1, \seqTwo}
	\semSpace
	\infer[(\text{L}\odot)]{\hseqOne\ |\ \seqOne, \forml_0 \odot \forml_1 
	\vdash 
	\seqTwo}{\hseqOne\ |\ \seqOne, \forml_0, \forml_1 \vdash \seqTwo}
	\\
	\\
	\infer[(\text{L}\impl)]{\hseqOne\ |\ \seqOne, \forml_0 \impl \forml_1 
	\vdash 
	\seqTwo}{\hseqOne\ |\ \seqOne, \neg \forml_0 \vdash \seqTwo \ \ \ \hseqOne\ 
	|\ \seqOne, \forml_1 \vdash \forml_0, \seqTwo}
	\semSpace
	\infer[(\text{R}\impl)]{\hseqOne\ |\ \seqOne \vdash \forml_0 \impl 
	\forml_1, 
	\seqTwo}{\hseqOne\ |\ \seqOne \vdash \seqTwo\ \ \ \hseqOne\ |\ \seqOne, 
	\forml_0 \vdash \forml_1, \seqTwo}
\end{gather*}}
\caption{Hypersequent calculus for \product{} logic where $\hseqOne$ and 
$\hseqTwo$ are hypersequents, $\seqOne,\ \seqThree,\ \seqTwo,\ \seqFour$ 
are sequents, and $\forml_0,\ \forml_1$ are formulas.}
\label{fig:product-seq-calc}
\end{figure}

Figures~\ref{fig:godel-seq-calc}--\ref{fig:product-seq-calc} present three 
known sequent calculi covering the minimal fragments of the \godel{}, 
\product{},
and \Luka{} logics. As is standard in formalisation of sequent calculi, we 
define each of these calculi as an inductive type.
For example, Fig.~\ref{fig:luka-seq-calc-rocq} shows the \rocq{} definition 
of the \Luka{} logic defined in Fig.~\ref{fig:luka-seq-calc}. 

\begin{figure}[t]
\begin{minted}[numbers=left,xleftmargin=2em,escapeinside=88]{ssr}
Let formula := @expr R boolT_fuzzy.
Let hypersequent := seq (seq formula * seq formula).

Implicit Type Q P S : hypersequent. 8\label{line:implicithyper}8
Implicit Type A B C D X Y : seq formula. 8\label{line:implicitseq}8

Inductive seq_calc_luka_impl : hypersequent -> Prop :=
...
| split_l : forall Q A B C D,
    seq_calc_luka_impl (((A ++ B) |- (C ++ D)) :: Q) ->
    seq_calc_luka_impl ((A |- C) ::  (B |- D) :: Q)
| mix_l : forall Q A B C D,
    seq_calc_luka_impl ((A |- C) :: Q) ->
    seq_calc_luka_impl ((B |- D) :: Q) ->
    seq_calc_luka_impl ((A ++ B |- C ++ D) :: Q)
...
| bot_l : forall Q A (b : formula),
    seq_calc_luka_impl (((dl_bool neg_def _ _ _ false :: A)
                         |- [:: b]) :: Q)
| implL_l : forall Q A B (a b : formula),
    seq_calc_luka_impl (((b :: B) |- a:: A) :: Q ) ->
    seq_calc_luka_impl ((((a `=> b) :: B) |- A) :: Q)
| implR_l : forall Q A B (a b : formula),
    seq_calc_luka_impl ((A |- B) :: Q ) ->
    seq_calc_luka_impl  ((a :: A |- b :: B) :: Q)  ->
    seq_calc_luka_impl ((A |- (a `=> b) :: B) :: Q )
...
\end{minted}
\caption{Fragment of the implementation of the hypersequent calculus for 
\Luka{} logic (Fig.~\ref{fig:luka-seq-calc}). For the remainder of the rules 
see~\cite{github}. We 
use implicit typing for readability (lines 
\ref{line:implicithyper}--\ref{line:implicitseq}).}
\label{fig:luka-seq-calc-rocq}
\end{figure}

The distinct features of these fuzzy logic calculi come from their 
unusual structural rules. Firstly, handling the hypersequents necessitates 
introduction of structural rules for manipulation of hypersequents, the 
so-called \newterm{external structural rules} \ew, \ec{}, and \eex{} for hypersequent 
weakening, contraction, and exchange, respectively. 
Note that internal structural rules differ across different 
calculi: the \godel{} logic has both internal weakening and contraction, while the \product{} and \Luka{} only have weakening. 
All these logics have the internal exchange rule.

The structural rules called 
communication \com{}, split \spl{}, and mix \mix{} are introduced for the 
purposes of completeness, as they enable proofs of some formulas that are of 
special interest for real-valued logics. For example, the pre-linearity property
$$(p_0 \impl p_1) \lor (p_1 \impl p_0)$$
is supposed to reflect the fact that the real line forms a total order, in 
which for any two given elements, one can prove that one of them is greater 
than the other. In absence of internal weakening or contraction, this property 
usually requires some combination of the structural rules \com, 
\spl{}, and \mix{}, depending on the choice of the fuzzy 
logic. 

\paragraph{Formalisation.} When formalising in \rocq{}, we encode sequents as pairs of lists of formulas 
(using the polymorphic 
lists of type \coqin{seq})
 and hypersequents as lists of 
sequents.
For readability we use \coqin{Q |- P} to denote entailment. Each 
calculus is encoded as a separate inductive type. In the interest of 
conciseness we use a short-term \coqin{formula} for the type of logical formulas 
(e.g., \coqin{formula := @expr R boolT_fuzzy} for fuzzy \DL{}s) as well as 
\coqin{hypersequent := seq (seq formula * seq formula)} for hypersequents.

Note that we initially considered defining sequents and hypersequents in terms of 
multisets, 
as this would have allowed us to omit exchange rules.
However, due to the nesting of the structures needed in hypersequents and the 
richer library support for polymorphic lists, lists tend to 
result in shorter proofs. 

\subsection{Soundness of fuzzy \DL{}s}
\label{sec:soundness-fuzzy}

To define soundness, the semantics of $\tdl{\seqOne_i \vdash 
\seqTwo_i}$ need to be defined for each \DL{} in terms of their connectives.
While \godel{} logic has both monoidal and lattice connectives it is 
usually defined in terms of its lattice connectives (which can be equivalently 
replaced by the multiplicative analogues). 
However, the \product{} and \Luka{} logics are defined in terms of their 
multiplicative fragments. 

\begin{definition}[Sequent soundness]
\label{def:interpr-sequent}
A sequent $\seqOne \vdash \seqTwo = q_0, ..., q_n \vdash p_0, ..., p_m$, 
where 
$n, m \in \Nat$, is 
sound with respect to a given \DL{}, if the following holds for the 
appropriate 
\DL{}:
\begin{itemize}
\item $\tGodel{\seqOne \vdash \seqTwo} \mydef \tGodel{q_0 \wedge ... \wedge 
q_n} \leq \tGodel{p_0 \vee ... \vee p_m}$,
\item $\tlukasiewicz{\seqOne \vdash \seqTwo} \mydef \tlukasiewicz{q_0 \odot 
... \odot q_n} \leq \tlukasiewicz{p_0 \odot ... \odot 
p_m}$,
\item $\tproduct{\seqOne \vdash \seqTwo} \mydef \tproduct{q_0 \odot ... \odot 
q_n} \leq \tproduct{p_0 \odot ... \odot p_m}$,
\end{itemize}
and where, if either $P$ or $Q$ is empty, we have:
\begin{itemize}
\item $\tGodel{\vdash \seqTwo} \mydef 1 \leq \tGodel{p_0 \vee ... \vee p_m}$,
\item $\tGodel{\seqOne \vdash } \mydef \tGodel{q_0 \wedge ... \wedge 
q_n} \leq 0$,
\item $\tlukasiewicz{\vdash \seqTwo} \mydef 1 \leq \tlukasiewicz{p_0 \odot ... 
\odot 
p_m}$,
\item $\tlukasiewicz{\seqOne \vdash } \mydef \tlukasiewicz{q_0 \odot 
... \odot q_n} \leq 1 $, 
\item $\tproduct{\vdash \seqTwo} \mydef 1 \leq \tproduct{p_0 \odot ... \odot 
p_m}$,
\item $\tproduct{\seqOne \vdash} \mydef \tproduct{q_0 \odot ... \odot 
q_n} \leq 1$.
\end{itemize}
\end{definition}

Note that the sequent soundness for the \Godel{} logic follows the
traditional approach of the sequent calculus, that interprets
the left side of the sequent conjunctively 
and the right side of the sequent disjunctively.
However, \Luka{} and \product{} logics depart
from this tradition and interpret both sides using the same monoidal operator~\cite{fuzzy-proof}.


\begin{definition}[Soundness of a hypersequent calculus]
\label{def:soundness}
A hypersequent calculus for a $\DL$ is sound, if for every hypersequent 
$\hseqOne = \seqOne_0 \vdash \seqTwo_0\ |\ \ldots\ |\ \seqOne_n \vdash \seqTwo_n$
proven in this calculus there exists $i \in {1,...,n}$ such that 
$\tdl{\seqOne_i \vdash \seqTwo_i}$ for DL $\in \{\Godel{}, \Luka{}, \product\}$.
\end{definition}


We start with stating the soundness results that we prove in this section.
\begin{theorem}[Soundness of fuzzy logics]
The fuzzy logics \Godel{}, \product{} and \Luka{} are sound, in the sense of 
Definition~\ref{def:soundness}.
\end{theorem}
\begin{proof} For each logic, the proof proceeds by structural induction and uses the 
	properties of the interpretation functions. A similar proof following the same 
	structure for \DLtwo{} is shown in more detail in Sect.~\ref{sec:sound-dl2}.
\end{proof}

We will devote the rest of this section to explaining the formalisation of the 
calculi in \rocq.
Definitions~\ref{def:interpr-sequent}  and \ref{def:soundness} can be 
encoded directly. We illustrate this with the 
implementation of soundness of \Luka{} logic as an example. For ease of 
reading we use notation such as \coqin{eval_luka} which corresponds to the 
interpretation of sequents for \Luka{} logic (which uses $\odot$ on both sides, Definition~\ref{def:interpr-sequent}). Definition~\ref{def:soundness} is 
then:

\begin{minted}{ssr}
Lemma sound_luka_impl Q : seq_calc_luka_impl Q ->
  exists2 q : seq formula * seq formula,
  q \in Q & eval_luka q.1 <= eval_luka q.2.
\end{minted}
Recall that \coqin{Q} is a list of pairs; we use \coqin{.1} and \coqin{.2} to 
access the first and second element of the pair respectively. We use \coqin{exists2}, a version of the existential quantifier that asserts the existence a variable that satisfies two separate conditions.

The proofs are highly dependent on the semantics of the \DL{}, especially the 
lattice or monoidal connectives used in the soundness definition. At the most 
basic
level, most of the cases can be reduced to a set of inequalities---therefore, 
while they 
require significant simplification of the assumptions and proofs of 
intermediate steps, we can make extensive use of the \coqin{lra} and 
\coqin{nra} 
tactics, decision 
procedures for linear and non-linear arithmetic respectively~\cite{algebratactics}. Those simplifications include intermediate 
lemmas but also removal of 
unnecessary knowledge from the context and occasional renaming of complex 
formulas (as both of those can cause said tactics to fail).
In the case of the \product{} logic especially it is often necessary to prove 
more of 
the intermediate steps manually, as they can be too complex for the 
\coqin{nra} tactic.

The proofs for \Godel{} logic rely more on clever lemma application order (the 
presence of \coqin{minr}/\coqin{maxr} necessitates either defining helper 
lemmas or unfolding each occurrence). 
Finding the pattern of proof common among general cases significantly limits 
code duplication. In 
some cases, 
many nested occurrences of \coqin{minr}/\coqin{maxr} can increase 
computation time needed to process the combination. 
Therefore it is preferable to both make full use of targeted unfolding of only 
specific occurrences (\coqin{rewrite {n}/minr} where \coqin{n} is the number of 
the occurrence to unfold) and develop additional lemmas where 
possible---especially lemmas on the interactions between \coqin{minr}, 
\coqin{maxr} and big operations.

Lastly, the majority of the cases can be proven directly, with the exception of 
the rules 
for conjunction for \godel{} logic, which have to be proven by contradiction 
instead.

\paragraph{Extension of the calculi to \DL{} syntax with derivable rules}
\label{sec:sound-equiv}

The calculi presented for the fuzzy \DL{}s so far were for
minimal subsets of the language 
(Figures~\ref{fig:godel-seq-calc}--\ref{fig:product-seq-calc}), not the 
entirety of the \DL{} syntax. We, in fact, prove soundness both of
those standardly used calculi, as well as the calculi for the full set
of connectives used for \DL{}s (following the syntax of 
Fig.~\ref{fig:syntax-types-math}), with the 
exception of $\oplus$ for \product{} logic, because since it is not definable 
in 
terms of other connectives its rule cannot be derived.

We do not formalise completeness of the extended calculi in this work, but we 
compensate by showing that the rules in the extended calculi that include
the full set of connectives can be derived from the calculi for the minimal 
fragments (Figures~\ref{fig:godel-seq-calc}--\ref{fig:product-seq-calc}) which 
are known to be complete.

As an interesting observation, the derivability sometimes requires an addition 
of 
extended (yet still derivable) logical rules to the calculus. For example, to 
derive the sequent 
rules for monoidal conjunction for \Luka{} logic, an ``extended'' left 
implication rule, $(L\impl)_e$, is necessary:
$$
\infer[(L\impl)_e]{\hseqOne\ |\ \seqOne, \forml_0 \impl \forml_1 \vdash 
\seqTwo}{\hseqOne\ |\ \seqOne \vdash \seqTwo\ |\ \hseqOne\ |\ \seqOne, \forml_1 
\vdash 
\forml_0, \seqTwo}.
$$
This rule can be derived using the standard implication rule $(L\impl)$, 
weakening $(W)$ and external exchange $(EC)$. We note it here explicitly as the 
use of such alternate rules in derivation is often not stated in the original 
paper proofs~\cite{fuzzy-proof}.

\subsection{Hypersequent calculus for \DLtwo{}}
\label{sec:sound-dl2}

The main challenge of designing the calculus for \DLtwo{}~\cite{fischer2019dl2} 
lies in making the 
adjustments 
to its original definition and fitting it into the residuated lattice 
framework (Table~\ref{tab:semantics-math}).
 Recall that 
its 
original proposed interpretation of $\oplus$ (multiplication) is not a dual to the 
interpretation of $\odot$ and 
furthermore, that due to lack of the structural negation (as it is interpreted 
at 
the level of comparisons between reals) there is no simple way to obtain a 
dual (Sect.~\ref{sec:algebraic}). Therefore, it is not included in the 
calculus. 
Secondly, its implication was also only defined in terms of negation, 
and therefore was similarly non-structural and
only expressible with a 
combination of logical connectives and comparisons between real numbers instead 
of being a separate connective. 
In order to include it within the calculus, a new 
implication (Table~\ref{tab:semantics-math}) satisfying the residuation rule 
was added.

To prove the soundness of \DLtwo{} we first propose a new calculus
(Fig.~\ref{fig:dl2-calc}).
Unusually for a calculus, it does not have a rule for $\bot$, because $\bot$
is not defined for \DLtwo{}. If $\bot$ was clearly
defined, negation would be trivial to obtain.

\begin{figure}
		\centering
                \footnotesize{

\begin{gather*}
\infer[\init]{\hseqOne\ |\ \forml_0 \vdash \forml_0}{}
\semSpace
\infer[\emp]{\hseqOne\ |\  \vdash }{}
\semSpace
\infer[(\top)]{\hseqOne\ |\ \seqOne \vdash 
	\top}{}
	\\
	\\
\infer[\ew]{\hseqOne\ |\ \hseqTwo}{\hseqOne}
\semSpace
\infer[\ec]{\hseqOne\ |\ \hseqTwo}{\hseqOne\ |\ \hseqTwo\ |\ \hseqTwo}
\semSpace
\infer[\eex]{\hseqOne_1\ |\ \hseqTwo_2\ |\ \hseqTwo_1\ |\ 
\hseqOne_2}{\hseqOne_1\ |\ \hseqTwo_1\ |\ \hseqTwo_2\ |\ \hseqOne_2}
\\
\\
\infer[\weak]{\hseqOne\ |\ \seqOne,\ \seqThree \vdash\seqTwo}{
	\hseqOne\ |\ \seqOne \vdash \seqTwo}
\semSpace
\infer[\com]{\hseqOne\ |\ \seqOne_1, \seqOne_2 \vdash \seqTwo\ |\ 
\seqThree_1, \seqThree_2 \vdash \seqFour}{\hseqOne\ |\ \seqOne_1, \seqThree_1 
\vdash \seqTwo\ \ \ \ \hseqOne|\ \seqOne_2, \seqThree_2 \vdash 
 \seqFour}
\\ \\
 \infer[\lex]{\hseqOne\ |\ \seqOne_1, \forml_0, \forml_1, \seqOne_2 \vdash 
 \seqTwo}{
 \hseqOne\ |\ \seqOne_1, \forml_1, \forml_0, \seqOne_2 \vdash \seqTwo}
 \semSpace
 \infer[\rex]{\hseqOne\ |\ \seqOne \vdash 
 \seqTwo_1, \forml_0, \forml_1, \seqTwo_2 }{
 \hseqOne\ |\ \seqOne \vdash 
 \seqTwo_1, \forml_1, \forml_0, \seqTwo_2}
\\ \\
	\infer[(\text{L}\odot)]{\hseqOne\ 
		|\ \seqOne, \forml_0 \odot \forml_1\vdash \seqTwo}{
		\hseqOne\ |\ \seqOne, \forml_0, \forml_1 \vdash \seqTwo}
\semSpace
	\infer[(\text{R}\odot)]{\hseqOne\ |\ \seqOne_0,  \seqOne_1,
	\vdash \forml_0 \odot \forml_1,  
	\seqTwo_0, \seqTwo_1}{\hseqOne\ |\ \seqOne_0 \vdash \forml_0, \seqTwo_0 \ \ 
	\ 
	\hseqOne\ |\ \seqOne_1 \vdash \forml_1, \seqTwo_1}
	\\
	\\
	\infer[(\text{L}\wedge)]{\hseqOne\ |\ \seqOne, \forml_0 \wedge \forml_1 
	\vdash 
	\seqTwo}{\hseqOne\ |\ \seqOne, \forml_0 \vdash \seqTwo\  |\ 
	\seqOne, \forml_1 \vdash \seqTwo}
	\semSpace
	\infer[(\text{R}\wedge)]{\hseqOne\ 
	|\ \seqOne \vdash \forml_0 \wedge \forml_1, \seqTwo}{
	\hseqOne\ |\ \seqOne \vdash \forml_0, \seqTwo\ \ \ \ \hseqOne\ |\ \seqOne 
	\vdash \forml_1, \seqTwo}
	\\
	\\
	\infer[(\text{L}\vee)]{\hseqOne\ |\ \seqOne, \forml_0 \vee \forml_1 \vdash 
	\seqTwo}{\hseqOne\ |\ \seqOne, 
	\forml_0 \vdash \seqTwo \ \  \hseqOne\ |\ \seqOne, \forml_1 \vdash \seqTwo}
	\semSpace
	\infer[(\text{R}\vee)]{\hseqOne\ |\ 
		\seqOne \vdash 
		\forml_0 \vee \forml_1, \seqTwo}{\hseqOne\ |\ \seqOne \vdash \forml_0, 
		\seqTwo\ |\ 
		\seqOne \vdash \forml_1, \seqTwo}
	\\
	\\
	\infer[(\text{L}\impl)]{\hseqOne\ |\ \seqOne, \forml_0 \impl \forml_1 
	\vdash \seqTwo}{
	\hseqOne\ |\ \seqOne \vdash \seqTwo \ \ \ 
	\hseqOne\ |\ \seqOne, \forml_1 \vdash \forml_0, \seqTwo}
	\semSpace
	\infer[(\text{R}\impl)]{\hseqOne\ |\ \seqOne \vdash \forml_0 \impl 
	\forml_1, 
	\seqTwo}{\hseqOne\ |\ \seqOne \vdash \seqTwo \ \ \ 
		\hseqOne\ |\ \seqOne, \forml_0 \vdash \forml_1, \seqTwo}
\end{gather*}}
\caption{Hypersequent calculus for \DLtwo{} logic where $\hseqOne$ and 
$\hseqTwo$ 
are hypersequents, $\seqOne,\ \seqThree,\ \seqTwo,\ \seqFour$ 
are sequents and $\forml_0,\ \forml_1$ are formulas.}
\label{fig:dl2-calc}
\end{figure}

Recall from Table~\ref{tab:semantics-math} that \DLtwo{} has monoidal 
connectives. Therefore sequent soundness is defined as follows:

\begin{definition}[Sequent soundness for \DLtwo{}]
A sequent $\seqOne \vdash \seqTwo = q_1, ..., q_n \vdash p_1, ..., p_m$, 
where 
$n, m \in \Nat\setminus\{0\}$, is 
sound with respect to \DLtwo, if the following holds:
$$\tDLtwo{\seqOne \vdash \seqTwo} \mydef \tDLtwo{q_1 \odot 
... \odot q_n} \leq \tDLtwo{p_1 \odot ... \odot 
p_m}.$$
and where, if either $P$ or $Q$ is empty, we have:
\begin{align*}
\tDLtwo{\seqOne \vdash } &\mydef \tDLtwo{q_1 \odot ... \odot q_n} \leq 0, \\
\tDLtwo{\vdash \seqTwo} &\mydef 0 \leq \tDLtwo{p_1 \odot ... \odot p_m}.
\end{align*}
\end{definition}

Let us introduce an abbreviated notation for clarity: given 
$\tDLtwo{p_1 \odot 
... \odot p_n}$ where $p_1 \odot 
... \odot p_n = \seqTwo$ we write $\tDLtwo{\seqTwo}$.

\begin{theorem}[Soundness of \DLtwo{}]
\DLtwo{} is sound, in the sense of Definition~\ref{def:soundness}.
\end{theorem}

\begin{proof} 

The proof proceeds by structural induction.
Starting with the initial sequents---\init, \emp{}, and ($\top$)---which are 
straightforward, and then proceeding to the remaining rules (see~\cite[\coqin{seq_calc_dl2.v}]{github}). 

We do so by proving each individual rule to be sound by proving the bottom 
hypersequent is sound using the assumptions of the upper hypersequent: a rule 
$\infer{\hseqTwo}{\hseqOne_1 \ \ ...\ \ \hseqOne_n}$ is sound if whenever 
$\hseqOne 
,...,\hseqOne_n$ are sound $\hseqTwo$ is also sound.
 We show a selection of cases as an example (the remainder of the proof 
 proceeds analogously), starting with one of the initial sequents: \init.

\medskip
\textbf{\init} To prove  $\hseqOne\ |\ p \vdash p$ 
is sound, we observe, by evaluation, that $\tDLtwo{p \vdash p} = \tDLtwo{p} \leq \tDLtwo{p}$, which holds trivially. 
\medskip
%

\medskip
Now consider one of the structural rules of the calculus, \com:
\medskip



\textbf{\com} Recall that $|$ in hypersequents can be seen as a disjunction 
while the space between two separate hypersequents in the top of a rule can be 
seen as a conjunction. Therefore, to prove the soundness of the \com{} rule

$$\infer[\com]{\hseqOne\ |\ \seqOne_1, \seqOne_2 \vdash \seqTwo\ |\ 
\seqThree_1, \seqThree_2 \vdash \seqFour}{\hseqOne\ |\ \seqOne_1, \seqThree_1 
\vdash \seqTwo\ \ \ \ \hseqOne|\ \seqOne_2, \seqThree_2 \vdash 
 \seqFour}$$
we need to prove that there exists a sequent 
which is sound such that it either:
\begin{enumerate}
\item belongs to $\hseqOne$,
\item is equal to $\seqOne_1, \seqOne_2 \vdash \seqTwo$, or
\item is equal to $ 
\seqThree_1, 
\seqThree_2 \vdash \seqFour$.
\end{enumerate}
 From the assumptions we know that there exists a sound sequent $A$ such that $A \in \hseqOne \cup \{\seqOne_1, 
 \seqThree_1 \vdash \seqTwo\}$ and a sound sequent $B$ such that $B \in \hseqOne \cup \{\seqOne_2, 
 \seqThree_2 \vdash \seqFour\}$ .
%
If either $A$ or $B$ are in $\hseqOne$ then this directly concludes the proof (case 1. holds).

\smallskip

That leaves us with $A = \seqOne_1, 
\seqThree_1 \vdash \seqTwo$ and $B = \seqOne_2, 
\seqThree_2 \vdash \seqFour$. Consider the two cases:
\begin{align*}
\tDLtwo{\seqOne_2} &\leq \tDLtwo{\seqThree_1} \text{, and} \\
\tDLtwo{\seqOne_2} &> \tDLtwo{\seqThree_1}.
\end{align*}

In the first case, we add $ \tDLtwo{\seqOne_1}$ on both sides of that 
inequality, getting 

$$\tDLtwo{\seqOne_1} 
+ \tDLtwo{\seqOne_2} \leq \tDLtwo{\seqOne_1} + \tDLtwo{\seqThree_1}.$$

Note that
$$\tDLtwo{A} = \tDLtwo{\seqOne_1, 
\seqThree_1 \vdash \seqTwo} = \tDLtwo{\seqOne_1} + \tDLtwo{\seqThree_1} \leq 
\tDLtwo{\seqTwo}.$$

  Combining the two inequalities, we get $\tDLtwo{\seqOne_1} 
+ \tDLtwo{\seqOne_2} \leq \tDLtwo{\seqTwo}$ and from that we know that 
$\seqOne_1, \seqOne_2 
\vdash \seqTwo$ is sound, which concludes the proof (case 2. holds).

\smallskip

The second case proceeds analogously, proving that $\seqThree_1, 
\seqThree_2 \vdash \seqFour$ is sound (case 3. holds).

\medskip

Consider now the rules for logical connectives. We will provide in detail the proof for one of the cases---namely  (L$\odot$)---as the proofs for 
right rules are 
otherwise similar.

\medskip

\textbf{(L$\odot$)}  To prove this case, we need to prove that there 
exists a sequent which either:
\begin{enumerate}
\item belongs to $\hseqOne$, or
\item is equal to $\seqOne, p_0 \odot p_1 \vdash \seqTwo.$
\end{enumerate}
%
 We know that there exists a sound sequent $A$, such that $A \in \hseqOne \cup \{\seqOne, p_0, p_1 \vdash \seqTwo\}$
%

\smallskip

If $ A \in 
\hseqOne$, the proof concludes (case 1. holds). 

\smallskip

Otherwise we have $A = \seqOne, p_0, p_1 \vdash \seqTwo$. 
In this case, we need to prove $\seqOne, p_0 \odot p_1 
\vdash \seqTwo$ is sound; to do so, we need $\tDLtwo{\seqOne, p_0 \odot p_1 
\vdash \seqTwo}$ to be true. We therefore need to prove
$$\tDLtwo{\seqOne} + \tDLtwo{p_0} + 
\tDLtwo{p_1} \leq \tDLtwo{\seqTwo} .$$

Since $A$ is sound, 
we have
 $$\tDLtwo{\seqOne, p_0, p_1 \vdash \seqTwo} = \tDLtwo{\seqOne} + 
\tDLtwo{p_0} + \tDLtwo{p_1} \vdash \tDLtwo{\seqTwo}$$
which concludes the 
proof (case 2. holds).

\medskip

This concludes the proof. All other cases are proven analogously.
\end{proof}

Due to the simplicity of the $\odot$ semantics, the proof of 
soundness of \DLtwo{} itself is easier from the formalisation standpoint than 
the 
remaining fuzzy \DL{}s.
 Both the sequent rules and the soundness statement itself gain complexity with 
 the semantics of the relevant connectives (particularly the monoidal and 
 lattice connectives used to define soundness).

\subsection{Hypersequent calculus for \STLinfty{}}
\label{sec:sound-stl}

From initial inspection, \STLinfty{} is notably similar to \Godel{} 
logic through the shared semantics of both lattice and monoidal connectives 
(Table~\ref{tab:semantics-math}). Hence, we define its soundness of 
sequents analogically:

\begin{definition}[Sequent soundness for \STLinfty{}]
A sequent $\seqOne \vdash \seqTwo = q_1, ..., q_n \vdash p_1, ..., p_m$, 
where 
$n, m \in \Nat\setminus\{0\}$, is 
sound with respect to \STLinfty, if the following holds:
$$
\tSTLinfty{\seqOne \vdash \seqTwo} \mydef \tSTLinfty{q_1 \wedge ... \wedge q_n} 
\leq \tSTLinfty{p_1 \vee ... \vee p_m},
$$
and where, if either $P$ or $Q$ is empty, we have:
\begin{align*}
\tSTLinfty{\seqOne \vdash } &:\mydef \tSTLinfty{q_1 \wedge ... \wedge q_n} \leq 
-\infty, \\
\tSTLinfty{\vdash \seqTwo} &\mydef +\infty \leq \tSTLinfty{p_1 \vee ... \vee p_m}.
\end{align*}
\end{definition}

It is therefore natural that the calculus of \STLinfty{}, which we define in 
Fig.~\ref{fig:stl-seq-calc}, also shares the majority of its rules with 
\Godel{} logic, including the implication. 

The only difference is negation, 
which is 
not analogous to the binarised approached of \Godel{} logic (the sequent rules 
for which are derivable from the set defined in Fig.~\ref{fig:godel-seq-calc}). 
Standard negation 
rules are not sound for \STLinfty{}. 
An 
alternative 
approach would be to take negation satisfying axiom N1 in Table~\ref{table:axiom-schema-neg}, where $\tSTLinfty{\neg e} = \tSTLinfty{e \impl \bot}$ 
instead, with similarly derived sequent rules; this merits future 
investigation. 

We leave the problem of defining either custom sequent rules (sound and 
complete) for negation or a different negation for future work.

\begin{figure}[t]
		\centering
        \footnotesize{
\begin{gather*}
\infer[\init]{\hseqOne\ |\ \forml_0 \vdash \forml_0}{}
\semSpace
\infer[(\bot)]{\hseqOne\ |\ \seqOne, \bot \vdash \seqTwo}{}
\semSpace
\infer[(\top)]{\hseqOne\ |\ \seqOne \vdash 
\top}{}
\\
\\
\infer[\ew]{\hseqOne\ |\ \hseqTwo}{\hseqOne}
\semSpace
\infer[\ec]{\hseqOne\ |\ \hseqTwo}{\hseqOne | \hseqTwo | \hseqTwo}
\semSpace
\infer[\eex]{\hseqOne_1\ |\ \hseqTwo_2\ |\ \hseqTwo_1\ |\ 
\hseqOne_2}{\hseqOne_1\ |\ \hseqTwo_1\ |\ \hseqTwo_2\ |\ \hseqOne_2}
\\ \\
\infer[\com]{\hseqOne\ |\ \seqOne_1, \seqOne_2 \vdash \seqTwo\ |\ 
\seqThree_1, \seqThree_2 \vdash \seqFour}{\hseqOne\ |\ \seqOne_1, \seqThree_1 
\vdash \seqTwo\ \ \ \  \hseqOne\ |\ \seqOne_2, \seqThree_2 \vdash 
 \seqFour}
\\ \\
\infer[\weak]{\hseqOne\ |\ \seqOne,\ \seqThree \vdash\seqTwo}{
	\hseqOne\ |\ \seqOne \vdash \seqTwo}
\semSpace
\infer[\contr]{\hseqOne\ |\ \seqOne,\ \seqThree \vdash\seqTwo}{
	\hseqOne\ |\ \seqOne,\seqThree,\seqThree \vdash\seqTwo}
\\ 
\\
\infer[\lex]{\hseqOne\ |\ \seqOne_1, \forml_0, \forml_1, \seqOne_2 \vdash 
\seqTwo}{
\hseqOne\ |\ \seqOne_1, \forml_1, \forml_0, \seqOne_2 \vdash \seqTwo}
\semSpace
\infer[\rex]{\hseqOne\ |\ \seqOne \vdash 
\seqTwo_1, \forml_0, \forml_1, \seqTwo_2 }{
\hseqOne\ |\ \seqOne \vdash 
\seqTwo_1, \forml_1, \forml_0, \seqTwo_2}
	\\
	\\
	\infer[(\text{L}\wedge)]{\hseqOne\ |\ \seqOne, \forml_0 \wedge \forml_1 
	\vdash 
	\seqTwo}{\hseqOne\ |\ \seqOne, \forml_0 \vdash \seqTwo\  |\ 
	\seqOne, \forml_1 \vdash \seqTwo}
	\semSpace
	\infer[(\text{R}\wedge)]{\hseqOne\ 
	|\ \seqOne \vdash \forml_0 \wedge \forml_1}{
	\hseqOne\ |\ \seqOne \vdash \forml_0\ \ \ \ \hseqOne\ |\ \seqOne 
	\vdash \forml_1}
	\\
	\\
	\infer[(\text{L}\vee)]{\hseqOne\ |\ \seqOne, \forml_0 \vee \forml_1 \vdash 
	\seqTwo}{\hseqOne\ |\ \seqOne, 
	\forml_0 \vdash \seqTwo \ \  \hseqOne\ |\ \seqOne, \forml_1 \vdash \seqTwo}
	\semSpace
	\infer[(\text{R}\vee)]{\hseqOne\ |\ 
		\seqOne \vdash 
		\forml_0 \vee \forml_1}{\hseqOne\ |\ \seqOne \vdash \forml_0\ |\ 
		\seqOne \vdash \forml_1}
	\\
	\\
	\infer[(\text{L}\impl)]{\hseqOne\ |\ \seqOne, \forml_0 \impl \forml_1 
	\vdash 
	\seqTwo}{\hseqOne\ |\ \seqOne \vdash \forml_0 \ \ \hseqOne\ |\ 
	\seqOne, \forml_1 \vdash 
	\seqTwo}
	\semSpace
	\infer[(\text{R}\impl)]{\hseqOne\ |\ \seqOne \vdash \forml_0 \impl 
	\forml_1}{\hseqOne\ |\ \seqOne, 
	\forml_0 \vdash \forml_1}
\end{gather*}}
\caption{Hypersequent calculus for \STLinfty{} logic without negation where 
$\hseqOne$ and 
$\hseqTwo$ 
are hypersequents, $\seqOne,\ \seqThree,\ \seqTwo,\ \seqFour$ 
are sequents and $\forml_0,\ \forml_1$ are formulas.}
\label{fig:stl-seq-calc}
\end{figure}

\begin{theorem}[Soundness of \STLinfty{}]
\STLinfty{} is sound, in the sense of Definition~\ref{def:soundness}.
\end{theorem}
\begin{proof}
The proof structure is analogous to the proof for soundness of \DLtwo{} (for 
the formalised proof see~\cite[\coqin{seq_calc_stli.v}]{github}).
\end{proof}

From the perspective of formalisation in \rocq{} this proof has posed a 
different type of technical challenge. While, mathematically, the proof is very 
similar to the one of \Godel{} logic, in the majority of the cases, it instead 
uses 
\coqin{\bar R} (the type of extended real numbers in \analysis{}) as its type.
This is due to the explicit inclusion of $-\infty$ 
and $\infty$ in the semantics of \STLinfty{}, which was necessary for proving 
it is a residuated lattice (Sect.~\ref{sec:algebraic}): specifically $\infty$ was 
needed 
as an identity element for its monoidal operator $\odot$.
Working with 
extended reals can provide an additional challenge due to a smaller library of 
available lemmas. 
In this particular case it also leads to the inability of using \coqin{lra} and 
\coqin{nra} tactics \cite{algebratactics}, which the \Godel{} formalisation 
relies heavily on, and 
instead solving every inequality set manually.

We do not propose a sequent calculus for full \STL{}, as we assume that,
 due to the lack of associativity of this DL, both algebraic and 
 proof-theoretic semantics will necessitate approaches that are rather 
 different to the ones we considered in this paper.
For future work, one could try to take inspiration  from the proof 
theory for non-associative logics, such as non-associative Lambek 
calculus~\cite{lambek1961calculus,moot2012non}, when studying the full \STL{}.
%


\subsection{Weak completeness}
\label{sec:weak-compl}

Recall the axiom schema arising from residuated lattices 
(Table~\ref{table:axiom-schema-lattice}), which we established for the semantics of 
\DL{}s in Sect.~\ref{sec:algebraic}. 
In the absence of a full completeness result, we can employ a collection of 
axioms to demonstrate weak completeness for the calculi. 
By ensuring that these axioms are provable within the target calculus---rather 
than treating them as properties of the semantics---we guarantee that 
desirable 
properties are preserved by the calculus. 
This approach can provide reasonable 
confidence that the system is well-behaved. 
Most importantly if one of the axioms is not provable, thus the calculus is not 
weakly complete, it gives sufficient 
evidence against completeness of that given sequent calculi.
This is therefore especially important for the calculi of \DLtwo{} and 
\STLinfty{}, for which we do not provide a completeness proof.

\begin{definition}[Weak Completeness]
A hypersequent calculus is weakly complete if relevant axioms of R1--R9 are 
provable.
\end{definition}

We omit specifically axiom R10, as residuation is a semantic property. We 
determine relevant axioms based on the properties of the \DL{} in question 
(Table~\ref{tab:properties}).

By using the derivation within the respective sequent calculi we show that 
\DLtwo{}~\cite[\coqin{seq_calc_dl2.v}]{github} and 
\STLinfty{}~\cite[\coqin{seq_calc_stli.v}]{github} are 
weakly 
complete.
%
%
%
%
%

%
We leave the problem of finding complete sequent calculus rules for negation of 
\STLinfty{} for future work.

%
%
%

\section{Example of applications in machine learning}
\label{sec:example}

We return to the motivation behind differentiable logics,
property-guided training for verification.  
For this, we use as an example $\epsilon$-$\delta$-robustness, 
a property often used in neural network verification~\cite{CasadioKDKKAR22}.

\subsection{Property-guided training, by means of an example}
\label{sec:ex}
\paragraph*{Neural network properties}
Given a neural network $N: \Real^m \to \Real^n$, the a property to be verified 
usually takes the form of a Hoare triple
$\forall \x \in \Real^m . \Pp(\x) \longrightarrow \Sp(\x)$,
 where $\Pp$ and~$\Sp$ can be arbitrary properties 
obtained by using variables $\x \in \Real^m$,
constants, vector, arithmetic operations, $\leq$, $=$, $\land$, $\lor$, and 
$\neg$. Additionally, 
$\Sp$ may contain the neural network $N$ as a function.
 
\begin{example}[\newterm{$\epsilon$-$\delta$-robustness}~\cite{CasadioKDKKAR22}]\label{ex:prop}
\label{example:properties}
Given a neural network $N$ and a vector~$\vect{v}$, consider the
specification that requires that for all inputs $\x$ that are within
an $\epsilon$~distance from~$\vect{v}$, the output of $N(\x)$ should
not deviate by more than $\delta$ from $N(\vect{v})$:
$$\forall \x. | \x - \vect{v} |_{L_\infty} \leq \epsilon \Rightarrow
    | N(\x) - N(\vect{v})|_{L_{\infty}} \leq \delta.$$ 
This property can be used to avoid misclassifying images when only a few pixels are 
perturbed. 
This particular example uses the \newterm{$L_{\infty}$ norm}:
$| \x - \y|_{L_{\infty}} \mydef \max_{i\in\{0,\ldots,n-1\}} ([\x]_i - [\y]_i)$,
where $[\x]_i$ stands for the $i$th element of~$\x$.

\end{example}

Unfortunately, as demonstrated by Fischer et
al.~\cite{fischer2019dl2}, even most accurate neural networks fail
even the most natural verification properties, such as
$\epsilon$-$\delta$-robustness. This motivated the search for better
ways to train the networks.

\paragraph*{Property-guided training} Methods for property-guided
training have received considerable attention in the AI literature, as the 
survey~\cite{ijcai2022p767} shows.
We will only illustrate the method that was suggested by Fischer et
al.~\cite{fischer2019dl2}, and refer the reader to the survey for more
examples.

\begin{example}[Generating a loss function from a logical 
property~\cite{fischer2019dl2}]   
Recall that standard supervised learning trains a neural network $N$ with 
trainable parameters $\theta$
to optimise the objective $ \min_{\theta} \lossfn(\x, \y)$,
for the loss function $\losssymbol: \Real^m \times \Real^n \rightarrow \Real$.
Generally, $\losssymbol$ measures the difference between 
the network's output and the given data for each input point. Examples of 
$\losssymbol$ are cross-entropy loss or mean squared error.
But now we want to train the neural network to satisfy any arbitrary property
$\forall \x. \Pp(\x) \longrightarrow \Sp(\x)$. For this,  
we replace the above optimisation objective  with
$$\min_{\theta} \left(\max_{\x \in \mathbb{H}_{\Pp(\x)}} 
\lossfn_{\Sp}(\x)\right)$$
where $ \mathbb{H}_{\Pp(\x)} \subseteq \Real^m$ refines the type $\Real^m$ to a 
subset for which the property $\Pp$ holds, 
  and $\lossfn_{\mathcal{S}}: \Real^m  \rightarrow \Real$ is obtained by 
  applying a suitable
\newterm{interpretation function} for $\Sp$.

We omit the exact details of how such optimisation algorithms are
defined: they are known and can be found in a suitable machine
learning tutorial, for example~\cite{KM18}.
Intuitively, the optimisation algorithm will search for $\x \in 
\mathbb{H}_{\Pp(\x)}$ such that $\x$ 
maximises the loss $\lossfn_{\Sp}(\x)$, 
in order to train the neural network parameters $\theta$ to minimise that loss. 
Concretely, if the property is $\epsilon$-$\delta$-robustness, it will look for 
the worst perturbation of $\vect{v}$
that violates the property, and will optimise the neural network to classify 
that bad example correctly. 
\end{example}

\paragraph*{Differentiable logics for loss functions}  In the above
example, we did not explain how to define the interpretation function
$\lossfn_{\Sp}$ for an arbitrary property $\Sp$; we need differentiable
logics for that purpose.

\begin{example}[Loss functions from properties in a fuzzy logic]\label{ex:loss}

Taking the property from Example~\ref{ex:prop}, by the Fischer et al.\ method
we must be able to interpret the right-hand side of the implication, i.e.,
$| N(\x) - N(\vect{v})|_{L_{\infty}} \leq \delta$, 
given concrete values for $\epsilon$, $\delta$, 
a concrete vector~$\vect{v}$, neural network~$N$, and a suitable definition of 
the 
$L_{\infty}$ norm.
For example, interpretation for our property in \STL{}~\cite{varnai} is given by
$
 \tempty{| N(\x) - N(\vect{v})|_{L_{\infty}} \leq \delta}_{\STL} = $
$\delta - |\tempty{N(\x)} - \tempty{N(\vect{v})}|_{L_{\infty}}$.
On the left-hand side, the $L_{\infty}$ distance between vectors as well as~$N$ 
are defined
in the syntax of \STL{}; on the right-hand side, they are given by real vector
arithmetic operations. Example~\ref{ex:ldlcoq} will make the relation between 
syntax and interpretation clear. 
The obtained function can be used directly for training neural networks.
\end{example}

\subsection{Formalisation of $\epsilon$-$\delta$-robustness}
\label{sec:ex-formalised}

We can formalise $\epsilon$-$\delta$-robustness once for all DLs:

\begin{example}[Formalisation of $\epsilon$-$\delta$-robustness]
\label{ex:loss-formalization}

We use the formal syntax of Fig.~\ref{fig:syntax-coq} to formalise 
Example~\ref{ex:loss}, i.e., $\epsilon$-$\delta$-robustness:
$$\forall \x. | \x - \vect{v} |_{L_\infty} \leq \epsilon \Rightarrow
    | N(\x) - N(\vect{v})|_{L_{\infty}} \leq \delta.$$
First, we embed into our language the $L_{\infty}$ norm operator and
vector subtraction, and declare appropriate notations:
\begin{minted}{ssr}
Let dl_norm_infty n : expr (funT n.+1 1) := dl_fun
  (fun t : R ^ n.+1 => [ffun x : 'I_1 => \big[maxr/t 0]_(i < n.+1) t i ])%R.
Let idx0 := @dl_idx R 1 ord0.
Local Notation "'`|' v '|'" := ((dl_norm_infty _ `@ v) `! idx0).

Let dl_vec_sub n :=
  dl_fun2 (fun (x y : R ^ n) => [ffun i => x i - y i]%R).
Local Notation "x `- y" := (dl_vec_sub _ `@2 (x, y)) (at level 42).
\end{minted}

We use notation \coqin{`@} for function application \coqin{dl_app} and 
\coqin{`@2} for \coqin{dl_app2}.

Then we can express the $\epsilon$-$\delta$-robustness property as follows:
\begin{minted}{ssr}
Context {n m : nat} (eps delta : expr realT) (f : expr (funT n.+1 m.+1))
  (v : expr (vectorT n.+1)) (x : expr (vectorT n.+1)).

Definition eps_delta_robust fn fm fl : expr (boolT fn impl_def fm fl) :=
  `| x `- v | `<= eps `=> `| (f `@ x) `- (f `@ v) | `<= delta.
\end{minted}

\end{example}

\begin{example}\label{ex:ldlcoq}
Taking the interpretation task of Example~\ref{ex:loss} and its formalisation
in Example~\ref{ex:loss-formalization}, we can call the interpretation function
for any of the logics defined in our framework to obtain the desired loss 
function
for the property, e.g., for \DLtwo{}:
\begin{minted}{ssr}
Let eps_delta_robust_dl2 :=
  ([[ eps_delta_robust neg_undef m_undef l_undef ]]_dl2).
Compute eps_delta_robust_dl2.
\end{minted}
\end{example}

\section{Related work}
\label{sec:relatedwork}

Neural network verification is a new field, and its nascent methods
need validation and further refinement. Better understanding,
formalisation and implementation of \DL{}s is often named as one of
the main challenges in the field~\cite{CordeiroDGIJKKLMSW25}.

In terms of programming language support for neural network verification, 
the CAISAR tool~\cite{girardsatabin2022caisar}, implemented as an OCaml DSL, 
puts emphasis on the smooth integration of a general specification language 
with many existing neural network solvers. 
However, CAISAR does not support property-guided training.  
The aspiration of languages like Vehicle~\cite{DaggittKAKS025} and CAISAR is to 
accommodate compilation
of specifications into both neural network solver and machine learning backends.
For the former, there is an on-going work on certifying the neural network 
solver backends~\cite{Des25,DaggittAKKA23,DaggittKAKS025}. 

On the side of machine-learning backends, \DL{}s have been previously
formalised in the Agda interactive theorem prover as part of the formalisation 
of the Vehicle language~\cite{ADK24}, but did not extend to
shadow-lifting---the part that requires extensive mathematical
libraries.  Property-guided training certified via theorem proving was
also proposed by Chevallier et al.~\cite{ChevallierWF22}.

Relevant work on formalisation of neural networks in ITPs includes:
verification of neural networks in Isabelle/HOL
\cite{brucker2023verifying} and Imandra~\cite{Des25},
formalisation of piecewise affine activation functions in
\rocq{}~\cite{aleksandrov2023formalizing}, providing formal guarantees of
the degree to which the trained neural network will generalise to new
data in \rocq~\cite{bagnall2019certifying}, convergence of a
single-layered perceptron in \rocq~\cite{murphy2017verified}, and
verification of neural archetypes in \rocq~\cite{de2022use}.  The
formalisation presented here does not directly formalise neural
networks.

Prompted by the increasing demand from AI, there has been a surge of 
research into mathematical foundations of quantitative logics, 
of which \DL{}s can be seen as a special case. 
Firstly, Bacci et al.~\cite{bacci2023mfps,abs-2402-03543} formalised a logic of 
Lawvere quantale,
over the domain $[0, \infty]$. The logic bears some resemblance to \DLtwo{} 
thanks to its domain, and also 
the authors have shown that the \Luka{} logic is an instance of the logics of 
Lawvere quantale.
The structure of Lawvere Quantale generalises the definition of residuated 
lattice, and also this logic 
has a proof-theoretic interpretation~\cite{bacci2023mfps,abs-2402-03543}. 
Bacci and M\"{o}gelberg~\cite{abs-2501-18275} further extended that logic with 
the induction principle,
and showed how it can be applied to encode Probabilistic Hoare logic and Markov 
Processes. 

At the same time, Capucci~\cite{abs-2406-04936} presented a logic that promises 
even a greater generality.
It also takes a Lawvere quantale as its interpretation, thus retaining the link 
with the tradition of residuated lattices.
But at the same time it defines multiplicative and additive connectives that 
yield
a sequent calculus similar to that of the Linear logic. 
Notably, its monoidal operators resemble those of the \product{} logic we have 
considered here. 
For the lattice/additive component, that logic offers to use soft additive 
connectives $\land$ and $\lor$.
These connectives have a softness parameter, that, at~$\infty$, turns them into 
$\min$ and $\max$; while at all other values
retains the shadow-lifting and smoothness properties. This resembles a lot the 
\STL{} logic intuition, with its use of parameter $\nu$.
Capucci's logic is under development, yet it promises to bring a conceptually 
elegant resolution
to the quest for harmonious co-existence of algebraic, analytic and 
proof-theoretic properties in one logic.  
We are looking into formalisation of that logic in 
\rocq~\cite{marulanda-giraldo2025our}.  

Different variants of the logics we use in this work---fuzzy logics in 
particular---can also be considered as alternative \DL{}s. The 
primary example being the choice of type of 
implication~\cite{van2022analyzing}. 
The 
implication required by the properties of residuated lattices and used through 
this paper are R-implications, but S-implications as well as Yager f- and 
g-implications are other well-known classes with different semantic 
properties~\cite{baczynski2007yager}. 
We have investigated some of the alternate semantics, in particular involving 
S-implications and S-negations, residuation does not always hold.
On the other hand, the extension of \DL{}s, particularly fuzzy logics, with 
comparisons between real numbers is reminiscent of 
SMT~\cite{barrett2018satisfiability}, as both involve adding domain-specific 
theories on top of propositional logic. 

A property found desirable for \DL{}s which we do not consider in this work is 
adequacy relative to Boolean logic (in ML works sometimes referred to as 
``soundness''~\cite{fischer2019dl2}). Proving a logic to be adequate shows it 
to be weaker than classical logic---a property commonly seen within the field 
of substructural logics which are shown to be stronger than intuitionistic 
logic and weaker than classical~\cite{galatos2007residuated}. While adequacy of 
fuzzy logics is well known, the introduction of comparisons between real 
numbers has a non-trivial effect on the result and is left as future work.

Lastly, we would be remiss not to mention the developments in the
application of \DL{}s in machine learning. In the literature, fuzzy
logics have been introduced alongside comparative
analyses of their performance in training
tasks~\cite{van2022analyzing}. Krieken et
al.~\cite{van2025neurosymbolic} have also loosely included them in a
broader scene of neurosymbolic methods in terms of their assumptions
of independence of different concepts. More recently, Flinkow et
al.~\cite{flinkow2025comparing} conducted a comprehensive evaluation
of various \DL{}s---including all logics considered in the present
work except \STL---employing a neural network verifier; although the
\DL{}s' semantics differ slightly, their evaluation provides
information about the performance of \DL{}s in training.

%

\section{Conclusions and future work}
\label{sec:conclusion}

We have presented a complete \rocq{} formalisation of a range of
existing \DL{}s, making the following three main contributions:
\begin{enumerate}
\item We provided a {\em uniform formalisation} of algebraic,
  analytic, and proof-theoretic results for \DL{}s that span a wide
  range of substructural logic and machine learning papers, published
  by different communities in the last 20 years.  We showed that
  unification of these three aspects is crucial for adequate
  understanding and future design of logics for machine
  learning. Moreover, we showed that such a unifying approach is
  possible.
\item We contribute to the \DL{} community by {\em revisiting
    semantics of \STL{} and \DLtwo\/} in a way more amenable to formal
  verification. We found and fixed errors in the literature.
  Indeed, our previous attempt~\cite{ldl} to do this with pen and
  paper proofs was drowned in low-level case analysis and resulted in
  some errors, see Section~\ref{sec:errors}. This complexity was our
  initial motivation to undertake the formalisation.
  We moreover presented a new sequent calculus for \DLtwo{} that could 
  be derived thanks to our general unifying approach
  to \DL{}s.
\item We proposed a {\em general formalisation strategy\/} to \DL{}s based on
  dependent types and formal mathematics. The proposed formalisation
  is built to be easily extendable for future studies of different
  \DL{}s.
\end{enumerate}

\begin{table}
\centering
\begin{tabular}{lll}
File & Contents & L.o.c. \\
\hline
\multicolumn{3}{|l|}{\it Additions to \mathcomp{} libraries} \\
\hline
{\tt mathcomp\us{}extra.v} & Lemmas for iterated $\min$/$\max$, etc. & 520 \\
{\tt analysis\us{}extra.v} & Lemmas for $\min$/$\max$ with \coqin{\bar R}, 
partial 
derivatives, etc. & 745 \\
\hline
\multicolumn{3}{|l|}{\it Generic logic and generic definitions of properties } 
\\
\hline
{\tt dl.v} & \DL{} syntax and semantics(\S~\ref{sec:syntax}), shadow-lifting 
(Sect.~\ref{sec:shadow-lifting}) & 681 \\
\hline
\multicolumn{3}{|l|}{{\it Algebraic and analytic properties of concrete 
logics}} \\
\hline
{\tt dl2.v} & \DLtwo{}: algebraic properties (\S~\ref{sec:algebraic}), analytic 
properties
(\S~\ref{sec:sldltwoproduct}) & 325  \\
{\tt fuzzy.v} & \Godel, \Luka, \Yager, \product:  & 1590 \\
& algebraic properties (\S~\ref{sec:algebraic}), analytic properties 
(\S~\ref{sec:sldltwoproduct}) & \\
{\tt stl.v} & \STL{}: algebraic properties (\S~\ref{sec:algebraic}), analytic 
properties
(\S~\ref{sec:slstl}) & 1206 \\
\hline
\multicolumn{3}{|l|}{{\it Formalisations of algebraic properties 
using extended reals }} \\
\hline
{\tt dl2\us{}ereal.v} & \DLtwo{}: algebraic properties (\S~\ref{sec:algebraic}) 
& 207\\
{\tt stl\us{}ereal.v} & \STL{}: algebraic properties (\S~\ref{sec:algebraic})& 
478\\
{\tt stl\us{}infty.v} & \STLinfty{}: algebraic properties 
(\S~\ref{sec:algebraic}) & 234\\
\hline
\multicolumn{3}{|l|}{{\it Proof-theoretic properties}} \\
\hline
{\tt seq\us{}calc.v} & \Godel, \Luka, \product: soundness 
(\S~\ref{sec:proof-theoretic}--~\ref{sec:soundness-fuzzy}) & 2354\\
{\tt seq\us{}calc\us{}dl2.v} & \DLtwo{}: soundness (\S~\ref{sec:sound-dl2}) & 
532\\
{\tt seq\us{}calc\us{}stli.v} & \STLinfty{}: soundness 
(\S~\ref{sec:sound-stl}) & 
491\\
\hline
& \multicolumn{1}{r}{Total} & 9372 \\
\end{tabular}
\caption{Overview of the formalisation~\cite{github}.}
\label{tab:overview}
\end{table}

Table~\ref{tab:overview} summarises the \rocq{} implementation. 
Both L'H\^opital's rule (now a part of 
\analysis~\cite[\coqin{realfun.v}]{analysis} \cite{lhopital}) and analytic 
properties, 
especially in complex cases such 
as \STL{}, form a substantial
part of the development.
The proofs for fuzzy \DL{}s (file \coqin{fuzzy.v}) are grouped together
thanks to their similarity and even share some of the proofs.
The files \coqin{mathcomp_extra.v} and \coqin{analysis_extra.v}
contain utility lemmas for, respectively, the \mathcomp{} and
\analysis{} libraries.  The file \coqin{mathcomp_extra.v} has a
selection of lemmas on iterated operators (e.g., iterated sums,
$n$-ary maximum), including lemmas for said operations when restricted
to the domain $\itvcc{0}{1}$ used by fuzzy logics. The file
\coqin{analysis_extra.v}, on the other hand, contains multiple lemmas
on properties of $\min$ and $\max$ for
extended reals.  During our work, we completed missing parts of the
shadow-lifting proof for \STL{}; for example, the original \STL{}
proof failed to acknowledge the need for L'Hôpital's rule.
%

Regarding the concrete formalisation strategy, it was revealing that
most of our formalisation was coherent with the standard \mathcomp{}
libraries (and standard mathematical results), and the library
extensions we needed were natural (e.g., L'Hôpital's rule).  This
work hence demonstrates that \rocq{} and \mathcomp{} are effective
working tools to formalise state-of-the-art AI results: \DLtwo{} and \STL{}
were published in recent conferences\cite{varnai,fischer2019dl2} and
this paper formalises the most significant results from both.

In the end, we have a uniform formalisation where all \DL{}s are ``tamed'' 
which 
provides 
solid ground for formalisation of methods deployed in verification of neural 
networks.

\paragraph{Future work}
We plan to consider other definitions of soundness, and other \DL{}s,
including \STL{} with revised negation.  We conjecture that
Definition~\ref{def:shadow} allows for generalisation (removing the
condition ``$p_j = p$'') but this is left for future work.  
%
%
The trade-off between idempotence, associativity, and shadow-lifting
that was conjectured by Varnai and Dimarogonas~\cite{varnai} is reminiscent of 
substructural
logics and suggests investigating the connection.
Establishing connection of this work with the logics of Lawvere
quantale by Bacci et al.~\cite{bacci2023mfps}  or by 
Capucci~\cite{abs-2406-04936}
might also provide new tools to study \DL{}s.

Additionally, we have several possible improvements to the formalisation itself 
in mind. 
Firstly, we could generalise our formalisation by introducing
a structure for residuated lattices instead of specialising to real numbers.
Furthermore, we believe that the use of multisets~\cite{mset} instead
of polymorphic lists (the \coqin{seq} type in \mathcomp{}) in the
implementation of the hypersequent calculi would ultimately provide
better support when coupled with appropriate automation.
Namely, due to lack of exchange rules in calculi modelled on
multisets, automation of the proofs could be more modular; this would
however require substantial effort to match the level of support
enjoyed with polymorphic lists.

We leave the formalisation of completeness of the calculi, as well as 
result related to cut-free proofs~\cite{fuzzy-proof}, for future work.
We also suggest as a simpler but weaker alternative a notion of weak 
completeness---ensuring that axioms from Table~\ref{table:axiom-schema-lattice} are 
provable within the target calculus. 
This, while it does not guarantee completeness, can be seen as a 
pre-requisite to it, as it would standardly be proven first before algebraically 
proving how this set of axioms is complete relative to real numbers.
This provides some intuition as tho why \Yager{} logic may not have a 
calculus---it is easy to see that in it, prelinearity is not satisfied. If we 
take $\tyager{p_0} = \tyager{p_1} = 0.5$ and $r = 2$ and the 
semantic interpretation of the axiom, we have:
\begin{align*}
&\tyager{(p_0 \impl p_1) \vee (p_1 \impl p_0)} = \\ 
&\max(\min(((1-\tyager{p_0})^r+\tyager{p_1}^r)^{\frac{1}{r}},1),
\min(((1-\tyager{p_0})^r+\tyager{p_1}^r)^{\frac{1}{r}},1))
 =  \\
&\max(\min(((1-0.5)^2+0.5^2)^{\frac{1}{2}},1),\min(((1-0.5)^2+0.5^p)^{\frac{1}{2}},1))
 \approx 0.86
\end{align*}
while to be considered true it must be equal to $\tyager{\top} = 1$.

Separately from the questions of scientific curiosity and mathematical elegance,
there is a question of lacking programming language support for machine 
learning. 
As tools like Vehicle~\cite{DaggittKAKS025} and 
CAISAR~\cite{girardsatabin2022caisar}
 are being proposed to provide a more principled 
approach to verification of machine learning, in the long term, compilers of 
these new emerging languages will require certification. And this, in turn, 
will demand formalisation
of results such as the ones we presented here. The formalisation of \DL{}s 
would hence directly contribute to 
certified compilation of specification languages to machine learning 
libraries.   

Finally, first-order quantification for quantitative logics generally and 
\DL{}s in particular, has been a long-standing open problem, discussed in our 
work~\cite{ldl} and the work of others~\cite{abs-2501-18275,abs-2406-04936}. 
Thanks to accumulated knowledge of these logics, we stand very close to its 
resolution. 

\paragraph{Acknowledgments}
The authors would like to thank Zachary Stone for his review
of the formalisation of L'H\^opital's rule, Jairo Marulanda-Giraldo for his 
insights into soundness for \Godel{} logic, Enrico Marchioni for his comments 
on fuzzy logics, and
the anonymous reviewers of the 15th International
Conference on Interactive Theorem Proving. 

\section{Funding}
Natalia Ślusarz acknowledges the support of EPSRC DTA grant from Heriot-Watt 
University. Natalia Ślusarz and Ekaterina Komendantskaya acknowledge the 
support of the EPSRC Grant EP/T026952/1 \emph{AISEC: AI Secure and Explainable 
by Construction}.
Ekaterina Komendantskaya's work was partially supported by the ARIA grant 
\emph{``Mathematics of Safeguarded AI''}.
Reynald Affeldt acknowledges support by JSPS KAKENHI Grant Number 22H00520.
Alessandro Bruni acknowledges support by the Carlsberg Foundation,
grant CF24-2347 \emph{``Verified Functional Analysis for Safe AI (VeriFunAI)''}.

\bibliography{stl.bib}

\begin{thebibliography}{10}

\bibitem{analysis}
Reynald Affeldt, Yves Bertot, Alessandro Bruni, Cyril Cohen, Marie Kerjean,
  Assia Mahboubi, Damien Rouhling, Pierre Roux, Kazuhiko Sakaguchi, Zachary
  Stone, Pierre-Yves Strub, and Laurent Théry.
\newblock {MathComp-Analysis}: {Mathematical} {Components} compliant analysis
  library.
\newblock \url{https://github.com/math-comp/analysis}, 2025.
\newblock Since 2017. Version 1.13.0.

\bibitem{ldl-coq}
Reynald Affeldt, Alessandro Bruni, Ekaterina Komendantskaya, Natalia Slusarz,
  and Kathrin Stark.
\newblock Taming differentiable logics with {Coq} formalisation.
\newblock In {\em 15th International Conference on Interactive Theorem Proving
  ({ITP} 2024), September 9--14, 2024, Tbilisi, Georgia}, volume 309 of {\em
  LIPIcs}, pages 4:1--4:19. Schloss Dagstuhl - Leibniz-Zentrum f{\"{u}}r
  Informatik, 2024.
\newblock \href {https://doi.org/10.4230/LIPICS.ITP.2024.4}
  {\path{doi:10.4230/LIPICS.ITP.2024.4}}.

\bibitem{lhopital}
Reynald Affeldt, Alessandro Bruni, and Natalia Ślusarz.
\newblock L'{H}ôpital's rule.
\newblock MathComp-Analysis Pull Request
  \url{https://github.com/math-comp/analysis/pull/1479}.

\bibitem{github}
Reynald Affeldt, Alessandro Bruni, and Natalia Ślusarz.
\newblock {Formalisation of Differentiable Logics in {Coq}}.
\newblock \url{https://github.com/ndslusarz/formal_LDL}, 2025.

\bibitem{affeldt2018jfr}
Reynald Affeldt, Cyril Cohen, and Damien Rouhling.
\newblock Formalization techniques for asymptotic reasoning in classical
  analysis.
\newblock {\em J. Formaliz. Reason.}, 11(1):43--76, 2018.
\newblock \href {https://doi.org/10.6092/ISSN.1972-5787/8124}
  {\path{doi:10.6092/ISSN.1972-5787/8124}}.

\bibitem{agliano2025algebraic}
Paolo Aglian{\`o}.
\newblock An algebraic investigation of linear logic.
\newblock {\em Archive for Mathematical Logic}, pages 1--23, 2025.

\bibitem{albarghouthi2021introduction}
Aws Albarghouthi et~al.
\newblock Introduction to neural network verification.
\newblock {\em Foundations and Trends{\textregistered} in Programming
  Languages}, 7(1--2):1--157, 2021.

\bibitem{aleksandrov2023formalizing}
Andrei Aleksandrov and Kim V{\"{o}}llinger.
\newblock Formalizing piecewise affine activation functions of neural networks
  in {Coq}.
\newblock In {\em 15th International {NASA} Symposium on Formal Methods ({NFM}
  2023), Houston, TX, USA, May 16--18, 2023}, volume 13903 of {\em Lecture
  Notes in Computer Science}, pages 62--78. Springer, 2023.
\newblock \href {https://doi.org/10.1007/978-3-031-33170-1\_4}
  {\path{doi:10.1007/978-3-031-33170-1\_4}}.

\bibitem{ADK24}
Robert Atkey, Matthew~L. Daggitt, and Wen Kokke.
\newblock Vehicle formalisation, 2024.
\newblock URL: \url{https://github.com/vehicle-lang/vehicle-formalisation}.

\bibitem{bacci2023mfps}
Giorgio Bacci, Radu Mardare, Prakash Panangaden, and Gordon~D. Plotkin.
\newblock Propositional logics for the {Lawvere} quantale.
\newblock In {\em 39th Conference on the Mathematical Foundations of
  Programming Semantics ({MFPS} XXXIX), Indiana University, Bloomington, IN,
  USA, June 21--23, 2023}, volume~3 of {\em {EPTICS}}. EpiSciences, 2023.
\newblock \href {https://doi.org/10.46298/ENTICS.12292}
  {\path{doi:10.46298/ENTICS.12292}}.

\bibitem{abs-2402-03543}
Giorgio Bacci, Radu Mardare, Prakash Panangaden, and Gordon~D. Plotkin.
\newblock Polynomial {L}awvere {L}ogic.
\newblock {\em CoRR}, abs/2402.03543, 2024.
\newblock \href {https://arxiv.org/abs/2402.03543} {\path{arXiv:2402.03543}},
  \href {https://doi.org/10.48550/ARXIV.2402.03543}
  {\path{doi:10.48550/ARXIV.2402.03543}}.

\bibitem{abs-2501-18275}
Giorgio Bacci and Rasmus~Ejlers M{\o}gelberg.
\newblock Induction and recursion principles in a higher-order quantitative
  logic.
\newblock {\em CoRR}, abs/2501.18275, 2025.
\newblock \href {https://arxiv.org/abs/2501.18275} {\path{arXiv:2501.18275}},
  \href {https://doi.org/10.48550/ARXIV.2501.18275}
  {\path{doi:10.48550/ARXIV.2501.18275}}.

\bibitem{baczynski2007yager}
Micha{\l} Baczy{\'n}ski and Balasubramaniam Jayaram.
\newblock Yager’s classes of fuzzy implications: some properties and
  intersections.
\newblock {\em Kybernetika}, 43(2):157--182, 2007.

\bibitem{bagnall2019certifying}
Alexander Bagnall and Gordon Stewart.
\newblock Certifying the true error: Machine learning in {Coq} with verified
  generalization guarantees.
\newblock In {\em The 33rd {AAAI} Conference on Artificial Intelligence ({AAAI}
  2019), The 31st Innovative Applications of Artificial Intelligence Conference
  ({IAAI} 2019), The 9th {AAAI} Symposium on Educational Advances in Artificial
  Intelligence ({EAAI} 2019), Honolulu, Hawaii, USA, January 27--February 1,
  2019}, pages 2662--2669. {AAAI} Press, 2019.
\newblock \href {https://doi.org/10.1609/AAAI.V33I01.33012662}
  {\path{doi:10.1609/AAAI.V33I01.33012662}}.

\bibitem{barrett2018satisfiability}
Clark Barrett and Cesare Tinelli.
\newblock Satisfiability modulo theories.
\newblock In {\em Handbook of model checking}, pages 305--343. Springer, 2018.

\bibitem{brucker2023verifying}
Achim~D. Brucker and Amy Stell.
\newblock Verifying feedforward neural networks for classification in
  {Isabelle/HOL}.
\newblock In {\em 25th International Symposium on Formal Methods ({FM} 2023),
  L{\"{u}}beck, Germany, March 6--10, 2023}, volume 14000 of {\em Lecture Notes
  in Computer Science}, pages 427--444. Springer, 2023.
\newblock \href {https://doi.org/10.1007/978-3-031-27481-7\_24}
  {\path{doi:10.1007/978-3-031-27481-7\_24}}.

\bibitem{abs-2406-04936}
Matteo Capucci.
\newblock On quantifiers for quantitative reasoning.
\newblock {\em CoRR}, abs/2406.04936, 2024.
\newblock \href {https://arxiv.org/abs/2406.04936} {\path{arXiv:2406.04936}},
  \href {https://doi.org/10.48550/ARXIV.2406.04936}
  {\path{doi:10.48550/ARXIV.2406.04936}}.

\bibitem{CasadioKDKKAR22}
Marco Casadio, Ekaterina Komendantskaya, Matthew~L. Daggitt, Wen Kokke, Guy
  Katz, Guy Amir, and Idan Refaeli.
\newblock Neural network robustness as a verification property: {A} principled
  case study.
\newblock In {\em 34th International Conference on Computer Aided Verification
  ({CAV} 2022), Haifa, Israel, August 7--10, 2022, Part {I}}, volume 13371 of
  {\em Lecture Notes in Computer Science}, pages 219--231. Springer, 2022.
\newblock \href {https://doi.org/10.1007/978-3-031-13185-1\_11}
  {\path{doi:10.1007/978-3-031-13185-1\_11}}.

\bibitem{ChevallierWF22}
Mark Chevallier, Matthew Whyte, and Jacques~D. Fleuriot.
\newblock Constrained training of neural networks via theorem proving (short
  paper).
\newblock In {\em 4th Workshop on Artificial Intelligence and Formal
  Verification, Logic, Automata, and Synthesis hosted by the 21st International
  Conference of the Italian Association for Artificial Intelligence (AIxIA
  2022), Udine, Italy, November 28, 2022}, volume 3311 of {\em {CEUR} Workshop
  Proceedings}, pages 7--12. CEUR-WS.org, 2022.
\newblock URL: \url{https://ceur-ws.org/Vol-3311/paper2.pdf}.

\bibitem{mset}
Cyril Cohen and Kazuhiko Sakaguchi.
\newblock Mathematical components finite maps library.
\newblock Last stable version: 2.2.1 (2025).
\newblock URL: \url{https://github.com/math-comp/finmap}.

\bibitem{CordeiroDGIJKKLMSW25}
Lucas~C. Cordeiro, Matthew~L. Daggitt, Julien Girard{-}Satabin, Omri Isac,
  Taylor~T. Johnson, Guy Katz, Ekaterina Komendantskaya, Augustin Lemesle,
  Edoardo Manino, Artjoms Sinkarovs, and Haoze Wu.
\newblock Neural network verification is a programming language challenge.
\newblock In {\em 34th European Symposium on Programming ({ESOP} 2025),
  Hamilton, ON, Canada, May 5--8, 2025}, volume 15694 of {\em Lecture Notes in
  Computer Science}, pages 206--235. Springer, 2025.
\newblock \href {https://doi.org/10.1007/978-3-031-91118-7\_9}
  {\path{doi:10.1007/978-3-031-91118-7\_9}}.

\bibitem{FoMLAS2023}
Matthew Daggitt, Wen Kokke, Ekaterina Komendantskaya, Robert Atkey, Luca
  Arnaboldi, Natalia Ślusarz, Marco Casadio, Ben Coke, and Jeonghyeon Lee.
\newblock The {Vehicle} tutorial: Neural network verification with {Vehicle}.
\newblock In {\em 6th Workshop on Formal Methods for ML-Enabled Autonomous
  Systems}, volume~16 of {\em Kalpa Publications in Computing}, pages 1--5.
  EasyChair, 2023.
\newblock \href {https://doi.org/10.29007/5s2x} {\path{doi:10.29007/5s2x}}.

\bibitem{DaggittAKKA23}
Matthew~L. Daggitt, Robert Atkey, Wen Kokke, Ekaterina Komendantskaya, and Luca
  Arnaboldi.
\newblock Compiling higher-order specifications to {SMT} solvers: How to deal
  with rejection constructively.
\newblock In {\em 12th {ACM} {SIGPLAN} International Conference on Certified
  Programs and Proofs ({CPP} 2023), Boston, MA, USA, January 16--17, 2023},
  pages 102--120. {ACM}, 2023.
\newblock \href {https://doi.org/10.1145/3573105.3575674}
  {\path{doi:10.1145/3573105.3575674}}.

\bibitem{DaggittKAKS025}
Matthew~L. Daggitt, Wen Kokke, Robert Atkey, Ekaterina Komendantskaya, Natalia
  Ślusarz, and Luca Arnaboldi.
\newblock Vehicle: Bridging the embedding gap in the verification of
  neuro-symbolic programs (invited talk).
\newblock In {\em 10th International Conference on Formal Structures for
  Computation and Deduction ({FSCD} 2025), July 14--20, 2025, Birmingham,
  {UK}}, volume 337 of {\em LIPIcs}, pages 2:1--2:20. Schloss Dagstuhl -
  Leibniz-Zentrum f{\"{u}}r Informatik, 2025.
\newblock \href {https://doi.org/10.4230/LIPICS.FSCD.2025.2}
  {\path{doi:10.4230/LIPICS.FSCD.2025.2}}.

\bibitem{Des25}
Remi Desmartin, Omri Isac, Ekaterina Komendantskaya, Kathrin Stark, Grant~O.
  Passmore, and Guy Katz.
\newblock A certified proof checker for deep neural network verification.
\newblock In {\em 16th International Conference on Interactive Theorem Proving
  ({ITP} 2025), Reykjav\'ik, Iceland, 28 September--1 October 2025}, LIPIcs.
  Schloss Dagstuhl - Leibniz-Zentrum f{\"{u}}r Informatik, 2025.

\bibitem{mathcomp}
The~MathComp development team.
\newblock Mathematical components.
\newblock \url{https://github.com/math-comp/math-comp}, 2005.
\newblock Last stable version: 2.4.0 (2025).

\bibitem{fischer2019dl2}
Marc Fischer, Mislav Balunovic, Dana Drachsler{-}Cohen, Timon Gehr, Ce~Zhang,
  and Martin~T. Vechev.
\newblock {DL2:} training and querying neural networks with logic.
\newblock In {\em 36th International Conference on Machine Learning ({ICML}
  2019), 9--15 June 2019, Long Beach, California, {USA}}, volume~97 of {\em
  Proceedings of Machine Learning Research}, pages 1931--1941. {PMLR}, 2019.
\newblock URL: \url{http://proceedings.mlr.press/v97/fischer19a.html}.

\bibitem{flinkow2025comparing}
Thomas Flinkow, Barak~A Pearlmutter, and Rosemary Monahan.
\newblock Comparing differentiable logics for learning with logical
  constraints.
\newblock {\em Science of Computer Programming}, 244:103280, 2025.

\bibitem{galatos2007residuated}
Nikolaos Galatos, Peter Jipsen, Tomasz Kowalski, and Hiroakira Ono.
\newblock {\em Residuated Lattices: An Algebraic Glimpse at Substructural
  Logics, Volume 151}.
\newblock Studies in Logic and the Foundations of Mathematics. Elsevier
  Science, 2007.

\bibitem{girardsatabin2022caisar}
Julien Girard-Satabin, Michele Alberti, Fran{\c c}ois Bobot, Zakaria Chihani,
  and Augustin Lemesle.
\newblock {CAISAR}: A platform for characterizing artificial intelligence
  safety and robustness.
\newblock In {\em The IJCAI-ECAI-22 Workshop on Artificial Intelligence Safety
  (AISafety 2022), July 24--25, 2022, Vienna, Austria}, CEUR-Workshop
  Proceedings, July 2022.
\newblock URL: \url{https://hal.archives-ouvertes.fr/hal-03687211}.

\bibitem{ijcai2022p767}
Eleonora Giunchiglia, Mihaela~Catalina Stoian, and Thomas Lukasiewicz.
\newblock Deep learning with logical constraints.
\newblock In {\em 31st International Joint Conference on Artificial
  Intelligence ({IJCAI-22})}, pages 5478--5485. International Joint Conferences
  on Artificial Intelligence Organization, 7 2022.
\newblock Survey Track.
\newblock \href {https://doi.org/10.24963/ijcai.2022/767}
  {\path{doi:10.24963/ijcai.2022/767}}.

\bibitem{klement2013triangular}
Erich~Peter Klement, Radko Mesiar, and Endre Pap.
\newblock {\em Triangular norms}, volume~8.
\newblock Springer Science \& Business Media, 2013.

\bibitem{KM18}
Zico Kolter and Aleksander Madry.
\newblock Adversarial robustness---theory and practice.
\newblock NeurIPS 2018 tutorial, 2018.
\newblock Available at \url{https://adversarial-ml-tutorial.org/}.

\bibitem{lambek1961calculus}
Joachim Lambek.
\newblock On the calculus of syntactic types.
\newblock {\em Structure of language and its mathematical aspects},
  12:166--178, 1961.

\bibitem{Law73}
F.~William Lawvere.
\newblock Metric spaces, generalized logic, and closed categories.
\newblock {\em Rendiconti del Seminario Matematico e Fisico di Milano},
  43(1):135 -- 166, 1973.

\bibitem{de2022use}
Elisabetta~De Maria, Abdorrahim Bahrami, Thibaud L'Yvonnet, Amy~P. Felty,
  Daniel Gaff{\'{e}}, Annie Ressouche, and Franck Grammont.
\newblock On the use of formal methods to model and verify neuronal archetypes.
\newblock {\em Frontiers Comput. Sci.}, 16(3):163404, 2022.
\newblock \href {https://doi.org/10.1007/S11704-020-0029-6}
  {\path{doi:10.1007/S11704-020-0029-6}}.

\bibitem{marulanda-giraldo2025our}
Jairo~Miguel Marulanda-Giraldo, Ekaterina Komendantskaya, Alessandro Bruni,
  Reynald Affeldt, Matteo Capucci, and Enrico Marchioni.
\newblock Quantifiers for differentiable logics in {Rocq} (extended abstract).
\newblock In {\em International Symposium on AI Verification (SAIV 2025),
  Zagreb, Croatia, July 21--22, 2025}, 2025.
\newblock Paper available at \url{https://openreview.net/forum?id=VbQNQ9hjJE}.

\bibitem{metcalfe2004analytic}
George Metcalfe, Nicola Olivetti, and Dov Gabbay.
\newblock Analytic calculi for product logics.
\newblock {\em Archive for Mathematical Logic}, 43:859--889, 2004.

\bibitem{fuzzy-proof}
George Metcalfe, Nicola Olivetti, and Dov~M. Gabbay.
\newblock {\em Proof Theory for Fuzzy Logics}, volume~36 of {\em Applied
  Logic}.
\newblock Springer, 2008.

\bibitem{moot2012non}
Richard Moot and Christian Retor{\'e}.
\newblock The non-associative {L}ambek calculus.
\newblock In {\em The Logic of Categorial Grammars: A Deductive Account of
  Natural Language Syntax and Semantics}, pages 101--147. Springer.

\bibitem{murphy2017verified}
Charlie Murphy, Patrick Gray, and Gordon Stewart.
\newblock Verified perceptron convergence theorem.
\newblock In {\em 1st ACM SIGPLAN International Workshop on Machine Learning
  and Programming Languages}, pages 43--50, 2017.

\bibitem{rudin1976}
Walter Rudin.
\newblock {\em Principles of Mathematical Analysis}.
\newblock McGraw-Hill, 3rd edition, 1976.

\bibitem{sakaguchi2022itp}
Kazuhiko Sakaguchi.
\newblock Reflexive tactics for algebra, revisited.
\newblock In {\em 13th International Conference on Interactive Theorem Proving
  ({ITP} 2022), August 7--10, 2022, Haifa, Israel}, volume 237 of {\em LIPIcs},
  pages 29:1--29:22. Schloss Dagstuhl - Leibniz-Zentrum f{\"{u}}r Informatik,
  2022.
\newblock \href {https://doi.org/10.4230/LIPICS.ITP.2022.29}
  {\path{doi:10.4230/LIPICS.ITP.2022.29}}.

\bibitem{algebratactics}
Kazuhiko Sakaguchi and Pierre Roux.
\newblock Algebra tactics: Ring, field, lra, nra, and psatz tactics for
  {Mathematical} {Components}.
\newblock \url{https://github.com/math-comp/algebra-tactics}, 2021.
\newblock Last stable release: 1.2.3 (2024).

\bibitem{van2022analyzing}
Emile van Krieken, Erman Acar, and Frank van Harmelen.
\newblock Analyzing differentiable fuzzy logic operators.
\newblock {\em Artif. Intell.}, 302:103602, 2022.
\newblock \href {https://doi.org/10.1016/J.ARTINT.2021.103602}
  {\path{doi:10.1016/J.ARTINT.2021.103602}}.

\bibitem{van2025neurosymbolic}
Emile van Krieken, Pasquale Minervini, Edoardo Ponti, and Antonio Vergari.
\newblock Neurosymbolic diffusion models.
\newblock {\em arXiv e-prints}, pages arXiv--2505, 2025.

\bibitem{varnai}
Peter Varnai and Dimos~V. Dimarogonas.
\newblock On robustness metrics for learning {STL} tasks.
\newblock In {\em 2020 American Control Conference ({ACC} 2020), Denver, CO,
  USA, July 1--3, 2020}, pages 5394--5399. {IEEE}, 2020.
\newblock \href {https://doi.org/10.23919/ACC45564.2020.9147692}
  {\path{doi:10.23919/ACC45564.2020.9147692}}.

\bibitem{lukasiewicz1920three}
J~Łukasiewicz.
\newblock {\em O logice trójwartościowej (in Polish). English translation: On
  Three-Valued Logic, in Borkowski, L.(ed.) 1970. Jan Łukasiewicz: Selected
  Works, Amsterdam: North Holland}.
\newblock Ruch Filozoficzny, 1920.

\bibitem{ldl}
Natalia Ślusarz, Ekaterina Komendantskaya, Matthew~L. Daggitt, Robert~J.
  Stewart, and Kathrin Stark.
\newblock Logic of differentiable logics: Towards a uniform semantics of {DL}.
\newblock In {\em 24th International Conference on Logic for Programming,
  Artificial Intelligence and Reasoning ({LPAR} 2023), Manizales, Colombia,
  June 4--9, 2023}, volume~94 of {\em EPiC Series in Computing}, pages
  473--493. EasyChair, 2023.
\newblock \href {https://doi.org/10.29007/C1NT} {\path{doi:10.29007/C1NT}}.

\end{thebibliography}

\end{document}